\pdfoutput=1
\documentclass[11pt]{article}

\newif\ifdraft
\draftfalse

\usepackage[letterpaper,margin=1in,marginparwidth=0.8in,marginparsep=0.1in]{geometry}
\usepackage{microtype}
\usepackage{amsmath,amssymb,amsthm,mathtools}
\usepackage[dvipsnames]{xcolor}
\usepackage{enumitem}
\usepackage{booktabs,array,longtable}
\usepackage{float}
\floatstyle{boxed}
\newfloat{Problem}{H}{prb}
\newfloat{Protocol}{H}{prt}
\usepackage{tikz}
\usetikzlibrary{arrows,positioning,decorations.pathreplacing}
\usepackage[colorlinks=true,linkcolor=blue!60!black,citecolor=green!40!black,urlcolor=blue!60!black]{hyperref}
\usepackage[capitalise,nameinlink]{cleveref}
\crefname{Protocol}{Protocol}{Protocols}
\crefname{Problem}{Problem}{Problems}

\theoremstyle{plain}
\newtheorem{theorem}{Theorem}[section]
\newtheorem{lemma}[theorem]{Lemma}
\newtheorem{claim}[theorem]{Claim}
\newtheorem{corollary}[theorem]{Corollary}
\newtheorem{proposition}[theorem]{Proposition}
\newtheorem{fact}[theorem]{Fact}
\theoremstyle{definition}
\newtheorem{definition}[theorem]{Definition}
\theoremstyle{remark}

\crefname{fact}{Fact}{Facts}
\makeatletter
\newtheorem*{rep@theorem}{\rep@title}
\newcommand{\newreptheorem}[2]{%
\newenvironment{rep#1}[1]{%
\def\rep@title{#2 \ref*{##1}}%
\begin{rep@theorem}}%
{\end{rep@theorem}}}
\makeatother
\newreptheorem{corollary}{Corollary}
\crefname{claim}{Claim}{Claims}

\newcommand{\eps}{\varepsilon}
\newcommand{\E}{\mathop{\mathbb{E}}}
\newcommand{\abs}[1]{\left|#1\right|}
\newcommand{\cS}{\mathcal{S}}
\providecommand{\cA}{\mathcal{A}}
\providecommand{\cB}{\mathcal{B}}
\newcommand{\cI}{\mathcal{I}}
\newcommand{\cE}{\mathcal{E}}
\newcommand{\cD}{\mathcal{D}}
\providecommand{\cU}{\mathcal{U}}
\newcommand{\DISJ}{\mathrm{DISJ}}

\providecommand{\norm}[1]{\left\lVert #1\right\rVert}
\newcommand{\Hmin}{H_\infty}
\newcommand{\tHmin}{\widetilde{H}_\infty}
\newcommand{\vS}{\vec{S}}
\newcommand{\Csam}{\widetilde{C}_i}
\newcommand{\Aall}{A_{\mathrm{all}}}

\newenvironment{proofof}[1]{\smallskip\noindent{\bf Proof of #1.}}%
{\hspace*{\fill}$\Box$\par}

\usepackage{xspace}
\newcommand{\Patrascu}{P\u{a}tra\c{s}cu\xspace}

\hypersetup{
  pdftitle={Unconditional Boolean Dynamic Lower Bounds from Disjointness},
  pdfauthor={Young Kun Ko},
  pdfkeywords={cell-probe model, dynamic data structures, lower bounds, communication complexity, directed reachability, set disjointness}
}

\title{Unconditional Boolean Dynamic Lower Bounds \\from Disjointness}

\author{Young Kun Ko\thanks{Department of Computer Science and Engineering, Pennsylvania State University, University Park, PA, USA. \texttt{ykko@psu.edu}}}

\date{}

\begin{document}
\maketitle

\begin{abstract}
We prove an unconditional $\Omega\big((\log n/\log\log n)^2\big)$ lower bound on the maximum of the worst-case update time and expected query time for the insert-only version of Pătraşcu's multiphase Disjointness problem [STOC 2010]. The bound holds in the zero-error cell-probe model with logarithmic word size. Pătraşcu's reductions transfer it to directed reachability under edge insertions, already on directed acyclic graphs of depth $2$. This improves the $\widetilde\Omega(\log^{3/2} n)$ bound of Larsen and Yu [SICOMP 2025] to the quadratic-logarithmic scale previously attained for dynamic range counting and multiphase Inner Product. The bound also transfers to a broad collection of dynamic graph, algebraic, geometric, and string problems.

We also obtain a separation between amortized and worst-case time for streaming regular-expression membership. For a simple family of expressions of maximum length $M$ over a polynomial-size alphabet, a deterministic algorithm uses $O(1)$ amortized time per symbol on every input stream, whereas every zero-error cell-probe algorithm with word size $w=\Theta(\log M)$ requires $\Omega\big((\log M/\log\log M)^2\big)$ time on some symbol.

The proof adapts Ko's 2.5-round multiphase communication framework [FOCS 2026] from Inner Product to Disjointness. Our main technical observation is that the protocol's verification step allows a one-sided corruption bound to replace the discrepancy property used for Inner Product.
\end{abstract}

\newpage
\section{Introduction}\label{sec:intro}

Recent algorithmic breakthroughs have exposed the vulnerability of central hardness conjectures. At the same time, Ko~\cite{ko_omega_2026} offered \emph{a new hope} for unconditional Boolean dynamic lower bounds: a program to unify bounds at the quadratic-logarithmic scale through 2.5-round communication games. In this work, \emph{the unconditional lower bounds strike back}: we extend the framework from Inner Product to Disjointness and obtain unconditional bounds for a broad collection of dynamic graph, set, matrix, geometric, and string problems.

The common source is \Patrascu's multiphase Disjointness problem~\cite{patrascu_mihai_towards_2010}: preprocess a family of sets, insert elements into a set $X$, and finally test whether $X$ is disjoint from a specified member of the family. For $n$ preprocessed subsets of $[n]$, we prove
\[
\max(t_u,t_q)=\Omega\!\left(\left(\frac{\log n}{\log\log n}\right)^2\right)
\]
in the zero-error cell-probe model with word size $w=\Theta(\log n)$, where $t_u$ is the worst-case insertion time and $t_q$ the expected query time on each input. Preprocessing is unrestricted. \Cref{thm:main-informal} gives the stronger update/query trade-off from which all our applications follow.

Disjointness connects this theorem to problems with very different query interfaces. Known reductions, and adaptations of published constructions, transfer the bound to incremental directed reachability and bipartite matching, planar distances, range mode, inverse-entry and determinant maintenance, geometric counting, and string matching. They use $O(1)$ target operations per source operation and change the instance size only polynomially, so they preserve the quadratic-logarithmic scale. \Cref{tab:intro-applications} lists these problems with references, and \cref{sec:graphs} collects the constructions, attributes their sources, and checks that they meet the requirements of our model.

Directed reachability is a leading example. Our bound holds already under edge insertions on depth-$2$ directed acyclic graphs and improves the $\widetilde\Omega(\log^{3/2}n)$ bound of Larsen and Yu~\cite{larsen_super-logarithmic_2025} to the quadratic-logarithmic scale previously attained for dynamic range counting and multiphase Inner Product~\cite{larsen_cell_2012,ko_omega_2026}.

We also obtain a separation between amortized and worst-case time for streaming regular-expression membership, which asks after each arriving symbol whether the prefix read so far matches a preprocessed expression. For a simple family of expressions of maximum length $M$ over a polynomial-size alphabet, a deterministic algorithm uses $O(1)$ amortized time per symbol on every input stream, whereas every zero-error cell-probe algorithm with word size $w=\Theta(\log M)$ requires $\Omega\big((\log M/\log\log M)^2\big)$ time on some symbol (\cref{thm:streaming-regex-informal}). Thus the Disjointness bound also gives an unconditional obstruction to deamortizing this streaming computation.

\paragraph{The vulnerability of hardness conjectures.}
Many of these dynamic problems have been studied through polynomial conditional lower bounds. Such a bound shows that a faster data structure would give a faster algorithm for a problem conjectured to be hard, such as 3SUM, all-pairs shortest paths (APSP), or CNF satisfiability~\cite{patrascu_mihai_towards_2010,abboud_popular_2014}. These reductions have been remarkably successful at charting relationships among problems, but they do not by themselves establish that any of the problems is hard. An unconditional lower bound proves hardness within the stated computational model, independently of such conjectures.

Recent events have made this distinction concrete. An algorithm discovered by Claude (Anthropic), subsequently analyzed, strengthened, and presented by Alman and Vassilevska Williams~\cite{alman_truly_2026}, yields truly subquadratic 3SUM and truly subcubic APSP algorithms, refuting the corresponding conjectures. Two recent manuscripts released by OpenAI\footnote{\emph{The Exponential Time Hypothesis is False}, OpenAI, October 25, 2026; available from the \href{https://github.com/openai/math/blob/main/preprints/The-Exponential-Time-Hypothesis-is-False-October-25-2026/README.md}{OpenAI Math repository}.} and Anthropic\footnote{\emph{A $2^{o(n)}$ Time Algorithm for 3-SAT}, Claude (Anthropic), 2026; \href{https://github.com/anthropics/formal-math/blob/main/eth/README.md}{paper and Lean formalization}.} claim subexponential algorithms for 3-SAT, which would refute the Exponential Time Hypothesis (ETH) and hence the Strong Exponential Time Hypothesis (SETH) \cite{impagliazzo_which_2001}.

The refuted conjectures no longer supply hardness for the corresponding lower bounds. The web of reductions remains intact, but its starting premise is gone, and reductions alone cannot establish the hardness of their targets. Our goal is to supply the missing foundation with \textbf{an unconditional lower bound} and put the existing reduction network to use. We do so for multiphase Disjointness: our quadratic-logarithmic lower bound is a proved source of hardness, which established reductions propagate to a broad collection of dynamic problems.

\subsection{Previous work}\label{sec:previous}

\subsubsection{Conditional lower bounds across dynamic problems}

Conditional lower bounds motivated this reduction network. \Patrascu~\cite{patrascu_mihai_towards_2010} introduced multiphase Disjointness, proved its conditional hardness under the 3SUM conjecture, and gave reductions to reachability, shortest paths, subgraph connectivity, and set and array problems. Subsequent work developed conditional bounds for graph problems~\cite{abboud_popular_2014,henzinger_unifying_2015}, planar graphs~\cite{abboud_popular_2016}, dynamic matrices~\cite{van_den_brand_dynamic_2019}, geometric counting~\cite{dallant_conditional_2022,jin_tight_2022}, and string problems~\cite{clifford_upper_2018,charalampopoulos_internal_2019}. The 3SUM and APSP refutations~\cite{alman_truly_2026} have removed central foundations of this program: those conjectures no longer support the associated lower bounds. The claimed ETH refutations discussed above would likewise refute SETH. The reductions remain valid, and many can be specialized or adapted to a single evolving set and a final selected query. The refutations thus leave the affected reduction chains intact, but with no valid source of hardness.

\subsubsection{Unconditional cell-probe lower bounds}
The cell-probe model~\cite{yao_should_1981} charges only for memory access: memory consists of $w$-bit cells, each read or write costs one probe, and computation is free. A lower bound therefore also applies to word-RAM algorithms with the same word size. We work with $w=\Theta(\log n)$ and allow arbitrary preprocessing, so the lower bound concerns the information that must be accessed while processing updates and queries. The model's generous assumptions make lower bounds especially difficult to prove, helping explain the slow progress even for basic dynamic problems.

Fredman and Saks~\cite{fredman_cell_1989} introduced the Chronogram method and proved an $\Omega(\log n/\log\log n)$ lower bound, which transfers to dynamic connectivity~\cite{miltersen_complexity_1994,henzinger_lower_1998}. \Patrascu and Demaine~\cite{patrascu_logarithmic_2006} improved the connectivity bound to $\Omega(\log n)$. For almost two decades, the strongest unconditional lower bounds for natural graph problems remained at the (near-)logarithmic scale. For problems with word-size answers, Larsen~\cite{larsen_cell_2012} reached $\Omega((\log n/\log\log n)^2)$ for weighted range counting. Reaching stronger bounds for Boolean queries required a further breakthrough: Larsen, Weinstein and Yu~\cite{larsen_crossing_2020} proved an $\widetilde\Omega(\log^{3/2}n)$ bound. Building on their technique, Larsen and Yu~\cite{larsen_super-logarithmic_2025} obtained the first super-logarithmic bound for a natural graph problem, proving
\[
t_q=\Omega\!\left(\frac{\log^{3/2}n}{\log^2(t_u w)}\right)
\]
for incremental reachability on DAGs, with worst-case update time $t_u$ and expected query time $t_q$~\cite[Theorem~1]{larsen_super-logarithmic_2025}. \Cref{fig:reach-timeline} places our reachability consequence in this progression.

\begin{figure}[!ht]
\centering
\begin{tikzpicture}[
  x=1cm, y=1cm,
  pt/.style={circle, draw, thick, fill=black, minimum size=7pt, inner sep=0pt},
  lab/.style={align=center, font=\small},
  yr/.style={font=\small},
]
  \draw[thick, ->] (-0.6,0) -- (13.6,0);
  \foreach \x/\yr/\bnd/\ref/\h in {
      1.0/1989/{$\Omega\!\left(\frac{\log n}{\log\log n}\right)$}/{Fredman--Saks \cite{fredman_cell_1989};\\ via \cite{miltersen_complexity_1994,henzinger_lower_1998}}/1.35,
      5.2/2004/{$\Omega(\log n)$}/{\Patrascu--Demaine \cite{patrascu_logarithmic_2006}}/1.35,
      8.9/2023/{$\widetilde{\Omega}\big(\log^{3/2} n\big)$}/{Larsen--Yu \cite{larsen_super-logarithmic_2025}}/1.35,
      12.3/2026/{$\Omega\!\left(\big(\frac{\log n}{\log\log n}\big)^{2}\right)$}/{This work}/1.35} {
    \draw[thick] (\x,-0.12) -- (\x,0.12);
    \node[yr] at (\x,-0.45) {\yr};
    \draw[gray] (\x,0.12) -- (\x,\h-0.35);
    \node[pt] at (\x,\h-0.35) {};
    \node[lab, anchor=south] at (\x,\h-0.15) {\bnd\\[2pt] \ref};
  }
\end{tikzpicture}
\caption{Unconditional cell-probe lower bounds for dynamic directed reachability, with $w = \Theta(\log n)$, stated as bounds on $\max(t_u, t_q)$. For prior results, dates refer to conference appearances. The first two bounds transfer from fully dynamic undirected connectivity to fully dynamic directed reachability; the later bounds already hold for incremental reachability. The first is Fredman and Saks's bound for prefix parity \cite{fredman_cell_1989}, which transfers to dynamic connectivity by a simple reduction observed independently by Miltersen, Subramanian, Vitter and Tamassia~\cite{miltersen_complexity_1994} and by Henzinger and Fredman \cite{henzinger_lower_1998}.}
\label{fig:reach-timeline}
\end{figure}

The Boolean-answer restriction is central to this history. The Chronogram method separates information according to the epochs in which updates arrive. Larsen's cell-sampling argument forces many probes to an individual epoch, but its original encoding argument uses the many bits returned by a query. Larsen, Weinstein and Yu instead reduce a Boolean data structure to a one-way communication game in which the receiver predicts the answer with a small advantage over random guessing; a peak-to-average lemma based on Chebyshev polynomials bounds this advantage~\cite{larsen_crossing_2020}. This is the framework underlying Larsen and Yu's reachability bound.

For multiphase Disjointness itself, \Patrascu's conjectured polynomial lower bound remains open. He proposed a three-party communication game whose hardness would imply the conjecture in the cell-probe model~\cite{patrascu_mihai_towards_2010}. Earlier progress includes bounds for non-adaptive data structures~\cite{brody_adapt_2015}, regimes with very fast queries~\cite{clifford_new_2015}, semi-adaptive data structures via two-round communication~\cite{ko_adaptive_2020,dvorak_lower_2020}, and a 1.5-round version of the communication game~\cite{chattopadhyay_upper_2016}. Proving \Patrascu's conjectured communication lower bound even for three-round protocols would imply major circuit lower bounds~\cite{ko_adaptive_2020}.

\paragraph{The unconditional lower bounds strike back.}
The new framework of~\cite{ko_omega_2026} is a common route to quadratic-logarithmic unconditional lower bounds for Boolean dynamic problems. Its 2.5-round game lets Bob verify Alice's simulation and detect its failure. Its first instance, multiphase Inner Product over $\mathbb F_2$, reaches the scale previously attained by Larsen for word-size answers~\cite{larsen_cell_2012}. Beyond that instance, the framework gives a general recipe based on discrepancy bounds, initiating a program to unify dynamic lower bounds at this scale through communication hardness and reductions.

We take the next step in this program by extending the framework to Disjointness. This requires a further insight into the 2.5-round game: the hard distributions for Disjointness do not satisfy the discrepancy property used for Inner Product in~\cite{ko_omega_2026}. We exploit verification by conditioning on certified ``disjoint'' answers, allowing a one-sided corruption bound to replace discrepancy. This change matters beyond the source problem: \Patrascu's reductions \cite{patrascu_mihai_towards_2010} begin from Disjointness, and many later constructions also encode it. The extension therefore connects the 2.5-round framework to a broad collection of natural dynamic problems. In the spirit of \Patrascu's program of reductions and communication lower bounds~\cite{patrascu_unifying_2008,patrascu_mihai_towards_2010}, the same constructions now transfer a proved lower bound, independently of the conjectures that motivated them.

\subsubsection{Streaming and amortization}
Lower bounds for dynamic data structures are closely related to lower bounds on streaming processing time. A streaming membership algorithm updates its state as symbols arrive and reports after each arrival whether the current prefix belongs to a fixed language. This connection lets us apply our Disjointness bound to streaming regular-expression membership, studied by Bille and Thorup~\cite{bille_regular_2010} and Dudek, Gawrychowski, Gourdel and Starikovskaya~\cite{naor_streaming_2022}. Earlier separations between amortized and worst-case time include union--find~\cite{tarjan_efficiency_1975,blum_single-operation_1985,fredman_cell_1989} and variance-optimal reservoir sampling~\cite{cohen_efficient_2011}; we refer the reader to \cref{sec:streaming-regex} for the detailed discussion.

\subsection{Our results}\label{sec:results}

\paragraph{Main result: multiphase Disjointness.}
Our main technical result is an $\Omega\big((\log n/\log\log n)^2\big)$ lower bound for the multiphase problem of \Patrascu~\cite{patrascu_mihai_towards_2010} with Disjointness, in its insert-only form. All applications in this paper follow from this theorem.

\begin{Problem}
\begin{itemize}[leftmargin=*]
  \item \emph{Phase I.} Sets $S_1,\dots,S_m\subseteq[n]$ arrive and are preprocessed arbitrarily.
  \item \emph{Phase II.} Elements are inserted one at a time into a set $X\subseteq[n]$, initially empty; each insertion probes at most $t_u$ cells.
  \item \emph{Phase III.} An index $q\in[m]$ arrives, and the data structure outputs whether $S_q\cap X=\emptyset$, reading at most $t_q$ cells in expectation; it does not write.
\end{itemize}
\caption{Insert-only multiphase Disjointness\label{prob:multiphase}}
\end{Problem}

\begin{samepage}
\begin{theorem}[Main technical theorem, informal]\label{thm:main-informal}
Every zero-error cell-probe data structure for insert-only multiphase Disjointness with $m=n$ query sets, word size $w=\Theta(\log n)$, worst-case update time $t_u$ per insertion and expected query time $t_q$ satisfies
\[
t_q=\Omega\!\left(\left(\frac{\log n}{\log(t_u w)}\right)^2\right).
\]
In particular,
\[
\max(t_u,t_q)=\Omega\!\left(\left(\frac{\log n}{\log\log n}\right)^2\right).
\]
\end{theorem}
\end{samepage}

\noindent \Cref{thm:main} gives the formal statement; \cref{sec:model} specifies the randomness and cost conventions.

\Patrascu~\cite{patrascu_mihai_towards_2010} conjectured a polynomial lower bound for the multiphase problem when $m=n^{1+\Omega(1)}$ and showed that Disjointness serves as a common source of hardness for several natural dynamic problems. Our theorem does not resolve that conjecture, but reaches the quadratic-logarithmic scale already in the regime $m=n$. Ko~\cite{ko_omega_2026} and Kouck\'y, Loff, Molli and Saks~\cite{koucky_natural_2026} discuss barriers to going beyond this scale with current techniques. We next give two consequences.

\begin{samepage}
\paragraph{Corollary I: dynamic problems.}
Applying \Patrascu's reduction gives a quadratic-logarithmic lower bound for incremental directed reachability, improving the bound of Larsen and Yu~\cite{larsen_super-logarithmic_2025}.

\begin{corollary}[Directed reachability, informal]\label{thm:reach-informal}
Every zero-error cell-probe data structure for incremental directed reachability on $N$ vertices, starting with no edges, with word size $w=\Theta(\log N)$, worst-case insertion time $t_u$ and expected query time $t_q$, satisfies
\[
t_q=\Omega\!\left(\left(\frac{\log N}{\log(t_u w)}\right)^2\right);
\]
in particular,
\[
\max(t_u,t_q)=\Omega\!\left(\left(\frac{\log N}{\log\log N}\right)^2\right).
\]
This holds already on depth-$2$ DAGs whose queries ask whether a fixed sink is reachable.
\end{corollary}
\end{samepage}

\noindent \Cref{thm:reach} gives the formal statement. The bound trivially applies to fully dynamic directed reachability, since a fully dynamic data structure must handle insertion-only sequences.

The same $\Omega\big((\log N/\log\log N)^2\big)$ lower bound on $\max(t_u,t_q)$ follows for the dynamic problems in \cref{tab:intro-applications}, which describes the tasks and gives representative references to prior work. The existing web of reductions gives these consequences their breadth, and we expect further applications beyond this list. Each consequence is an unconditional lower bound for its target problem, obtained by transferring proved hardness from Disjointness.
\begingroup
\small
\renewcommand{\arraystretch}{1.05}
\setlength{\LTcapwidth}{\textwidth}
\setlength{\LTpre}{\smallskipamount}
\setlength{\LTpost}{\medskipamount}
\begin{longtable}{@{}>{\raggedright\arraybackslash}p{0.25\textwidth}>{\raggedright\arraybackslash}p{0.58\textwidth}>{\raggedright\arraybackslash}p{\dimexpr0.17\textwidth-4\tabcolsep\relax}@{}}
\caption{Dynamic applications: brief descriptions and representative prior work.}
\label{tab:intro-applications}\\
\toprule
\textbf{Problem} & \textbf{Task} & \textbf{Prior work} \\
\midrule
\endfirsthead
\multicolumn{3}{@{}l}{\tablename~\thetable\ (continued)}\\
\toprule
\textbf{Problem} & \textbf{Task} & \textbf{Prior work} \\
\midrule
\endhead
\multicolumn{3}{r@{}}{\footnotesize Continued on the next page.}\\
\endfoot
\bottomrule
\endlastfoot
\multicolumn{3}{@{}l}{\emph{Graphs (edge insertions unless stated otherwise)}}\\*
\hyperref[sec:graph-problems]{Directed reachability} & Is one vertex reachable from another? Includes fixed $s$--$t$ queries and counting vertices reachable from a fixed source. & \cite{patrascu_mihai_towards_2010,abboud_popular_2014} \\
\addlinespace[2pt]
\hyperref[sec:graph-problems]{Strong connectivity} & Can two vertices reach one another along directed paths? & \cite{haeupler_incremental_2012} \\
\addlinespace[2pt]
\hyperref[sec:graph-problems]{Directed cycle detection} & Does the graph contain a directed cycle? & \cite{haeupler_incremental_2012,chen_almost-linear_2024} \\
\addlinespace[2pt]
\hyperref[sec:graph-problems]{Undirected shortest paths} & Return the distance between two vertices, possibly fixed vertices $s$ and $t$. & \cite{patrascu_mihai_towards_2010,abboud_popular_2014} \\
\addlinespace[2pt]
\hyperref[sec:graph-problems]{Subgraph connectivity} & Activate vertices of a fixed undirected graph; are two vertices connected in the induced subgraph? & \cite{patrascu_mihai_towards_2010} \\
\addlinespace[2pt]
\hyperref[sec:graph-problems]{Triangle detection} & Does an undirected graph contain a triangle? & \cite{henzinger_unifying_2015} \\
\addlinespace[2pt]
\hyperref[sec:bipartite-matching]{Bipartite matching} & Return the maximum number of vertex-disjoint edges, or test whether such edges can cover every vertex. & \cite{abboud_popular_2014} \\
\addlinespace[2pt]
\hyperref[sec:app-density]{Exact densest subgraph} & Return the maximum ratio of induced edges to vertices over nonempty vertex subsets. & \cite{henzinger_unifying_2015} \\
\addlinespace[2pt]
\hyperref[sec:app-marked-distance]{Nearest marked vertex} & In a fixed graph with changing vertex marks, return the distance from a queried vertex to the nearest marked vertex. & \cite{henzinger_unifying_2015} \\
\midrule
\multicolumn{3}{@{}l}{\emph{Planar graphs (edge insertions)}}\\*
\hyperref[sec:planar-distance]{Undirected distance} & Return the distance between two vertices; edge weights are unit or positive. & \cite{abboud_popular_2016} \\
\addlinespace[2pt]
\hyperref[sec:planar-girth]{Directed girth} & Return the length of a shortest directed cycle; edge weights are unit or positive. & \cite{abboud_popular_2016} \\
\addlinespace[2pt]
\hyperref[sec:planar-diameter]{Diameter} & Return the largest distance between any two vertices in a graph with positive edge weights. & \cite{abboud_popular_2016} \\
\addlinespace[2pt]
\hyperref[sec:planar-matching]{Maximum-weight bipartite matching} & Return the maximum total weight of vertex-disjoint edges in a bipartite graph with positive edge weights. & \cite{abboud_popular_2016} \\
\midrule
\multicolumn{3}{@{}l}{\emph{Sets and arrays}}\\*
\hyperref[sec:set-disjointness]{Set disjointness} & Insert elements into a family of sets; are two queried sets disjoint? & \cite{kopelowitz_dynamic_2015} \\
\addlinespace[2pt]
\hyperref[sec:range-mode]{Range mode} & Change individual array entries; return a most frequent value in a queried subarray. & \cite{chan_linear-space_2014} \\
\addlinespace[2pt]
\hyperref[sec:langerman]{Langerman's problem} & Change individual entries of an integer array; does any nonempty prefix have sum zero? & \cite{patrascu_mihai_towards_2010} \\
\addlinespace[2pt]
\hyperref[sec:pagh]{Pagh's problem} & Add intersections of existing sets to a family; test whether a queried element belongs to a queried set. & \cite{patrascu_mihai_towards_2010} \\
\addlinespace[2pt]
\hyperref[sec:erickson]{Erickson's problem} & Increment an entire matrix row or column; return the largest matrix entry. & \cite{patrascu_mihai_towards_2010,ko_omega_2026} \\
\midrule
\multicolumn{3}{@{}l}{\emph{Matrices over prime fields of size $\Theta(N)$}}\\*
\hyperref[sec:app-inverse]{Inverse-entry zero testing} & Maintain a nonsingular matrix under single-entry updates; is a specified inverse entry zero? & \cite{van_den_brand_dynamic_2019} \\
\addlinespace[2pt]
\hyperref[sec:app-determinant]{Determinant} & Change individual matrix entries; return the determinant or test whether it equals one. & \cite{van_den_brand_dynamic_2019} \\
\midrule
\multicolumn{3}{@{}l}{\emph{Geometry}}\\*
\hyperref[sec:distinct-color-counting]{2D distinct-color counting} & Insert colored points; count the distinct colors inside a queried axis-aligned rectangle. & \cite{jin_tight_2022} \\
\addlinespace[2pt]
\hyperref[sec:hull]{3D convex-hull vertex counting} & Insert points; count the vertices of their convex hull. & \cite{dallant_conditional_2022} \\
\addlinespace[2pt]
\hyperref[sec:maxima]{3D maximal-point counting} & Insert points; count those for which no other point is at least as large in every coordinate. & \cite{dallant_conditional_2022,jin_tight_2022} \\
\midrule
\multicolumn{3}{@{}l}{\emph{Strings}}\\*
\hyperref[sec:internal-dictionary]{Internal dictionary matching} & Insert patterns specified as substrings of a fixed text; does a queried text interval contain an inserted pattern? & \cite{charalampopoulos_internal_2019} \\
\addlinespace[2pt]
\hyperref[sec:hamming-threshold]{Binary Hamming-distance threshold} & Change bits of a pattern; does it differ from an aligned substring of a fixed text in at most $k$ positions? The threshold $k$ is fixed and proportional to the pattern length. & \cite{clifford_upper_2018} \\
\addlinespace[2pt]
\hyperref[sec:active-documents]{Active-document existence} & Turn on documents from a fixed collection; does an active document contain both queried patterns? & \cite{jin_tight_2022} \\
\end{longtable}
\endgroup
These consequences follow from known reductions or simple adaptations of published constructions; for completeness and reader convenience, \cref{sec:graphs} gives the reductions and verifies their applicability to our precise model. The problem names in \cref{tab:intro-applications} link to their constructions, and \cref{tab:appendix-transfers} records the precise trade-offs.
All statements use the zero-error model, worst-case update time, and per-input expected query time of \cref{thm:reach-informal}, and many hold already under insertion-only updates. A reduction that answers the source query with target queries alone bounds $t_q$; one that also needs a constant number of target updates bounds $t_u+t_q$ (\cref{lem:app-transfer}).

\paragraph{Corollary II: streaming membership, amortized versus worst case.}
Our theorem also bounds the worst-case time per symbol of streaming regular-expression membership, a basic streaming task \cite{bille_regular_2010,naor_streaming_2022}. For each query set $S_i$, let $\Gamma_i$ be the alternation of the alphabet symbols corresponding to elements outside $S_i$, and consider the regular expression
\[
\bigcup_i \big(\#\, \Gamma_i^*\, c_i\big),
\]
where $\cup$ denotes alternation and the final symbol $c_i$ selects the query; when $S_i=[n]$, the $i$th branch is simply $\#c_i$. Data symbols can be recorded in constant time, keeping only their first occurrence, while on the final selector one scans the resulting list of distinct symbols. Consequently, this family admits deterministic $O(1)$ amortized processing time per arriving symbol on every input stream.

A matcher with small worst-case processing time per symbol would yield a multiphase data structure with the same update and query bounds. Writes performed while processing the final selector can be simulated in temporary workspace, since no later operation follows.

\begin{corollary}[Streaming regular-expression membership, informal]\label{thm:streaming-regex-informal}
There is a family of regular expressions of maximum length $M$ over a polynomial-size alphabet for which streaming membership, with an answer after every symbol, can be solved in deterministic $O(1)$ amortized time per arriving symbol on every input stream, but every zero-error cell-probe algorithm with word size $w=\Theta(\log M)$ requires
\[
\Omega\!\left(\left(\frac{\log M}{\log\log M}\right)^2\right)
\]
worst-case processing time on some symbol.
\end{corollary}

\noindent Binary encoding gives a weaker $\Omega(\log M/(\log\log M)^2)$ bound per bit. The symbols in the original alphabet occupy $O(1)$ words each. \Cref{sec:streaming-regex} gives both statements, their reductions, and a comparison with earlier separations.

\subsection{Proof overview}\label{sec:overview}

We sketch the proof of \cref{thm:main-informal}, which follows the framework of \cite{ko_omega_2026} and works in two layers. The outer layer is the Chronogram method of Fredman and Saks \cite{fredman_cell_1989}: we split the insertions into $\Theta(\log n/\log(t_u w))$ epochs, and show that a query must probe $\Omega(\log n/\log(t_u w))$ cells written in each of a constant fraction of them. The inner layer proves this bound for a single epoch, following \cite{ko_omega_2026}: if a query probes fewer cells of the epoch, we turn the data structure into a 2.5-round protocol whose probability of answering ``disjoint'' contradicts a communication lower bound. \Cref{sec:proof-main} gives the details.

\subsubsection{The outer layer: the Chronogram method}

The Chronogram method \cite{fredman_cell_1989} divides the insertions into $\ell$ epochs whose sizes grow geometrically going back in time. We number the epochs $\ell, \ell-1, \dots, 1$ in the order in which they occur, so that the larger epochs come first and epoch $1$ is the most recent. The main observation of Fredman and Saks is that we can decompose the cells by the epoch of their last write: let $C_i$ be the set of cells last written in epoch $i$. Each probe goes to at most one of the disjoint sets $C_1, \dots, C_\ell$, so for most epochs $i$, a random query probes only $O(t_q/\ell)$ cells of $C_i$ in expectation. We show that for a constant fraction of the epochs $i$, a random query must probe $\Omega(\log n/\log(t_u w))$ cells of $C_i$. With $\ell = \Theta(\log n/\log(t_u w))$, this gives $t_q = \Omega\big((\log n/\log(t_u w))^2\big)$.

\paragraph{The hard distribution.}
Recall that in multiphase Disjointness (\cref{prob:multiphase}), the query sets $S_1, \dots, S_m$ arrive first, the elements of $X$ are then inserted, and finally a query asks whether $S_q \cap X = \emptyset$. We choose disjoint blocks $B_1, \dots, B_\ell$ in $[n]$ and leave any remaining coordinates unused. In epoch $i$, $s_i = (t_u w)^{2i}$ random elements of $B_i$ are inserted, and every query set has $s_i$ random elements in $B_i$. In other words, this is the hard distribution of Babai, Frankl and Simon \cite{babai_complexity_1986} for Disjointness, two independent random sets of size about the square root of the universe, spread across the epochs: each block carries one copy of it, made slightly sparser so that the $\ell$ copies together meet with constant probability. Two features of this distribution matter.
\begin{itemize}[leftmargin=*]
  \item \emph{Recent epochs cannot record older ones.} The epoch sizes are geometric: going back in time, they grow by a factor $(t_u w)^2$. So the insertions after epoch $i$ number at most $s_{i-1} + s_{i-2} + \dots + s_1 \le 2 s_{i-1} = 2 s_i/(t_u w)^2$, and since each insertion writes at most $t_u$ cells of $w$ bits, they write only $O(s_i/(t_u w))$ bits. This is far too few to describe the $s_i$ elements inserted in epoch $i$, so a query that needs these elements must read cells of $C_i$. This is the key property of the Chronogram method.
  \item \emph{Every epoch matters.} Block $B_i$ has $\ell s_i^2$ elements. By the birthday paradox, two random sets of size $s_i$ in a universe of size $\ell s_i^2$ meet with probability about $1/\ell$. So the answer is ``disjoint'' with constant probability, and each block can change it.
\end{itemize}

\subsubsection{The inner layer: the 2.5-round game}

Suppose, for contradiction, that $t_q = o\big((\log n/\log(t_u w))^2\big)$, and fix an epoch $i$ in which a random query probes only $k = O(t_q/\ell)$ cells of $C_i$ in expectation. The idea of cell sampling, introduced to the Chronogram method by Larsen \cite{larsen_cell_2012}, is that a small random sample of $C_i$ then already suffices to answer a nontrivial fraction of the queries. If each cell of $C_i$ is sampled independently with probability $p = \Theta(1/(t_u w))$, then all $k$ cells of $C_i$ probed by a query are sampled with probability about $p^k$. Both $1/\ell$ and $\log(1/p)$ contribute a factor of $\log(t_u w)$, so
\[ p^{k} \;\ge\; \exp\Big(-O\Big(\frac{t_q \log^2(t_u w)}{\log n}\Big)\Big) \;=\; n^{-o(1)}. \]
This is the source of the quadratic form of the bound. To exploit it, we use the following variant of the 2.5-round multiphase communication game of \cite{ko_omega_2026}. Alice holds the query sets $S_1, \dots, S_m$, Bob holds $X$, and an omniscient adviser, Merlin, holds both. Merlin follows the prescribed message rule, supplying advice determined by the inputs.
\begin{Protocol}
\begin{itemize}[leftmargin=*]
  \item \emph{Round 0 (Merlin to Bob).} Merlin sends Bob a message $U$.
  \item \emph{Round 0.5 (Bob to Alice).} Bob sends Alice a message $\Phi$, which may depend on $U$, $X$ and his private randomness, but not on the query.
  \item \emph{Round 1 (Alice to Bob).} A random query index $Q \in [m]$ is revealed to Alice only, and she sends Bob a message $A_Q$.
  \item \emph{Round 2 (Bob).} Bob outputs whether $S_Q \cap X = \emptyset$, or FAIL.
\end{itemize}
\caption{The 2.5-round multiphase communication game (cf.\ \cite[Protocol~3]{ko_omega_2026})\label{prot:game}}
\end{Protocol}
\noindent This game is a restriction of \Patrascu's multiphase communication game, in which Alice and Bob both know the query and may then communicate freely. The required lower bound for \Patrascu's game remains open, while known bounds for more restricted games \cite{chattopadhyay_upper_2016,ko_adaptive_2020,dvorak_lower_2020} do not imply lower bounds for general data structures. The 2.5-round game sits in between. Bob's message $\Phi$ does not depend on the query, which makes the game amenable to information-theoretic analysis, and yet the game can still simulate a data structure.

\paragraph{Simulation.}
We turn the data structure into a protocol for the game as follows (Step~1, \cref{sec:step1}).
\begin{itemize}[leftmargin=*]
  \item Merlin's message $U$ is the list of all cells written in epochs $i, \dots, 1$, with their final contents.
  \item Bob's message $\Phi$ consists of the elements of $X$ outside $B_i$, the contents of all cells written after epoch $i$, and each cell of $C_i$ independently with probability $p = \eps/(32 t_u w)$, for a small constant $\eps > 0$. With $\Phi$, Alice knows the whole memory except the unsampled cells of $C_i$.
  \item Alice's message $A_Q$ is her simulation of query $Q$ on the memory she knows.
  \item Bob checks $A_Q$ against $U$. Alice's simulation can only go wrong on unsampled cells of $C_i$, which are listed in $U$, so Bob detects every mistake. He outputs Alice's answer if $A_Q$ is consistent with $U$, and FAIL otherwise.
\end{itemize}
Every non-FAIL answer is therefore correct. Moreover, Alice's simulation is correct whenever all $k$ cells of $C_i$ that the query probes are sampled, which happens with probability about $p^k = n^{-o(1)}$.

Verification is what distinguishes this protocol from the one-way games behind earlier bounds \cite{larsen_cell_2012,larsen_super-logarithmic_2025}. There, Alice answers on her own and cannot tell whether her simulation is correct, since an unsampled cell of $C_i$ looks to her like a cell that was not written in epoch $i$ at all; a failed simulation is indistinguishable from a successful one. Here, Bob knows exactly when the simulation fails.

The protocol nevertheless reveals little about block $i$. Bob's message reveals at most $\eps s_i/4$ bits about $X \cap B_i$, a small fraction of its entropy, and Alice's simulation has $O(t_q w)$ bits. Merlin's list $U$ is long, but it is shared by all $m$ query sets, so it reveals little about the random query set $S_Q$. Yet Bob learns, with probability $n^{-o(1)}$, that $S_Q$ and $X$ are disjoint. The communication lower bound rules out this combination of small information loss and substantial probability of a certified ``disjoint'' answer.

\subsubsection{From discrepancy to one-sided corruption}

In \cite{ko_omega_2026}, Bob outputs a random guess when verification fails, and the contradiction comes from the protocol's advantage over random guessing. After conditioning on suitable message values and smoothing, the inputs have nearly full min-entropy and are nearly independent. For Inner Product the proof ends with a discrepancy bound: under every product distribution with nearly full min-entropy, the answer is nearly unbiased. That discrepancy property does not hold for Disjointness under the distributions needed here.

Verification allows us to use a weaker, one-sided property. \cite[Remark~3.4]{ko_omega_2026} already relates detected failure to the protocol's acceptance probability. We exploit the certified answer itself: Bob outputs FAIL instead of guessing, and we condition on his returning ``disjoint''. Every answer he returns is correct. We therefore need only that every product distribution with nearly full min-entropy puts noticeable mass on intersecting pairs. This one-sided corruption bound does hold for Disjointness: for sets of size $\sqrt N$ in a universe of size $N$, it is the classic rectangle bound of Babai, Frankl and Simon \cite{babai_complexity_1986}, which we extend to our sparser sets in \cref{sec:corruption}. Thus a one-sided corruption bound compatible with the epoch distribution replaces discrepancy. Extending the argument to another block function would also require a suitable product hard distribution with constant probability of the accepted global answer.

\subsubsection{The complete argument}

Our proof proceeds in four steps. As in \cite{ko_omega_2026}, the first is a simulation, and the other three prove a communication lower bound.
\begin{enumerate}[leftmargin=*]
  \item \textbf{Simulation} (\cref{sec:step1}). A data structure violating the proposed bound yields a protocol for the game of \cref{prot:game} in which Bob outputs ``disjoint'' with probability $n^{-o(1)}$, and every answer he outputs is correct.
  \item \textbf{A monochromatic near-rectangle} (\cref{sec:step2}). Among runs on which Bob returns ``disjoint'', choose a suitable value of the messages and auxiliary conditioning variables. Then restrict $S_Q\cap B_i$ to an event of constant conditional probability. After this smoothing step, the two block inputs have nearly full min-entropy and small mutual information. They are always disjoint because Bob's answer is certified.
  \item \textbf{A nearly monochromatic rectangle} (\cref{cor:rect}). By Pinsker's inequality, replacing this distribution by the product of its marginals changes it little in total variation distance. The resulting rectangle therefore assigns little mass to intersecting pairs.
  \item \textbf{Corruption} (\cref{cor:product}). Every product of two distributions on sparse subsets of $B_i$ with nearly full min-entropy puts mass $\Omega(1/\ell)$ on intersecting pairs, a contradiction.
\end{enumerate}
The data structure may also be randomized, with expected query time, as in \cite{larsen_super-logarithmic_2025}: we fix its coins, and Alice stops each simulation after $8 t_q$ probes, which costs only a constant factor.

\paragraph{Organization.}
\Cref{sec:prelim} gives the model conventions and information-theoretic preliminaries. \Cref{sec:proof-main} proves the main theorem using the sparse-set corruption bound established in \cref{sec:corruption}. \Cref{sec:graphs} contains the dynamic reductions and the streaming application.

\section{Preliminaries}\label{sec:prelim}

\subsection{Cell-probe conventions}\label{sec:model}

Cells and addresses have $w$ bits. Computation and temporary workspace within an operation are free; all input-dependent information retained between operations resides in memory and is accessed through charged probes. Preprocessing is unrestricted and uncharged. As in~\cite{larsen_super-logarithmic_2025}, a random tape is drawn independently of the input, correctness holds for every input and tape, update time is worst case over inputs and tapes, and query time is expected over the tape for each fixed valid input history. We normalize $t_u,t_q\ge1$, take $t_u$ to be an integer, and state asymptotic bounds for sufficiently large instances.

\subsection{Information theory}\label{sec:prelim-it}

We collect the information-theoretic facts used in the proof; see \cite{cover_elements_2006} for background. Logarithms are base $2$ unless written $\ln$.

\begin{definition}[Entropy, mutual information, divergence]
The entropy of a random variable $X$ is $H(X) := \sum_x \Pr[X = x] \log\frac{1}{\Pr[X = x]}$, and the conditional entropy is $H(X | Y) := \E_{y \sim Y}\big[H(X | Y = y)\big]$. The mutual information between $X$ and $Y$, conditioned on $Z$, is $I(X; Y | Z) := H(X | Z) - H(X | Y, Z)$; when we need to specify the distribution $\mu$ of $(X, Y)$, we write $I_\mu(X; Y)$. The KL divergence and the total variation distance between two distributions $P$ and $Q$ are
\[ D(P \,\|\, Q) := \sum_x P(x)\log\frac{P(x)}{Q(x)} \qquad\text{and}\qquad \norm{P - Q}_{TV} := \frac12\|P - Q\|_1. \]
\end{definition}

\noindent For a random variable $X$ and an event $E$, we write $X|_E$ for the distribution of $X$ conditioned on $E$, for instance $X|_{Y = y}$. A random variable in place of a distribution, as in $D(X \,\|\, Y)$, stands for its distribution.

\begin{fact}[Basic properties]\label{fact:basic}
For any random variables $A$, $B$, $C$ and $D$:
\begin{enumerate}[label=(\alph*)]
  \item If $A$ takes at most $2^\lambda$ values, then $H(A) \le \lambda$. Moreover, $0 \le H(A | B) \le H(A)$.
  \item $I(A; B | C) = I(B; A | C)$ and $0 \le I(A; B | C) \le H(A | C)$. If $A$ and $B$ are independent given $C$, for instance if $A$ is a function of $C$, then $I(A; B | C) = 0$.
  \item $I(A; B | C) = \E_{c \sim C}\big[I(A; B | C = c)\big]$.
  \item $D(P \,\|\, Q) \ge 0$ for any two distributions $P$ and $Q$.
  \item If $B'$ is a function of $B$, then $I(A; B' | C) \le I(A; B | C)$.
  \item If $D$ is independent of $(A, B, C)$, then $I(A; B | C, D) = I(A; B | C)$.
\end{enumerate}
\end{fact}

\begin{fact}[Chain rules]\label{fact:chain}
\begin{enumerate}[label=(\alph*)]
  \item For any random variables $A_1, \dots, A_k$, $B$ and $C$,
  \[ I(A_1 \cdots A_k; B | C) \;=\; \sum_{j=1}^{k} I\big(A_j; B \bigm| C A_1 \cdots A_{j-1}\big). \]
  \item For any two pairs of random variables $(X, Y)$ and $(X', Y')$,
  \[ D\big(XY \,\|\, X'Y'\big) \;=\; D\big(X \,\|\, X'\big) + \E_{x \sim X}\Big[D\big(Y|_{X = x} \,\big\|\, Y'|_{X' = x}\big)\Big]. \]
\end{enumerate}
\end{fact}

\begin{fact}\label{fact:drop-cond}
For any random variables $A$, $B$, $C$ and $D$, if $I(B; D | A C) = 0$, then $I(A; B | C D) \le I(A; B | C)$.
\end{fact}
\begin{proof}
By the chain rule (\cref{fact:chain}(a)) and \cref{fact:basic}(b),
\[ I(A; B | C D) \;\le\; I(A D; B | C) \;=\; I(A; B | C) + I(B; D | A C) \;=\; I(A; B | C). \qedhere \]
\end{proof}

\begin{fact}[KL divergence and mutual information]\label{fact:kl-mi}
For any random variables $A$ and $B$, $I(A; B) = \E_{b \sim B}\big[D\big(A|_{B=b} \,\|\, A\big)\big]$.
\end{fact}

\begin{fact}[Pinsker's inequality]\label{fact:pinsker}
For any two distributions $P$ and $Q$,
\[ \norm{P - Q}_{TV} \;\le\; \sqrt{(\ln 2/2)\, D(P \,\|\, Q)}. \]
\end{fact}

\subsection{Min-entropy}\label{sec:prelim-minent}

The min-entropy of a random variable $A$ is $\Hmin(A) := -\log \max_a \Pr[A = a]$, and for a distribution $\nu$ we write $\Hmin(\nu)$ likewise. We use the average min-entropy of \cite{dodis_fuzzy_2008},
\[ \tHmin(A | B) := -\log \E_{b \sim B}\Big[\max_a \Pr[A = a | B = b]\Big]. \]

\begin{fact}[\cite{dodis_fuzzy_2008}]\label{fact:avg-minent}
Let $A$, $B$, $C$ be random variables.
\begin{enumerate}[label=(\alph*)]
  \item If $B$ has at most $2^\lambda$ possible values, then $\tHmin(A | B, C) \ge \tHmin(A | C) - \lambda$.
  \item For every $r > 0$, $\Pr_{b \sim B}\big[\Hmin(A | B = b) \ge \tHmin(A | B) - r\big] \ge 1 - 2^{-r}$.
  \item If $B'$ is a function of $B$, then $\tHmin(A | B') \ge \tHmin(A | B)$.
\end{enumerate}
\end{fact}

\begin{fact}\label{fact:event}
Let $E \in \{0,1\}$ be a random variable with $\Pr[E = 1] > 0$. Then $\tHmin(X | Y, E {=} 1) \ge \tHmin(X | Y) - \log(1/\Pr[E = 1])$, where the left side is with respect to the distribution conditioned on $E = 1$.
\end{fact}
\begin{proof}
For every $y$ with $\Pr[E = 1 | Y = y] > 0$, we have $\Pr[X = x | Y = y, E = 1] \le \Pr[X = x | Y = y]/\Pr[E = 1 | Y = y]$ and, by Bayes' rule, $\Pr[Y = y | E = 1]/\Pr[E = 1 | Y = y] = \Pr[Y = y]/\Pr[E = 1]$. Summing over such $y$, and then over all $y$,
\[ \E_{y \sim Y | E = 1} \max_x \Pr[X = x | Y = y, E = 1] \;\le\; \sum_{y} \frac{\Pr[Y = y] \max_x \Pr[X = x | Y = y]}{\Pr[E = 1]} \;=\; \frac{2^{-\tHmin(X | Y)}}{\Pr[E = 1]}. \qedhere \]
\end{proof}

\begin{fact}[Conditioning on one side]\label{fact:one-side}
Let $A, B$ be random variables, and let $E \in \{0,1\}$ be a function of $A$ with $\pi := \Pr[E = 1] > 0$. Then
\[ I(A; B | E {=} 1) \le \frac{I(A; B)}{\pi} \qquad\text{and}\qquad \Hmin(B | E {=} 1) \ge \Hmin(B) - \log\frac1\pi, \]
where the left sides are with respect to the distribution conditioned on $E = 1$.
\end{fact}
\begin{proof}
Since $E$ is a function of $A$, $I(B; E | A) = 0$, so \cref{fact:drop-cond} and \cref{fact:basic}(b) and (c) give
\[ I(A; B) \;\ge\; I(A; B | E) \;\ge\; \pi \cdot I(A; B | E {=} 1). \]
For the second bound, $\Pr[B = b | E = 1] \le \Pr[B = b]/\pi$ for every $b$.
\end{proof}

\noindent The next two facts relate min-entropy to flat distributions and Shannon entropy. For a finite set $\cS$, $\cU_{\cS}$ denotes the uniform distribution on $\cS$; thus $D(\mu \,\|\, \cU_{\cS}) = \log\abs{\cS} - H(\mu)$.

\begin{fact}[Flat decomposition; folklore, see e.g.\ {\cite[Lemma~6.10]{vadhan_pseudorandomness_2012}}]\label{fact:flat}
If $\Hmin(\mu) \ge \log K$ for an integer $K \ge 1$, then $\mu$ is a convex combination of uniform distributions $\cU_{T}$ on $K$-element sets $T$.
\end{fact}

\begin{fact}[Smoothing \cite{jain_property_2009}; see {\cite[Claim~2.16]{ko_omega_2026}}]\label{fact:smooth}
If $D(\mu \,\|\, \cU_{\cS}) \le t$ with $t > 1$, then for every $\alpha \ge 4$ there exists an event $E$ such that $\mu(E) \ge 1 - 1/\alpha$ and $\Hmin(\mu|_E) \ge \log\abs{\cS} - 2\alpha t$.
\end{fact}

\subsection{Rectangles}\label{sec:prelim-rect}

We express the size of a combinatorial rectangle $\mathcal{A}\times\mathcal{B}$ through min-entropy. This lets us treat product distributions and distributions that are close to products in a common notation.

\begin{definition}[Deficit]\label{def:deficit}
Let $\cS$ be a finite family of sets. The \emph{deficit} of a distribution $\nu$ on $\cS$ is $\log\abs{\cS} - \Hmin(\nu)$.
\end{definition}

\noindent The deficit is zero exactly for the uniform distribution on $\cS$, and the uniform distribution on a subfamily $\mathcal{A} \subseteq \cS$ has deficit $\log(\abs{\cS}/\abs{\mathcal{A}})$. A deficit bound of $d$ therefore limits the mass of every point to $2^d/\abs{\cS}$. It measures how much the distribution can concentrate relative to uniform. Below, $(A,B) \sim \mu$ for a distribution $\mu$ on $\cS \times \cS$ with marginals $\mu^A, \mu^B$.

\begin{definition}[Rectangles, near-rectangles, monochromatic]\label{def:rect}
Let $d, \kappa, \delta \ge 0$.
\begin{enumerate}[label=(\alph*)]
  \item $\mu$ is a \emph{$d$-rectangle} if $\mu = \mu^A \otimes \mu^B$ and both marginals have deficit at most $d$. That is,
  \begin{align*}
      \Hmin(\mu^A) \;\ge\; \log\abs{\cS} - d \qquad\text{and}\qquad \Hmin(\mu^B) \;\ge\; \log\abs{\cS} - d.
  \end{align*}
  \item $\mu$ is a \emph{$(d,\kappa)$-near-rectangle} if both marginals have deficit at most $d$ and $I_\mu(A;B) \le \kappa$.
  \item $\mu$ is \emph{$\delta$-monochromatic} if $\Pr_\mu[A \cap B \ne \emptyset] \le \delta$, and \emph{monochromatic} if it is $0$-monochromatic, i.e.\ supported on disjoint pairs.
\end{enumerate}
\end{definition}

\noindent A $d$-rectangle is the distributional analog of a large combinatorial rectangle. For $\mathcal{A}, \mathcal{B} \subseteq \cS$, the uniform distribution $\cU_{\mathcal{A} \times \mathcal{B}} = \cU_{\mathcal{A}} \otimes \cU_{\mathcal{B}}$ is a $d$-rectangle if and only if $\abs{\mathcal{A}}, \abs{\mathcal{B}} \ge 2^{-d}\abs{\cS}$, since $\cU_{\mathcal{A}}$ has deficit $\log(\abs{\cS}/\abs{\mathcal{A}})$. Conversely, every $d$-rectangle is a convex combination of uniform distributions on combinatorial rectangles of almost this size.

\begin{corollary}[Flat decomposition of rectangles]\label{cor:flat-rect}
Let $\mu$ be a $d$-rectangle on $\cS \times \cS$, and let $K$ be an integer with $1 \le K \le 2^{-d}\abs{\cS}$. Then there are subfamilies $\mathcal{A}_j, \mathcal{B}_k \subseteq \cS$ of size $K$ and weights $p_j, q_k \ge 0$ with $\sum_j p_j = \sum_k q_k = 1$ such that
\[ \mu \;=\; \sum_{j,k} p_j q_k\, \cU_{\mathcal{A}_j \times \mathcal{B}_k}. \]
\end{corollary}
\begin{proof}
Both marginals of $\mu$ have min-entropy at least $\log\abs{\cS} - d \ge \log K$. So by \cref{fact:flat}, $\mu^A = \sum_j p_j\, \cU_{\mathcal{A}_j}$ and $\mu^B = \sum_k q_k\, \cU_{\mathcal{B}_k}$ for such subfamilies and weights. Therefore
\[ \mu \;=\; \mu^A \otimes \mu^B \;=\; \sum_{j,k} p_j q_k\, \cU_{\mathcal{A}_j} \otimes \cU_{\mathcal{B}_k}. \qedhere \]
\end{proof}

\noindent A near-rectangle relaxes the product condition: its two sides may be correlated, and $\kappa$ bounds their mutual information. In particular, a $d$-rectangle is a $(d, 0)$-near-rectangle. Replacing a near-rectangle by the product of its marginals gives a rectangle that is close in total variation distance (see, e.g., \cite{ko_omega_2026}).

\begin{lemma}[From near-rectangles to rectangles]\label{lem:rect}
If $\mu$ is a $\delta$-monochromatic $(d,\kappa)$-near-rectangle, then $\mu^A \otimes \mu^B$ is a $\big(\delta + \sqrt{(\ln 2/2)\,\kappa}\big)$-monochromatic $d$-rectangle.
\end{lemma}
\begin{proof}
Let $R := \mu^A \otimes \mu^B$. It has the same marginals as $\mu$, so it is a $d$-rectangle. By the chain rule for KL divergence (\cref{fact:chain}(b)) and \cref{fact:kl-mi}, $D(\mu \,\|\, R) = \E_{b \sim \mu^B}\big[D(A|_{B = b} \,\|\, \mu^A)\big] = I_\mu(A; B) \le \kappa$, so $\norm{\mu - R}_{TV} \le \sqrt{(\ln 2/2)\,\kappa}$ by \cref{fact:pinsker}. Therefore
\[ \Pr_{(A, B) \sim R}\big[A \cap B \ne \emptyset\big] \;\le\; \Pr_{(A, B) \sim \mu}\big[A \cap B \ne \emptyset\big] + \norm{\mu - R}_{TV} \;\le\; \delta + \sqrt{(\ln 2/2)\,\kappa}. \qedhere \]
\end{proof}

\section{Main proof}\label{sec:proof-main}

In this section, we prove our main theorem, the formal version of \cref{thm:main-informal}.

\begin{theorem}[Main theorem]\label{thm:main}
Every zero-error cell-probe data structure for insert-only multiphase Disjointness with $m = n$ query sets, word size $w = \Theta(\log n)$, worst-case update time $t_u$ and expected query time $t_q$ has
\[ t_q = \Omega\!\left(\left(\frac{\log n}{\log(t_u w)}\right)^{2}\right). \]
In particular, $\max(t_u,t_q)=\Omega\big((\log n/\log\log n)^2\big)$.
\end{theorem}

We follow the outline in \cref{sec:overview}. After defining the hard distribution, we turn a data structure violating the bound into a 2.5-round protocol (\cref{sec:step1}). The communication lower bound in \cref{sec:step2} rules out this protocol, using the corruption bound proved in \cref{sec:corruption}.

\subsection{The hard distribution}\label{sec:dist}

Let $\beta := (t_u w)^4$, and choose the largest integer $\ell\ge1$ such that $\ell\sum_{i=1}^{\ell}\beta^i\le n$. The bound is trivial when $t_u w>n^{1/95}$, so we work in the complementary regime; in particular, such an $\ell$ exists. Choose pairwise disjoint blocks $B_1,\dots,B_\ell\subseteq[n]$ with
\[ \abs{B_i} = n_i := \ell\beta^i, \qquad\text{and let}\qquad s_i := \sqrt{n_i/\ell} = \beta^{i/2}. \]
The remaining coordinates occur in neither the query sets nor $X$. We keep $m=n$ query sets. Since $t_u$ and $w$ are integers, all block and set sizes are integers.
Epochs are indexed backward in time: epoch $\ell$ comes first and epoch $1$ comes last. Under the distribution $\cD$ of $(\vS, X)$, independently:
\begin{itemize}[leftmargin=*]
  \item each query set $S_j$, $j \in [m]$, consists of $s_i$ uniformly random elements of each block $B_i$; write $S_j^{(i)} := S_j \cap B_i$ and $\vS := (S_1,\dots,S_m)$;
  \item in epoch $i$, $s_i$ uniformly random elements of $B_i$ are inserted into $X$, in increasing order. We write $X^{(i)} := X \cap B_i$ for the set inserted in epoch $i$, $X^{(>i)} := X \cap (B_{i+1} \cup \dots \cup B_\ell)$, and $X^{(\ne i)} := X \setminus B_i$.
\end{itemize}
The query index $Q$ is uniform in $[m]$ and independent of $(\vS, X)$. We write $\Pr_{\cD,Q}$ and $\E_{\cD,Q}$ for probability and expectation over $(\vS, X) \sim \cD$ and $Q$.
In each block, $S_Q$ and $X$ share $s_i^2/n_i = 1/\ell$ elements in expectation over $\cD$ and $Q$; since the blocks are independent, $S_Q$ and $X$ are disjoint with probability at least $1/4$ (\cref{clm:params}(d) and \cref{fig:hard-dist}). The \emph{Chronogram epochs} are $\cI := \{ i \in [\ell] : i \ge 2(\ell+2)/3 \}$; any range of $\Omega(\ell)$ epochs with $s_i \ge n^{\Omega(1)}$ would do.

\begin{figure}[t]
\centering
\begin{tikzpicture}[x=0.42cm, y=0.42cm,
  lab/.style={font=\small},
  one/.style={fill=black!20}]
  \providecommand{\bitrow}[3]{%
    \foreach \b [count=\c from 0] in {#3} {%
      \ifnum\b=1 \fill[one] (#1+\c,#2) rectangle (#1+\c+1,#2+1);\fi
      \draw (#1+\c,#2) rectangle (#1+\c+1,#2+1);
      \node[font=\scriptsize] at (#1+\c+0.5,#2+0.5) {\b};}}
  \node[lab, anchor=east] at (-0.2,1.5) {$X$};
  \node[lab, anchor=east] at (-0.2,0.5) {$S_q$};
  \bitrow{0}{1}{0,1,0,0,0,0,1,0,0,0,0,0,1,0}
  \bitrow{0}{0}{0,0,0,1,0,0,0,0,1,0,1,0,0,0}
  \draw[very thick] (0,0) rectangle (14,2);
  \node[lab] at (7,-0.6) {$B_\ell$ ($n_\ell$ coordinates)};
  \bitrow{15}{1}{0,0,1,0,0,0,0,1}
  \bitrow{15}{0}{1,0,0,0,0,1,0,0}
  \draw[very thick] (15,0) rectangle (23,2);
  \node[lab] at (19,-0.6) {$B_{\ell-1}$};
  \node at (24.5,1) {$\cdots$};
  \bitrow{26}{1}{0,0,1,0}
  \bitrow{26}{0}{0,1,0,0}
  \draw[very thick] (26,0) rectangle (30,2);
  \node[lab] at (28,-0.6) {$B_1$};
  \draw[->, thick] (0,3) -- (30,3);
  \node[lab, anchor=south] at (2.2,3.05) {epoch $\ell$ first};
  \node[lab, anchor=south] at (15,3.05) {time};
  \node[lab, anchor=south] at (27.8,3.05) {epoch $1$ last};
  \path (-1.0,0) (31.0,0); 
  \node[lab] at (15,-1.9) {In each block $B_i$: $s_i$ ones in each row, placed independently at random; \quad $n_i = \ell\beta^i$, \ $s_i = \sqrt{n_i/\ell}$.};
\end{tikzpicture}
\caption{The hard distribution (not to scale), with $X$ and $S_q$ written as characteristic vectors. The ones of $X$ in block $B_i$ are inserted during epoch $i$. Block sizes grow by a factor $\beta = (t_u w)^4$ going back in time. In each block, the expected number of common ones is $s_i^2/n_i = 1/\ell$, so $S_Q$ and $X$ are disjoint with constant probability.}
\label{fig:hard-dist}
\end{figure}

The following bounds summarize the parameters of the hard distribution.

\begin{claim}[Parameters]\label{clm:params}
The hard distribution has the following properties.
\begin{enumerate}[label=(\alph*)]
  \item $\ell = \Theta(\log n/\log(t_u w))$, $\ell \ge 20$, and $\abs{\cI} \ge (\ell - 4)/3$.
  \item For $i \in \cI$: $s_i \ge n^{1/3}$.
  \item $s_i/s_{i-1} = \sqrt\beta = (t_u w)^2$, $\sum_{j < i} s_j \le 2 s_{i-1}$, and $s_\ell \le \sqrt n$.
  \item $\E_{\cD,Q} \left[ \abs{S_Q^{(i)} \cap X^{(i)}} \right] = s_i^2/n_i = 1/\ell$ for every $i$, and $\Pr_{\cD,Q}[S_Q \cap X = \emptyset] \ge (1 - 1/\ell)^\ell \ge 1/4$ for $\ell \ge 2$.
\end{enumerate}
\end{claim}
\begin{proof}
We verify the parameter bounds, including the unused coordinates.
\begin{enumerate}[label=(\alph*), leftmargin=*]
  \item By maximality of $\ell$ and $\beta\ge2$,
  \[
  \ell\beta^\ell\le n < (\ell+1)\sum_{j=1}^{\ell+1}\beta^j
  \le 2(\ell+1)\beta^{\ell+1}.
  \]
  Thus $\ell\le\log n$, and $\ell=\Theta(\log n/\log\beta)$. Also $\abs{\cI}\ge \ell-2(\ell+2)/3=(\ell-4)/3$. Since $t_u w\le n^{1/95}$, we have $\log\beta\le(4/95)\log n$ and
  \[
  \ell > \frac{\log n-\log(2(\ell+1))}{\log\beta}-1
  \ge \frac{95}{4}\left(1-\frac{\log(2(\log n+1))}{\log n}\right)-1
  \ge20.
  \]
  \item Since $w=\Theta(\log n)$ and $\ell\le\log n$, we have $\beta\ge2(\ell+1)$. For $i\ge2(\ell+2)/3$, $s_i=\beta^{i/2}\ge(\beta^{\ell+2})^{1/3}>(\beta n/(2(\ell+1)))^{1/3}\ge n^{1/3}$.
  \item The ratio is immediate; $\sum_{j<i} s_j \le s_{i-1}\sum_{r \ge 0} \beta^{-r/2} \le 2 s_{i-1}$ since $\beta \ge 4$; and $s_\ell^2 = n_\ell/\ell \le n$.
  \item Each element of $B_i$ lies in $S_Q^{(i)}$ with probability $s_i/n_i$ and in $X^{(i)}$ with probability $s_i/n_i$, independently, so the expectation is $n_i (s_i/n_i)^2 = 1/\ell$. By Markov, block $i$ is disjoint with probability at least $1 - 1/\ell$, and the blocks are independent. \qedhere
\end{enumerate}
\end{proof}

\paragraph{Setup.}
Fix the random tape by averaging over the hard distribution. We may assume that $t_u w \le n^{1/95}$ and $t_q \le \log^2 n$, since otherwise the bound is trivial. So the data structure is deterministic, with zero error and worst-case update time $t_u$, and the number $T(q)$ of cells read by a query $q \in [m]$ is a function of $(\vS, X)$ with $\E_{\cD,Q}[T(Q)] \le t_q$. By \cref{clm:params}, $\ell \ge 20$ and $s_i \ge n^{1/3}$ for $i \in \cI$. Since $w = \Theta(\log n)$, also $t_u w \ge 64/\eps$, where $\eps = 1/40$ is the constant of \cref{lem:corruption}.

\paragraph{Epochs.}
The updates are processed in decreasing order of the epoch index, so the larger epochs come first. We follow the notation and the setup of \cite{larsen_super-logarithmic_2025}. For $i \in \cI$, let $M_{>i}$ be the memory after Phase I and epochs $\ell, \dots, i+1$; it is a function of $(\vS, X^{(>i)})$. For $j \in [\ell]$, let $C_j$ be the set of cells last written in epoch $j$, and let $C_{<j} := \bigcup_{j' < j} C_{j'}$ be the set of cells written after epoch $j$. The sets $C_1, \dots, C_\ell$ are pairwise disjoint. Let $t_j(q)$ be the number of distinct cells of $C_j$ read by query $q$, so that
\[ \sum_{j \in [\ell]} t_j(q) \;\le\; T(q), \qquad\text{and let}\qquad \tau_j \;:=\; \E_{\cD,Q}\big[t_j(Q)\big]. \]

\subsection{Simulation theorem}\label{sec:step1}

The following theorem converts a data structure into a protocol for the 2.5-round game of \cref{prot:game} (cf.\ \cite[Protocol~3]{ko_omega_2026}), in which Bob verifies Alice's simulation of the query against the actual memory. We adapt the protocol of \cite[Protocol~4]{ko_omega_2026} and the analysis of \cite[Theorem~3.2]{ko_omega_2026}; numbered references to that work use arXiv version~2. The main difference is Bob's output when the simulation fails: in \cite{ko_omega_2026}, he outputs $\mathrm{FAIL}$ together with an independent random guess, while here he outputs only $\mathrm{FAIL}$. In addition, the epoch $i$ is fixed in advance rather than chosen by Merlin, Merlin's message lists only the cells written in epochs $i, \dots, 1$ rather than all updates, Bob sends Alice all elements of $X$ outside $B_i$, and, since $t_q$ bounds the query time only on average, Alice stops her simulation after $8 t_q$ probes.

Let $U$ consist of two lists, of the cells of $C_i$ and of the cells of $C_{<i}$, each cell with its final content, that is, its content after Phase~II. Let $p := \eps/(32\, t_u w)$. Bob's private randomness $R_{\mathrm{Bob}}$ consists of an independent bit $R_a$ with $\Pr[R_a = 1] = p$ for every cell $a$, and $(\vS, X)$, $Q$ and $R_{\mathrm{Bob}}$ are mutually independent. The length $\abs{\cdot}$ of a message is measured in bits, padded to a fixed length, so that a message of length $\lambda$ takes at most $2^\lambda$ values.

\begin{theorem}[Simulation theorem]\label{thm:too-good}
For every $i \in \cI$ there is a 2.5-round protocol $G_i$ for the game of \cref{prot:game}, in which Merlin sends $U$ and Bob's message has the form $\Phi = (X^{(\ne i)}, \Psi)$, with the following properties, where $B_Q$ is the indicator that Bob outputs ``disjoint'' and $\rho_1 := \Pr_{\cD, R_{\mathrm{Bob}}, Q}[B_Q = 1]$.
\begin{enumerate}[label=(\alph*)]
  \item \emph{(Zero error)} If Bob does not output $\mathrm{FAIL}$, his output is the correct answer. In particular, $B_Q = 1$ implies $S_Q \cap X = \emptyset$.
  \item \emph{(Lengths of $U$ and $\Psi$)} $\abs{U} \le 5\, t_u w\, s_i$ and $\abs{\Psi} \le \eps s_i/4$.
  \item \emph{(Length of Alice's messages)} $\abs{A_q} \le 17\, t_q w$ for every $q \in [m]$.
  \item \emph{(Probability of answering disjoint)} $\rho_1 \ge \frac18\, p^{8\tau_i} - e^{-n^{1/4}/3}$.
  \item \emph{(Structure)} $\Psi$ is a function of $(U, R_{\mathrm{Bob}})$, each $A_q$ is a function of $(\vS, \Phi)$, and $B_Q$ is a function of $(\Psi, A_Q, U)$.
\end{enumerate}
\end{theorem}

\begin{proof}
We simulate the data structure by the 2.5-round protocol $G_i$ of \cref{def:game}. A verified transcript may have either answer; $B_Q=1$ denotes only a verified ``disjoint'' answer. Part (e) is immediate from \cref{def:game}: Bob computes $\Psi$ from $U$ and $R_{\mathrm{Bob}}$, Alice computes $A_q$ from $\vS$ and $\Phi$, and Bob's decision reads only $\Psi$, which records whether he aborted, $A_Q$, which records whether the run finished and its output, and the final contents listed in $U$.

\begin{figure}[!t]
\centering
\fbox{\begin{minipage}{0.97\textwidth}
\begin{itemize}[leftmargin=*]
  \item Alice holds $\vS$. Bob holds $X$ and $R_{\mathrm{Bob}}$. Merlin holds $\vS$ and $X$.
  \item Merlin sends $U$ to Bob.
  \item As the $0.5$-round message, Bob lets $\Csam := \{a \in C_i : R_a = 1\}$. If $\abs{\Csam} > 2p\, t_u s_i$, Bob aborts: he sets $\Psi := \bot$, and outputs $\mathrm{FAIL}$ at the end. Otherwise $\Psi$ lists the cells of $C_{<i} \cup \Csam$ with their final contents. Bob sends $\Phi := (X^{(\ne i)}, \Psi)$ to Alice. The message depends on $U$, $X$ and $R_{\mathrm{Bob}}$, but not on the query $Q$.
  \item The query $Q \in [m]$ is then revealed to Alice. Alice computes $M_{>i}$ from $\vS$ and $X^{(>i)}$, and runs the query algorithm on $Q$. She answers a probe to a cell listed in $\Psi$ with its listed content, and a probe to any other cell with its content in $M_{>i}$. She stops the run if it has not finished after $\lfloor8 t_q\rfloor$ probes. She sends the transcript
  \[ A_Q \;:=\; \big((a_1, c_1), \dots, (a_r, c_r), z_{out}\big), \qquad r \le \lfloor8 t_q\rfloor, \]
  where $a_k$ is the $k$-th probed cell, $c_k$ the content she used, and $z_{out}$ the output, or a mark that the run was stopped.
  \item Bob \emph{verifies} $A_Q$ if (i) he did not abort, (ii) the run recorded in $A_Q$ finished, and (iii) $c_k$ is the final content of $a_k$ for every $k$ with $a_k \in C_i \cup C_{<i}$; he reads these final contents off $U$. If verification succeeds, he outputs $z_{out}$; otherwise he outputs $\mathrm{FAIL}$.
\end{itemize}
\end{minipage}}
\refstepcounter{Protocol}\label{def:game}
\vspace{4pt}
{\small Protocol \theProtocol: Simulation protocol $G_i$ (cf.\ \cite[Protocol~4]{ko_omega_2026})}
\end{figure}

\paragraph{Cost of the protocol.} Each cell is encoded by its $w$-bit address and $w$-bit content. We encode list lengths and the abort flag using $O(w)$ additional bits, and pad messages to fixed lengths. Merlin's message $U$ lists the cells of $C_i$ and of $C_{<i}$. Here $\abs{C_i} \le t_u s_i$, since each of the $s_i$ insertions of epoch $i$ writes at most $t_u$ cells, and similarly $\abs{C_{<i}} \le t_u \sum_{j<i} s_j \le 2 t_u s_{i-1}$ by \cref{clm:params}(c). Bob's list $\Psi$ has at most $\abs{C_{<i}} + 2p\, t_u s_i$ cells. Since $s_{i-1} = s_i/(t_u w)^2$ and $t_u w \ge 64/\eps$,
\begin{align*}
\abs{U} \;&\le\; 2w\big(\abs{C_i} + \abs{C_{<i}}\big) + O(w) \;\le\; 2w\big(t_u s_i + 2 t_u s_{i-1}\big) + O(w) \;\le\; 5\, t_u w\, s_i, \\
\abs{\Psi} \;&\le\; 2w\big(2 t_u s_{i-1} + 2p\, t_u s_i\big) + O(w) \;=\; \frac{4 s_i}{t_u w} + \frac{\eps s_i}{8} + O(w) \;\le\; \frac{\eps s_i}{4}.
\end{align*}
Alice's transcript records at most $8 t_q$ probes, together with the number of probes and the output or the stop mark, so
\[ \abs{A_q} \;\le\; 16\, t_q w + \lceil\log(\lfloor8 t_q\rfloor + 1)\rceil + 2 \;\le\; 17\, t_q w. \]
The header costs are absorbed since $s_i\ge n^{1/3}$ and $w=\Theta(\log n)$. This proves (b) and (c).

\paragraph{Verification and disjoint answers.} We first prove that verification certifies the true execution and succeeds when the sampled cells cover that execution, provided Bob does not abort and the query finishes within the probe limit. We then restrict to disjoint inputs to prove (d).

\begin{claim}[Verification]\label{lem:simulation}
Fix $(\vS, X)$ and $q \in [m]$. In (b) and (c), the probability is over $R_{\mathrm{Bob}}$.
\begin{enumerate}[label=(\alph*)]
  \item If Bob verifies $A_q$, then $A_q$ is the true execution of query $q$. In particular, its output is $\DISJ(S_q, X)$.
  \item $\Pr_{R_{\mathrm{Bob}}}[\text{Bob aborts}] \le e^{-n^{1/4}/3}$.
  \item If $T(q) \le 8 t_q$, then $\Pr_{R_{\mathrm{Bob}}}[\text{Bob verifies } A_q] \ge p^{t_i(q)} - e^{-n^{1/4}/3}$.
\end{enumerate}
\end{claim}
\begin{proof}
(a) If Bob verifies $A_q$, then every content Alice used is the final content. For a cell of $C_i \cup C_{<i}$, Bob checked this. Any other cell is not listed in $\Psi$ and is not written during epochs $i, \dots, 1$, so Alice used its content in $M_{>i}$, which is final. Hence her run is the true execution of query $q$, and its output is $\DISJ(S_q, X)$.

(b) $\abs{\Csam}$ is a sum of $\abs{C_i} \le t_u s_i$ independent indicators with mean $p$ each. By the Chernoff bound, applied with the upper bound $p\, t_u s_i$ on the mean,
\[ \Pr_{R_{\mathrm{Bob}}}\left[\abs{\Csam} > 2p\, t_u s_i\right] \;\le\; e^{-p\, t_u s_i/3}, \]
and $p\, t_u s_i = \eps s_i/(32 w) \ge n^{1/4}$.

(c) Let $P_q \subseteq C_i$ be the set of the $t_i(q)$ cells of $C_i$ probed by the true execution of query $q$. If Bob does not abort and $P_q \subseteq \Csam$, then every cell that the true execution probes is either listed in $\Psi$ or outside $C_i \cup C_{<i}$, so Alice reproduces the true execution, which finishes within $T(q) \le 8 t_q$ probes, and verification succeeds. Since $\Pr_{R_{\mathrm{Bob}}}[P_q \subseteq \Csam] = p^{t_i(q)}$, a union bound with (b) gives (c).
\end{proof}

Part (a) of the theorem follows from \cref{lem:simulation}(a), since Bob outputs $\mathrm{FAIL}$ unless he verifies $A_Q$. It remains to show (d). Let $F$ be the event that $S_Q \cap X = \emptyset$ and $T(Q) \le 8 t_q$; it is determined by $(\vS, X, Q)$. By Markov's inequality, $\Pr_{\cD,Q}[T(Q) > 8 t_q] \le 1/8$, and by \cref{clm:params}(d), $\Pr_{\cD,Q}[S_Q \cap X = \emptyset] \ge 1/4$, so $\Pr_{\cD,Q}[F] \ge 1/8$. If $F$ holds and Bob verifies $A_Q$, then the recorded output is ``disjoint'' by \cref{lem:simulation}(a), so $B_Q = 1$. Since $R_{\mathrm{Bob}}$ is independent of $(\vS, X, Q)$, \cref{lem:simulation}(c), applied for each fixed $(\vS, X, Q)$, and then the convexity of $x \mapsto p^x$ (with Jensen's inequality) give
\[ \rho_1 \;\ge\; \E_{\cD,Q}\Big[\mathbf{1}_F \cdot \big(p^{t_i(Q)} - e^{-n^{1/4}/3}\big)\Big] \;\ge\; \Pr_{\cD,Q}[F] \cdot p^{\E_{\cD,Q}[t_i(Q) | F]} - e^{-n^{1/4}/3}. \]
Since $t_i \ge 0$, $\E_{\cD,Q}[t_i(Q) | F] \le \tau_i/\Pr_{\cD,Q}[F] \le 8\tau_i$, which gives (d).
\end{proof}

\noindent As in \cite[Remark~3.4]{ko_omega_2026}, the protocol detects simulation failure: whenever Bob does not output $\mathrm{FAIL}$, he has verified Alice's transcript against $U$, and his output is correct. The remaining runs play no role in the argument: Steps 2--4 use only the event $B_Q = 1$, on which the answer is ``disjoint'' with certainty.

The next corollary quantifies the probability of answering ``disjoint'' when a query reads few cells from epoch $i$.

\begin{corollary}[Choice of parameters]\label{cor:too-good}
If $\tau_i = o\big(\log n/\log(t_u w)\big)$, then $\rho_1 \ge n^{-o(1)}$.
\end{corollary}
\begin{proof}
Since $t_u w \ge 64/\eps$, $\log(1/p) = \log(32\, t_u w/\eps) \le 2\log(t_u w)$, so $p^{8\tau_i} \ge 2^{-16\tau_i\log(t_u w)} = n^{-o(1)}$. By \cref{thm:too-good}(d), $\rho_1 \ge \frac18 n^{-o(1)} - e^{-n^{1/4}/3} \ge n^{-o(1)}$.
\end{proof}

\subsection{Communication lower bound}\label{sec:step2}

We prove that $\rho_1<n^{-1/6}$ for every $i\in\cI$ (\cref{lem:comm-lb}). Following \cite{ko_omega_2026,ko_adaptive_2020}, we collect the relevant messages and auxiliary variables into $Z$, including Bob's output $B_Q$. If disjoint answers occur too often, we can choose a value $Z=(z,1)$ and then restrict the query-side input $S_Q^{(i)}$ to an event $E$ of constant conditional probability. Under this joint conditioning, both block inputs have nearly full min-entropy and their mutual information is small. The smoothing event $E$ depends only on the query-side input once $z$ is fixed.

Every pair in the resulting distribution is disjoint, by \cref{thm:too-good}(a). Thus it is a monochromatic near-rectangle (\cref{def:rect}). The one-sided corruption bound rules out this distribution after we pass to the product of its marginals; this replaces the discrepancy step for Inner Product in \cite{ko_omega_2026}.

Fix $i \in \cI$ with $\rho_1 > 0$. Following \cite{ko_omega_2026}, we take
\[ Z \;:=\; \big(Z_{\mathrm{alice}}, B_Q\big), \qquad Z_{\mathrm{alice}} \;:=\; \big(\Phi, \Aall, Q, S_{<Q}\big), \]
where $S_{<Q} := (S_1, \dots, S_{Q-1})$, and $\Aall := (A_1, \dots, A_m)$ with $A_q$ the transcript that Alice sends when the query is $q$. The vector $\Aall$ is never transmitted; it is an analysis variable that permits direct sum over the query sets. The prefix $S_{<Q}$ supplies the conditioning required by the chain rule. As in \cite{ko_omega_2026}, $Z$ includes these transcripts but not Merlin's message $U$. By \cref{thm:too-good}(e), $\Aall$ is a function of $(\vS, \Phi)$, and $B_Q$ is a function of $(Z_{\mathrm{alice}}, U)$, since $\Psi$ is part of $\Phi$ and $A_Q$ is determined by $\Aall$ and $Q$. We only consider values of $Z$ with $B_Q = 1$, and write them as $(z, 1)$, where $z$ is a value of $Z_{\mathrm{alice}}$.

We use two facts throughout. First, Merlin's message is short compared to the number $m = n$ of query sets: by \cref{thm:too-good}(b) and \cref{fact:basic}(a),
\begin{equation}\label{eq:HU}
H(U) \;\le\; 5\, t_u w\, s_i \;\le\; n^{3/5},
\end{equation}
since $s_i \le s_\ell \le \sqrt n$ (\cref{clm:params}(c)) and $t_u w \le n^{1/95}$. Second, the following standard direct-sum identity holds.

\begin{claim}[Direct sum]\label{clm:direct-sum}
For any random variables $V$ and $W$ that are functions of $(\vS, X, R_{\mathrm{Bob}})$,
\[ I\big(S_Q; W \bigm| V, Q, S_{<Q}\big) \;=\; \frac1m\, I\big(\vS; W \bigm| V\big). \]
\end{claim}
\begin{proof}
Since $Q$ is uniform and independent of $(\vS, X, R_{\mathrm{Bob}})$, \cref{fact:basic}(c) and the chain rule (\cref{fact:chain}(a)) give
\[ I\big(S_Q; W \bigm| V, Q, S_{<Q}\big) \;=\; \frac1m \sum_{q \in [m]} I\big(S_q; W \bigm| V, S_{<q}\big) \;=\; \frac1m\, I\big(\vS; W \bigm| V\big). \qedhere \]
\end{proof}

\subsubsection{Information about \texorpdfstring{$S_Q$}{S\_Q}}

We first bound the divergence of the conditional query-side distribution from uniform, and then obtain high min-entropy by smoothing. Although $Z_{\mathrm{alice}}$ contains all $m$ Alice transcripts, they carry only $O(t_qw)$ bits about a random $S_Q$ on average. As in \cite{ko_omega_2026}, this yields a high-min-entropy distribution after restricting $S_Q^{(i)}$ to a likely event $E$.

\begin{proposition}[The query side]\label{prop:S-entropy}
If $\rho_1 \ge n^{-1/6}$, then
\[ \Pr_{z \sim Z_{\mathrm{alice}}|_{B_Q=1}}\left[\, \exists E \subseteq \binom{B_i}{s_i}:\;
\begin{aligned}
&\Pr\big[S_Q^{(i)} \in E \bigm| Z = (z, 1)\big] \;\ge\; \frac34 \quad\text{and} \\
&\Hmin\big(S_Q^{(i)} \bigm| Z = (z, 1),\, S_Q^{(i)} \in E\big) \;\ge\; \log\binom{n_i}{s_i} - 32\, n^{1/5}
\end{aligned}
\,\right] \;\ge\; \frac34. \]
\end{proposition}
\begin{proof}
First we upper bound $I\big(S_Q^{(i)}; Z_{\mathrm{alice}}, B_Q\big)$. Since $S_Q^{(i)}$ is part of $S_Q$, and $S_Q$ is independent of $(Q, S_{<Q})$, the chain rule (\cref{fact:chain}) gives
\[ I\big(S_Q^{(i)}; Z_{\mathrm{alice}}, B_Q\big) \;\le\; I\big(S_Q; \Phi, \Aall, Q, S_{<Q}, B_Q\big) \;=\; I\big(S_Q; \Phi, \Aall, B_Q \bigm| Q, S_{<Q}\big). \]
Since $B_Q$ is a function of $(\Phi, \Aall, Q, U)$, the chain rule and \cref{clm:direct-sum} (with $V$ constant) give
\[ I\big(S_Q; \Phi, \Aall, B_Q \bigm| Q, S_{<Q}\big) \;\le\; I\big(S_Q; \Phi, \Aall, U \bigm| Q, S_{<Q}\big) \;=\; \frac1m\, I\big(\vS; \Phi, \Aall, U\big). \]
By the chain rule and \cref{fact:basic},
\[ I\big(\vS; \Phi, \Aall, U\big) \;=\; I\big(\vS; \Phi, U\big) + I\big(\vS; \Aall \bigm| \Phi, U\big) \;\le\; H(U) + H(\Aall), \]
where $I(\vS; \Phi, U) \le I(\vS; X^{(\ne i)}, R_{\mathrm{Bob}}, U) = I(\vS; U | X^{(\ne i)}, R_{\mathrm{Bob}}) \le H(U)$ by \cref{fact:basic}(e) and the chain rule, since $\Psi$ is a function of $(U, R_{\mathrm{Bob}})$ (\cref{thm:too-good}(e)) and $\vS$ is independent of $(X^{(\ne i)}, R_{\mathrm{Bob}})$. By \eqref{eq:HU} and $m = n$, $H(U)/m \le n^{-2/5}$, and by \cref{thm:too-good}(c) and \cref{fact:basic}(a), $H(\Aall)/m \le 17 t_q w$. Since $t_q \le \log^2 n$ and $w = \Theta(\log n)$,
\[ I\big(S_Q^{(i)}; Z_{\mathrm{alice}}, B_Q\big) \;\le\; n^{-2/5} + 17 t_q w \;=\; O(\log^3 n). \]
Next we condition on $B_Q = 1$, an event of probability $\rho_1$. Let $\cU$ be the uniform distribution on $\binom{B_i}{s_i}$, which is the distribution of $S_Q^{(i)}$. By \cref{fact:kl-mi}, and since divergence is nonnegative (\cref{fact:basic}(d)),
\begin{align*}
I\big(S_Q^{(i)}; Z_{\mathrm{alice}}, B_Q\big) \;&=\; \E_{(z, b) \sim (Z_{\mathrm{alice}}, B_Q)}\Big[D\big(S_Q^{(i)}|_{Z = (z, b)} \,\big\|\, \cU\big)\Big] \\
\;&\ge\; \rho_1 \cdot \E_{z \sim Z_{\mathrm{alice}}|_{B_Q=1}}\Big[D\big(S_Q^{(i)}|_{Z = (z, 1)} \,\big\|\, \cU\big)\Big].
\end{align*}
Hence, since $\rho_1 \ge n^{-1/6}$,
\[ \E_{z \sim Z_{\mathrm{alice}}|_{B_Q=1}}\Big[D\big(S_Q^{(i)}|_{Z = (z, 1)} \,\big\|\, \cU\big)\Big] \;\le\; \frac{O(\log^3 n)}{\rho_1} \;\le\; n^{1/5}. \]
Finally, by Markov's inequality, $D\big(S_Q^{(i)}|_{Z = (z, 1)} \,\|\, \cU\big) \le 4 n^{1/5}$ with probability at least $3/4$ over $z \sim Z_{\mathrm{alice}}|_{B_Q=1}$. For such $z$, \cref{fact:smooth}, applied to the distribution of $S_Q^{(i)}$ given $Z = (z, 1)$ with $\cS = \binom{B_i}{s_i}$, $t = 4 n^{1/5}$ and $\alpha = 4$, gives the desired set $E$.
\end{proof}

\subsubsection{Information about \texorpdfstring{$X$}{X}}

Next, we show that conditioning on $Z = (z, 1)$ costs little min-entropy of $X^{(i)}$. The only information about $X^{(i)}$ in $Z_{\mathrm{alice}}$ is Bob's message $\Psi$, which is short, and conditioning on $B_Q = 1$ costs $\log(1/\rho_1)$ more bits.

\begin{proposition}[The update side]\label{prop:X-leak}
\[ \Pr_{z \sim Z_{\mathrm{alice}}|_{B_Q=1}}\left[\, \Hmin\big( X^{(i)} \bigm| Z = (z, 1) \big) \;\ge\; \log\binom{n_i}{s_i} - \frac{\eps s_i}{4} - \log\frac{1}{\rho_1} - 2 \,\right] \;\ge\; \frac34. \]
\end{proposition}
\begin{proof}
First we lower bound $\tHmin\big(X^{(i)} \bigm| Z_{\mathrm{alice}}\big)$. The set $X^{(i)}$ is uniform on $\binom{B_i}{s_i}$ and independent of $(\vS, X^{(\ne i)}, Q)$, and $Z_{\mathrm{alice}}$ is a function of $(\Psi, \vS, X^{(\ne i)}, Q)$ by \cref{thm:too-good}(e). By \cref{fact:avg-minent}(c), and then (a) with $\abs{\Psi} \le \eps s_i/4$ (\cref{thm:too-good}(b)),
\[ \tHmin\big(X^{(i)} \bigm| Z_{\mathrm{alice}}\big) \;\ge\; \tHmin\big(X^{(i)} \bigm| \Psi, \vS, X^{(\ne i)}, Q\big) \;\ge\; \log\binom{n_i}{s_i} - \frac{\eps s_i}{4}. \]
Next we condition on $B_Q = 1$, an event of probability $\rho_1$. By \cref{fact:event},
\[ \tHmin\big(X^{(i)} \bigm| Z_{\mathrm{alice}}, B_Q {=} 1\big) \;\ge\; \tHmin\big(X^{(i)} \bigm| Z_{\mathrm{alice}}\big) - \log\frac{1}{\rho_1} \;\ge\; \log\binom{n_i}{s_i} - \frac{\eps s_i}{4} - \log\frac{1}{\rho_1}, \]
where the left side is the average min-entropy with respect to the distribution conditioned on $B_Q = 1$:
\[ \tHmin\big(X^{(i)} \bigm| Z_{\mathrm{alice}}, B_Q {=} 1\big) \;=\; -\log \E_{z \sim Z_{\mathrm{alice}}|_{B_Q=1}}\Big[2^{-\Hmin(X^{(i)} \,|\, Z = (z, 1))}\Big]. \]
Finally, by \cref{fact:avg-minent}(b) with $r = 2$, applied to this distribution, $\Hmin(X^{(i)} | Z = (z, 1))$ is at least the left side minus $2$ with probability at least $1 - 2^{-2} = 3/4$ over $z \sim Z_{\mathrm{alice}}|_{B_Q=1}$.
\end{proof}

\subsubsection{Correlation between \texorpdfstring{$S_Q$}{S\_Q} and \texorpdfstring{$X$}{X}}

Finally, we show that $S_Q^{(i)}$ and $X^{(i)}$ are weakly correlated after conditioning on $Z = (z, 1)$. Merlin's message $U$ is long, but it is shared by all $m$ query sets, so it reveals little about the random query set $S_Q$. Since $B_Q$ is a function of $(Z_{\mathrm{alice}}, U)$, conditioning on it adds no correlation beyond what $U$ already carries. This is again a standard technique from \cite{ko_adaptive_2020} and \cite{ko_omega_2026}.

\begin{proposition}[Correlation on disjoint answers]\label{prop:key}
\[ \E_{z \sim Z_{\mathrm{alice}}|_{B_Q=1}}\Big[I\big(S_Q^{(i)}; X^{(i)} \bigm| Z = (z, 1)\big)\Big] \;\le\; \frac{n^{-2/5}}{\rho_1}. \]
\end{proposition}
\begin{proof}
We first bound the correlation before conditioning on $B_Q$. Since $S_Q^{(i)}$ is part of $S_Q$, the chain rule (\cref{fact:chain}) and \cref{clm:direct-sum} give
\[ I\big(S_Q^{(i)}; X^{(i)}, U \bigm| Z_{\mathrm{alice}}\big) \;\le\; I\big(S_Q; X^{(i)}, U \bigm| \Phi, \Aall, Q, S_{<Q}\big) \;=\; \frac1m\, I\big(\vS; X^{(i)}, U \bigm| \Phi, \Aall\big). \]
Since $\Aall$ is a function of $(\vS, \Phi)$, \cref{fact:drop-cond} lets us drop it from the conditioning:
\[ I\big(\vS; X^{(i)}, U \bigm| \Phi, \Aall\big) \;\le\; I\big(\vS; X^{(i)}, U \bigm| \Phi\big). \]
Recall that $\Phi = (X^{(\ne i)}, \Psi)$, where $\Psi$ is a function of $(U, R_{\mathrm{Bob}})$ (\cref{thm:too-good}(e)). Since $\vS$ is independent of $X$, and $R_{\mathrm{Bob}}$ is independent of $(\vS, X)$, the chain rule and \cref{fact:basic}(e) and (f) give
\[ I\big(\vS; X^{(i)}, U \bigm| X^{(\ne i)}, \Psi\big) \;\le\; I\big(\vS; X^{(i)}, U, R_{\mathrm{Bob}} \bigm| X^{(\ne i)}\big) \;=\; I\big(\vS; U \bigm| X^{(i)}, X^{(\ne i)}\big) \;\le\; H(U). \]
By \eqref{eq:HU}, we conclude that $I(S_Q^{(i)}; X^{(i)}, U | Z_{\mathrm{alice}}) \le H(U)/m \le n^{-2/5}$, since $m = n$. Next we condition on $B_Q$. Since $B_Q$ is a function of $(Z_{\mathrm{alice}}, U)$, the chain rule gives
\[ I\big(S_Q^{(i)}; X^{(i)} \bigm| Z_{\mathrm{alice}}, B_Q\big) \;\le\; I\big(S_Q^{(i)}; X^{(i)}, U, B_Q \bigm| Z_{\mathrm{alice}}\big) \;=\; I\big(S_Q^{(i)}; X^{(i)}, U \bigm| Z_{\mathrm{alice}}\big) \;\le\; n^{-2/5}. \]
Finally, by \cref{fact:basic}(b) and (c),
\begin{align*}
I\big(S_Q^{(i)}; X^{(i)} \bigm| Z_{\mathrm{alice}}, B_Q\big) \;&=\; \E_{(z,b) \sim (Z_{\mathrm{alice}}, B_Q)}\Big[I\big(S_Q^{(i)}; X^{(i)} \bigm| Z = (z, b)\big)\Big] \\
\;&\ge\; \rho_1 \cdot \E_{z \sim Z_{\mathrm{alice}}|_{B_Q=1}}\Big[I\big(S_Q^{(i)}; X^{(i)} \bigm| Z = (z, 1)\big)\Big]. \qedhere
\end{align*}
\end{proof}

\subsubsection{The main contradiction}

Combining the three bounds gives a monochromatic near-rectangle after conditioning $S_Q^{(i)}$ on an event of constant probability.

\begin{lemma}[A monochromatic near-rectangle]\label{lem:near-rect}
If $\rho_1 \ge n^{-1/6}$, then there are a value $(z, 1)$ of $Z$ and a set $E \subseteq \binom{B_i}{s_i}$ with $\Pr\big[S_Q^{(i)} \in E \bigm| Z = (z, 1)\big] \ge 3/4$ such that, writing $\cE$ for the event that $Z = (z, 1)$ and $S_Q^{(i)} \in E$,
\begin{enumerate}[label=(\roman*)]
  \item $\displaystyle \Hmin\big(S_Q^{(i)} \bigm| \cE\big) \;\ge\; \log\binom{n_i}{s_i} - \frac{\eps s_i}{2}$ \quad and \quad $\displaystyle \Hmin\big(X^{(i)} \bigm| \cE\big) \;\ge\; \log\binom{n_i}{s_i} - \frac{\eps s_i}{2}$, and
  \item $I\big(S_Q^{(i)}; X^{(i)} \bigm| \cE\big) \;\le\; n^{-1/5}$.
\end{enumerate}
In particular, the distribution of $(S_Q^{(i)}, X^{(i)})$ given $\cE$ is a monochromatic $(\eps s_i/2,\, n^{-1/5})$-near-rectangle on $\binom{B_i}{s_i} \times \binom{B_i}{s_i}$.
\end{lemma}
\begin{proof}
Draw $z \sim Z_{\mathrm{alice}}|_{B_Q=1}$. Each of the following fails with probability at most $1/4$:
\begin{enumerate}[label=(\alph*)]
  \item there is a set $E$ as in \cref{prop:S-entropy};
  \item $\Hmin(X^{(i)} | Z = (z, 1)) \ge \log\binom{n_i}{s_i} - \eps s_i/4 - \log(1/\rho_1) - 2$, by \cref{prop:X-leak};
  \item $I(S_Q^{(i)}; X^{(i)} | Z = (z, 1)) \le 4\, n^{-2/5}/\rho_1 \le 4\, n^{-7/30}$, by \cref{prop:key} and Markov's inequality.
\end{enumerate}
So some $z$ satisfies all three. Fix such $z$ and $E$, and let $\cE$ be the event that $Z = (z, 1)$ and $S_Q^{(i)} \in E$. Given $Z = (z, 1)$, the event $S_Q^{(i)} \in E$ is a function of $S_Q^{(i)}$ and has probability at least $3/4$, so \cref{fact:one-side} applies to it.
\begin{itemize}
  \item \emph{Min-entropy of $S_Q^{(i)}$.} By (a), $\Hmin(S_Q^{(i)} | \cE) \ge \log\binom{n_i}{s_i} - 32\, n^{1/5} \ge \log\binom{n_i}{s_i} - \eps s_i/2$, since $s_i \ge n^{1/3}$.
  \item \emph{Min-entropy of $X^{(i)}$.} By (b) and \cref{fact:one-side}, $\Hmin(X^{(i)} | \cE) \ge \log\binom{n_i}{s_i} - \eps s_i/4 - \log(1/\rho_1) - 3 \ge \log\binom{n_i}{s_i} - \eps s_i/2$, since $\log(1/\rho_1) \le (\log n)/6$.
  \item \emph{Correlation.} By (c) and \cref{fact:one-side}, $I(S_Q^{(i)}; X^{(i)} | \cE) \le \frac43 \cdot 4\, n^{-7/30} \le n^{-1/5}$.
\end{itemize}
Finally, by \cref{thm:too-good}(a), $S_Q \cap X = \emptyset$ whenever $B_Q = 1$, so the distribution of $(S_Q^{(i)}, X^{(i)})$ given $\cE$ is supported on disjoint pairs.
\end{proof}

\paragraph{Passing to a product distribution.}\phantomsection\label{sec:step3}
As in \cite{ko_omega_2026}, we replace this distribution by the product of its two marginals and apply \cref{lem:rect}.

\begin{corollary}[A nearly monochromatic rectangle]\label{cor:rect}
If $\rho_1 \ge n^{-1/6}$, let $\cE$ be the event of \cref{lem:near-rect}. Then
\[ R \;:=\; S_Q^{(i)}|_{\cE} \otimes X^{(i)}|_{\cE} \]
is an $(\eps s_i/2)$-rectangle on $\binom{B_i}{s_i} \times \binom{B_i}{s_i}$, and
\[ \Pr_{(S, X) \sim R}\big[S \cap X \ne \emptyset\big] \;\le\; n^{-1/10}. \]
\end{corollary}
\begin{proof}
By \cref{lem:near-rect}, the distribution of $(S_Q^{(i)}, X^{(i)})$ given $\cE$ is a monochromatic $(\eps s_i/2,\, n^{-1/5})$-near-rectangle, and its marginals are $S_Q^{(i)}|_{\cE}$ and $X^{(i)}|_{\cE}$. By \cref{lem:rect} with $\delta = 0$, $R$ is an $(\eps s_i/2)$-rectangle with
\[ \Pr_{(S, X) \sim R}\big[S \cap X \ne \emptyset\big] \;\le\; \sqrt{\frac{\ln 2}{2}\, n^{-1/5}} \;\le\; n^{-1/10}. \qedhere \]
\end{proof}

\paragraph{Applying the corruption bound.}\phantomsection\label{sec:step4}
A corruption bound shows that no such rectangle exists: two independent sparse sets with high min-entropy intersect with noticeable probability. For sets of size $\sqrt N$ in a universe of size $N$, this is the rectangle bound of Babai, Frankl and Simon \cite{babai_complexity_1986}. We need it for sparser sets, and prove it in \cref{sec:corruption}.

\begin{corollary}[High min-entropy rectangles are not monochromatic]\label{cor:product}
Let $\eps := 1/40$ and $c := 1/100$. Let $s \ge 200$, let $N$ be a multiple of $s^2$, let $\cS := \binom{[N]}{s}$, and let $\nu_A, \nu_B$ be distributions on $\cS$ with
\[ \Hmin(\nu_A),\, \Hmin(\nu_B) \;\ge\; \log\abs\cS - (\eps s - 1). \]
Then $\Pr_{a \sim \nu_A,\, b \sim \nu_B}[a \cap b \ne \emptyset] \ge c s^2/N$.
\end{corollary}

\noindent The following lemma completes the proof of the communication lower bound.

\begin{lemma}[Communication lower bound]\label{lem:comm-lb}
For every $i \in \cI$, the protocol $G_i$ of \cref{thm:too-good} has $\rho_1 < n^{-1/6}$.
\end{lemma}
\begin{proof}
Suppose, for contradiction, that $\rho_1 \ge n^{-1/6}$. By \cref{cor:rect}, $R$ is an $(\eps s_i/2)$-rectangle on $\binom{B_i}{s_i} \times \binom{B_i}{s_i}$ that puts mass at most $n^{-1/10}$ on intersecting pairs. On the other hand, apply \cref{cor:product} with $N := n_i$ and $s := s_i$, identifying $B_i$ with $[n_i]$; this is allowed since $n_i = \ell s_i^2$, $s_i \ge n^{1/3} \ge 200$ and $\eps s_i/2 \le \eps s_i - 1$. So $R$ puts mass at least $c s_i^2/n_i = c/\ell \ge c/\log n > n^{-1/10}$ on intersecting pairs, a contradiction.
\end{proof}

\begin{proofof}{\cref{thm:main}}
Recall that $\ell = \Theta(\log n/\log(t_u w))$ by \cref{clm:params}(a). Let $c_0:=1/192$, and suppose, for contradiction, that $t_q \le c_0\ell^2$. Since $t_u w \le n^{1/95}$, \cref{clm:params}(a) gives $\ell\ge20$, and hence $\abs{\cI}\ge(\ell-4)/3\ge\ell/4$. Hence, using $\sum_{i \in \cI}\tau_i\le \E_{\cD,Q}[T(Q)]\le t_q$, some $i\in\cI$ satisfies $\tau_i\le4c_0\ell$. As in \cref{cor:too-good}, $\log(1/p)\le2\log(t_u w)$, and therefore
\[
 p^{8\tau_i} \;\ge\; 2^{-16\tau_i\log(t_u w)}
 \;\ge\; 2^{-64c_0\ell\log(t_u w)}
 \;\ge\; n^{-16c_0}=n^{-1/12},
\]
where the last inequality uses $4\ell\log(t_u w)=\ell\log\beta\le\log n$. By \cref{thm:too-good}(d),
\[
\rho_1 \;\ge\; \frac18 n^{-1/12}-e^{-n^{1/4}/3} \;>\; n^{-1/6},
\]
contradicting \cref{lem:comm-lb}. Thus $t_q>c_0\ell^2$, and \cref{clm:params}(a) gives the claimed $\Omega((\log n/\log(t_u w))^2)$ bound. For the bound on $\max(t_u,t_q)$, if $t_u\ge\log^2 n$ there is nothing to prove; otherwise $\log(t_u w)=O(\log\log n)$.
\end{proofof}

\section{A corruption bound for sparse sets}\label{sec:corruption}

We prove a combinatorial bound: two independent sparse sets, each drawn from a sufficiently large family, intersect with noticeable probability. The case $N=s^2$ is the rectangle bound of Babai, Frankl and Simon \cite{babai_complexity_1986}. We extend their argument to $s$-sets in a universe of size $N$, where $N$ is a multiple of $s^2$.

Write $\cS:=\binom{[N]}s$, let $\cA,\cB\subseteq\cS$, and let $\cU_{\cA}$ and $\cU_{\cB}$ be the uniform distributions on these families. We write $\cU$ for the uniform distribution on $\cS$.

\begin{lemma}[Corruption for sparse sets]\label{lem:corruption}
Let $\eps := 1/40$ and $c := 1/100$. Let $s \ge 200$, and let $N$ be a multiple of $s^2$. If $\cA, \cB \subseteq \cS$ satisfy $\abs \cA, \abs \cB \ge 2^{-\eps s}\abs\cS$, then
\begin{equation*}
    \Pr_{\substack{ a \sim\cU_{\cA} \\ b \sim \cU_{\cB}}}[a \cap b \ne \emptyset] \;\ge\; c\cdot\frac{s^2}{N}
\end{equation*}
\end{lemma}

First we derive the min-entropy form used in Step~4 (\cref{cor:product}) by applying the flat decomposition in \cref{cor:flat-rect}. We restate the corollary for convenience.

\begin{repcorollary}{cor:product}
Let $s \ge 200$, let $N$ be a multiple of $s^2$, and let $\nu_A, \nu_B$ be distributions on $\cS$ with
\[ \Hmin(\nu_A),\, \Hmin(\nu_B) \;\ge\; \log\abs\cS - (\eps s - 1). \]
Then $\Pr_{a \sim \nu_A,\, b \sim \nu_B}[a \cap b \ne \emptyset] \ge c s^2/N$.
\end{repcorollary}
\begin{proof}
The product $\nu_A \otimes \nu_B$ is an $(\eps s - 1)$-rectangle. Since $\abs\cS \ge 2^{\eps s}$, some integer $K$ satisfies $2^{-\eps s}\abs\cS \le K \le 2^{1 - \eps s}\abs\cS$. By \cref{cor:flat-rect} with this $K$, the product is a convex combination of uniform distributions on combinatorial rectangles $\cA \times \cB$ with $\abs\cA = \abs\cB = K \ge 2^{-\eps s}\abs\cS$. By \cref{lem:corruption}, each of them puts mass at least $c s^2/N$ on intersecting pairs. Linearity gives the same bound for $\nu_A \otimes \nu_B$.
\end{proof}

The one-sided formulation is essential: an analogous lower bound on the probability of disjointness is false. If $\log(N/s) \le \eps s$, as in our application, then the family $\cA$ of sets that contain a fixed element has $\abs\cA = (s/N)\abs\cS \ge 2^{-\eps s}\abs\cS$, but all pairs in $\cA \times \cA$ intersect. This is why the main proof uses only runs whose answer is ``disjoint''.

\subsection{Counting facts}\label{sec:bfs-tools}

We use the following counting estimates, as in \cite{babai_complexity_1986}. Write $H_{\mathrm b}(\alpha) := -\alpha\log\alpha - (1-\alpha)\log(1-\alpha)$ for the binary entropy function.

\begin{fact}[Binomial estimates]\label{fact:binom}
For all integers $u \ge 0$ and $k \ge 1$, $\binom uk \le \big(\frac{eu}{k}\big)^k$. For every integer $g \ge 1$ and every $0 < \alpha \le 1/2$,
\[ \sum_{0 \le h \le \alpha g}\binom gh \;\le\; 2^{H_{\mathrm b}(\alpha)g}. \]
\end{fact}

\noindent For $a \sim \cU$, the probability that $a$ contains a fixed set or avoids a fixed set is expressed by binomial coefficients. For $T, V \subseteq [N]$ with $\abs T = k \le s$,
\begin{align*}
\Pr_{a \sim \cU}[T \subseteq a] \;&=\; \binom{N-k}{s-k}\Big/\binom Ns \;=\; \prod_{h=0}^{k-1}\frac{s-h}{N-h} \;\le\; \Big(\frac sN\Big)^k, \\
\Pr_{a \sim \cU}[a \cap V = \emptyset] \;&=\; \binom{N - \abs V}{s}\Big/\binom Ns \;=\; \prod_{h=0}^{s-1}\frac{N - \abs V - h}{N-h} \;\le\; \Big(1 - \frac{\abs V}{N}\Big)^s,
\end{align*}
and $\Pr[T \subseteq a] = 0$ if $\abs T > s$. We use the following two consequences.

\begin{claim}[Two counting facts]\label{clm:counting}
Let $1 \le s \le N$.
\begin{enumerate}[label=(\alph*)]
  \item For every $P \subseteq [N]$ and every integer $k \ge 1$, $\Pr_{a \sim \cU} \big[\abs{a \cap P} \ge k\big] \le \binom{\abs P}{k}\big(\frac sN\big)^k \le \big(\frac{e\abs{P}s}{kN}\big)^k$.
  \item For every $V \subseteq [N]$, $\Pr_{a \sim \cU} [a \cap V = \emptyset] \le \big(1 - \frac{\abs V}{N}\big)^s \le e^{-s\abs{V}/N}$.
\end{enumerate}
\end{claim}
\begin{proof}
(a) If $\abs{a \cap P} \ge k$, then $a$ contains some $T \subseteq P$ with $\abs T = k$. By a union bound over the $\binom{\abs P}{k}$ choices of $T$, the bound $\Pr[T \subseteq a] \le (s/N)^k$, and \cref{fact:binom},
\begin{align*}
\Pr_{a \sim \cU}\big[\abs{a \cap P} \ge k\big]
\;&\le\; \Pr_{a \sim \cU}\big[\exists\, T \subseteq P \text{ with } \abs T = k \text{ and } T \subseteq a\big]
\;\le\; \sum_{\substack{T \subseteq P \\ \abs T = k}} \Pr_{a \sim \cU}[T \subseteq a] \\
\;&\le\; \binom{\abs P}{k}\Big(\frac sN\Big)^k \;\le\; \Big(\frac{e\abs P s}{kN}\Big)^k.
\end{align*}
(b) If $N-\abs V<s$, the avoidance probability is zero and the bound follows. Otherwise, all factors below are nonnegative and satisfy $\frac{N - \abs V - h}{N-h} \le 1 - \frac{\abs V}{N}$. Using $1-x\le e^{-x}$,
\begin{align*}
\Pr_{a \sim \cU}[a \cap V = \emptyset]
\;&=\; \binom{N - \abs V}{s}\Big/\binom Ns
\;=\; \prod_{h=0}^{s-1}\frac{N - \abs V - h}{N-h} \\
\;&\le\; \Big(1 - \frac{\abs V}{N}\Big)^s \;\le\; e^{-s\abs V/N}.
\end{align*}

\end{proof}

\subsection{Proof of the corruption bound}\label{sec:bfs}

For $N = s^2$, \cref{lem:corruption} is the rectangle bound behind the $\Omega(\sqrt N)$ lower bound for Disjointness of \cite{babai_complexity_1986} (see also \cite[Section~3.2, proof of Theorem~6]{chattopadhyay_story_2010}). We first recall their argument before explaining the changes needed for sparser sets. Suppose that $\cA \times \cB$ is large and nearly disjoint.
\begin{itemize}[leftmargin=*]
  \item \emph{A disjoint chain in $\cA$:} Since $\cA$ is large, it contains rows $a_1, \dots, a_t$ with $t = \Theta(s)$, each of which has at least $\lceil s/2\rceil$ elements outside the earlier rows. These pieces of new elements are pairwise disjoint and together cover a constant fraction of $[N]$.
  \item \emph{Columns of $\cB$ avoid the chain:} A typical row meets few columns of $\cB$, so a typical column of $\cB$ meets few of the pieces. That is, it lies in the complement of all but a few of the pieces.
  \item \emph{Few sets avoid the chain:} A uniform $s$-set meets each piece with constant probability. A union bound over which pieces it meets shows that only a $2^{-\Omega(s)}$ fraction of all $s$-sets meet so few pieces. Hence $\cB$ is small, a contradiction.
\end{itemize}
For sparser sets, i.e.\ when $N$ is a multiple of $s^2$, we make two changes.
\begin{itemize}[leftmargin=*]
  \item \emph{A longer chain:} We take $t=(N/s^2)\lfloor s/6\rfloor=\Theta(N/s)$ rows. The $t$ rows cover at most $N/6$ elements, so another row with many new elements can still be found, while the disjoint pieces cover a constant fraction of the universe.
  \item \emph{Grouping the pieces:} We group the $t$ pieces into $g=\lfloor s/6\rfloor$ groups of $N/s^2$ pieces each, and apply the union bound to groups instead of pieces. Each group has $\Theta(N/s)$ elements, so a uniform $s$-set meets it with constant probability, just as it meets a single piece when $N=s^2$. For $N = s^2$, every group is a single piece.
\end{itemize}

\begin{proofof}{\cref{lem:corruption}}
For brevity set $\delta := c s^2/N$. Recall that $\abs\cA, \abs\cB \ge 2^{-\eps s}\abs\cS$. Suppose, for contradiction, that
\[ \Pr_{a \sim \cU_{\cA},\, b \sim \cU_{\cB}}[a \cap b \ne \emptyset] \;<\; \delta. \]
We show that $\abs \cB < 2^{-\eps s}\abs\cS$.

\emph{Good rows.} Let $\cA' := \{a \in \cA : \Pr_{b \sim \cU_{\cB}}[a \cap b \ne \emptyset] \le 2\delta\}$. By Markov's inequality, $\abs{\cA'} \ge \abs \cA/2 \ge 2^{-\eps s - 1}\abs\cS$.

\emph{A chain of rows with new elements.} Put $r:=N/s^2$, $k:=\lceil s/2\rceil$, $g:=\lfloor s/6\rfloor$, and $t:=rg$. These are positive integers. We select $t$ rows in $\cA'$, each of which brings at least $k$ new elements.

\begin{claim}[A chain of rows with new elements]\label{clm:spread}
There are $a_1, \dots, a_t \in \cA'$ and sets $D_1, \dots, D_t$ such that, for every $j \in [t]$,
\[ D_j \;\subseteq\; a_j \setminus (a_1 \cup \dots \cup a_{j-1}) \qquad\text{and}\qquad \abs{D_j} = k. \]
In particular, $D_1, \dots, D_t$ are pairwise disjoint.
\end{claim}
\begin{proof}
We choose $a_1, a_2, \dots$ greedily. Let $1 \le j \le t$, suppose that $a_1, \dots, a_{j-1}$ have been chosen, and let $P_j := a_1 \cup \dots \cup a_{j-1}$ (so $P_1 = \emptyset$). It suffices to find $a_j\in\cA'$ with $\abs{a_j\cap P_j}<k$. Then $\abs{a_j\setminus P_j}\ge s-k+1\ge k$, so we can choose $D_j$ to contain $k$ new elements.

Each earlier row has $s$ elements, so $\abs{P_j}\le(j-1)s<ts\le N/6$. By \cref{clm:counting}(a), using $k\ge s/2$,
\[
\Pr_{a\sim\cU}[\abs{a\cap P_j}\ge k]
\le\left(\frac{e\abs{P_j}s}{kN}\right)^k
\le(e/3)^k\le(e/3)^{s/2}\le2^{-0.071s}.
\]
Since $\abs{\cA'}\ge2^{-\eps s-1}\abs\cS>2^{-0.071s}\abs\cS$, some $a_j\in\cA'$ has $\abs{a_j\cap P_j}<k$. Finally, $D_j \subseteq a_j \setminus P_j$ is disjoint from $D_1, \dots, D_{j-1} \subseteq P_j$.
\end{proof}

\emph{Grouping the pieces.} Partition $[t]$ into $g$ sets $\Gamma_1,\dots,\Gamma_g$, each of size $r=N/s^2$. For $u \in [g]$, let
\[ E_u \;:=\; \bigcup_{j \in \Gamma_u} D_j. \]
The sets $E_1, \dots, E_g$ are pairwise disjoint, and each has $rk\ge N/(2s)$ elements.

\emph{Most columns miss most of the groups.} Since every $a_j$ lies in $\cA'$,
\[ \E_{b \sim \cU_{\cB}}\big[\,\abs{\{j \in [t] : a_j \cap b \ne \emptyset\}}\,\big] \;=\; \sum_{j=1}^{t}\Pr_{b \sim \cU_{\cB}}[a_j \cap b \ne \emptyset] \;\le\; 2\delta t, \]
so by Markov's inequality the set $\cB'$ of $b \in \cB$ that meet at most $4\delta t$ of the rows $a_j$ satisfies $\abs{\cB'} \ge \abs \cB/2$. If $b$ meets $E_u$, then it meets $a_j \supseteq D_j$ for some $j \in \Gamma_u$, and the sets $\Gamma_u$ are disjoint. So every $b \in \cB'$ meets at most $4\delta t$ of the sets $E_u$, and by our choice of parameters
\[4\delta t=4\cdot\frac{cs^2}{N}\cdot\frac{N}{s^2}g=4cg=\frac{g}{25}.\]

\emph{Counting $\cB'$.} If $b\in\cB'$ meets $h\le g/25$ groups, it avoids the union $V$ of the other $g-h$ groups. Since $\abs V=(g-h)rk\ge(g-h)N/(2s)$,
\[
\frac{s\abs V}{N}\ge\frac{g-h}{2}\ge\frac{12g}{25}.
\]
For each choice of the groups met, \cref{clm:counting}(b) bounds the number of such sets by $e^{-12g/25}\abs\cS$. A union bound and \cref{fact:binom} give
\begin{align*}
\abs{\cB'}
&\le \sum_{0\le h\le g/25}\binom gh e^{-12g/25}\abs\cS\\
&\le 2^{H_{\mathrm b}(1/25)g}e^{-12g/25}\abs\cS
\le2^{-0.45g}\abs\cS,
\end{align*}
because $H_{\mathrm b}(1/25)\le0.2423$ and $12/(25\ln2)>0.69249$. As $g\ge s/6-1$ and $s\ge200$,
\[
\abs\cB\le2\abs{\cB'}\le2^{-0.075s+1.45}\abs\cS
<2^{-s/40}\abs\cS,
\]
contradicting the assumed size of $\cB$.
\end{proofof}

\section*{AI methodology}
Claude Opus 5.5 assisted in developing the proof of \cref{lem:corruption}, following the author's proposal to adapt the argument of Babai, Frankl and Simon \cite{babai_complexity_1986} to sparser sets. The modifications are a longer chain of rows and a union bound over groups of pieces. AI tools also assisted with the adaptations and exposition of the published reductions in \cref{sec:graphs}. The applications other than those originally in \cite{patrascu_mihai_towards_2010} were found with the assistance of ChatGPT (GPT-6.0 Astra). For applications based on published reductions, the tool served as an aid to searching the literature. AI tools assisted with application writeups, independent logical checks, and editorial revisions. The author is responsible for the final manuscript.

\bibliographystyle{alpha}
\bibliography{references}

\newcommand{\etalchar}[1]{$^{#1}$}
\begin{thebibliography}{HKM{\etalchar{+}}12}

\bibitem[AD16]{abboud_popular_2016}
Amir Abboud and S{\o}ren Dahlgaard.
\newblock Popular {Conjectures} as a {Barrier} for {Dynamic} {Planar} {Graph}
  {Algorithms}.
\newblock In {\em 2016 {IEEE} 57th {Annual} {Symposium} on {Foundations} of
  {Computer} {Science} ({FOCS})}, pages 477--486, New Brunswick, NJ, USA,
  October 2016. IEEE.
\newblock Full version:
  \href{https://arxiv.org/abs/1605.03797}{arXiv:1605.03797}.

\bibitem[AW14]{abboud_popular_2014}
Amir Abboud and Virginia~Vassilevska Williams.
\newblock Popular {Conjectures} {Imply} {Strong} {Lower} {Bounds} for {Dynamic}
  {Problems}.
\newblock In {\em 2014 {IEEE} 55th {Annual} {Symposium} on {Foundations} of
  {Computer} {Science}}, pages 434--443, October 2014.
\newblock Full version:
  \href{https://arxiv.org/abs/1402.0054}{arXiv:1402.0054}.

\bibitem[AW26]{alman_truly_2026}
Josh Alman and Virginia~Vassilevska Williams.
\newblock Truly {Subquadratic} {3SUM} and {Truly} {Subcubic} {APSP} via
  {Triangles} in {Sparse} {Lopsided} {Graphs}, October 2026.
\newblock arXiv:2610.06783 [cs.DS].

\bibitem[BFS86]{babai_complexity_1986}
L{\'a}szl{\'o} Babai, P{\'e}ter Frankl, and J{\'a}nos Simon.
\newblock Complexity classes in communication complexity theory (preliminary
  version).
\newblock In {\em 27th {Annual} {Symposium} on {Foundations} of {Computer}
  {Science}}, pages 337--347, 1986.

\bibitem[BL15]{brody_adapt_2015}
Joshua Brody and Kasper~Green Larsen.
\newblock Adapt or {Die}: {Polynomial} {Lower} {Bounds} for {Non}-{Adaptive}
  {Dynamic} {Data} {Structures}.
\newblock {\em Theory of Computing}, 11(19):471--489, 2015.

\bibitem[Blu85]{blum_single-operation_1985}
Norbert Blum.
\newblock On the {Single}-{Operation} {Worst}-{Case} {Time} {Complexity} on the
  {Disjoint} {Set} {Union} {Problem}.
\newblock In {\em {STACS} 85, 2nd {Symposium} of {Theoretical} {Aspects} of
  {Computer} {Science}}, pages 32--38, 1985.

\bibitem[BT10]{bille_regular_2010}
Philip Bille and Mikkel Thorup.
\newblock Regular {Expression} {Matching} with {Multi}-{Strings} and
  {Intervals}.
\newblock In {\em Proceedings of the {Twenty}-{First} {Annual} {ACM}-{SIAM}
  {Symposium} on {Discrete} {Algorithms}}, pages 1297--1308. Society for
  Industrial and Applied Mathematics, January 2010.

\bibitem[CDK{\etalchar{+}}11]{cohen_efficient_2011}
Edith Cohen, Nick Duffield, Haim Kaplan, Carsten Lund, and Mikkel Thorup.
\newblock Efficient {Stream} {Sampling} for {Variance}-{Optimal} {Estimation}
  of {Subset} {Sums}.
\newblock {\em SIAM Journal on Computing}, 40(5):1402--1431, January 2011.

\bibitem[CDL{\etalchar{+}}14]{chan_linear-space_2014}
Timothy~M. Chan, Stephane Durocher, Kasper~Green Larsen, Jason Morrison, and
  Bryan~T. Wilkinson.
\newblock Linear-{Space} {Data} {Structures} for {Range} {Mode} {Query} in
  {Arrays}.
\newblock {\em Theory of Computing Systems}, 55(4):719--741, November 2014.

\bibitem[CEEP16]{chattopadhyay_upper_2016}
Arkadev Chattopadhyay, Jeff Edmonds, Faith Ellen, and Toniann Pitassi.
\newblock Upper and {Lower} {Bounds} on the {Power} of {Advice}.
\newblock {\em SIAM Journal on Computing}, 45(4):1412--1432, January 2016.

\bibitem[CGL15]{clifford_new_2015}
Rapha{\"e}l Clifford, Allan Gr{\o}nlund, and Kasper~Green Larsen.
\newblock New {Unconditional} {Hardness} {Results} for {Dynamic} and {Online}
  {Problems}.
\newblock In {\em 2015 {IEEE} 56th {Annual} {Symposium} on {Foundations} of
  {Computer} {Science}}, pages 1089--1107, October 2015.

\bibitem[CGLS18]{clifford_upper_2018}
Rapha{\"e}l Clifford, Allan Gr{\o}nlund, Kasper~Green Larsen, and Tatiana
  Starikovskaya.
\newblock {Upper and Lower Bounds for Dynamic Data Structures on Strings}.
\newblock In {\em 35th Symposium on Theoretical Aspects of Computer Science
  (STACS 2018)}, volume~96 of {\em LIPIcs}, pages 22:1--22:14, 2018.

\bibitem[CKL{\etalchar{+}}24]{chen_almost-linear_2024}
Li~Chen, Rasmus Kyng, Yang~P. Liu, Simon Meierhans, and Maximilian
  Probst~Gutenberg.
\newblock Almost-{Linear} {Time} {Algorithms} for {Incremental} {Graphs}:
  {Cycle} {Detection}, {SCCs}, s-t {Shortest} {Path}, and {Minimum}-{Cost}
  {Flow}.
\newblock In {\em Proceedings of the 56th {Annual} {ACM} {Symposium} on
  {Theory} of {Computing}}, pages 1165--1173, Vancouver BC Canada, June 2024.
  ACM.

\bibitem[CKM{\etalchar{+}}19]{charalampopoulos_internal_2019}
Panagiotis Charalampopoulos, Tomasz Kociumaka, Manal Mohamed, Jakub
  Radoszewski, Wojciech Rytter, and Tomasz Wale{\'n}.
\newblock {Internal Dictionary Matching}.
\newblock In {\em 30th International Symposium on Algorithms and Computation
  (ISAAC 2019)}, volume 149 of {\em LIPIcs}, pages 22:1--22:17, 2019.

\bibitem[CP10]{chattopadhyay_story_2010}
Arkadev Chattopadhyay and Toniann Pitassi.
\newblock The story of set disjointness.
\newblock {\em ACM SIGACT News}, 41(3):59--85, September 2010.

\bibitem[CT06]{cover_elements_2006}
Thomas~M. Cover and Joy~A. Thomas.
\newblock {\em Elements of {Information} {Theory} ({Wiley} {Series} in
  {Telecommunications} and {Signal} {Processing})}.
\newblock Wiley-Interscience, New York, NY, USA, 2006.

\bibitem[DGGS22]{naor_streaming_2022}
Bartłomiej Dudek, Paweł Gawrychowski, Garance Gourdel, and Tatiana
  Starikovskaya.
\newblock Streaming {Regular} {Expression} {Membership} and {Pattern}
  {Matching}.
\newblock In Joseph~(Seffi) Naor and Niv Buchbinder, editors, {\em Proceedings
  of the 2022 {Annual} {ACM}-{SIAM} {Symposium} on {Discrete} {Algorithms}
  ({SODA})}, pages 670--694. Society for Industrial and Applied Mathematics,
  Philadelphia, PA, January 2022.

\bibitem[DI22]{dallant_conditional_2022}
Justin Dallant and John Iacono.
\newblock {Conditional Lower Bounds for Dynamic Geometric Measure Problems}.
\newblock In {\em 30th Annual European Symposium on Algorithms (ESA 2022)},
  volume 244 of {\em LIPIcs}, pages 39:1--39:17, 2022.
\newblock Full version:
  \href{https://arxiv.org/abs/2112.10095}{arXiv:2112.10095}.

\bibitem[DL20]{dvorak_lower_2020}
Pavel Dvo{\v{r}}{\'a}k and Bruno Loff.
\newblock Lower {Bounds} for {Semi}-adaptive {Data} {Structures} via
  {Corruption}.
\newblock In {\em 40th IARCS Annual Conference on Foundations of Software
  Technology and Theoretical Computer Science (FSTTCS 2020)}, volume 182 of
  {\em LIPIcs}, pages 20:1--20:15, 2020.

\bibitem[DORS08]{dodis_fuzzy_2008}
Yevgeniy Dodis, Rafail Ostrovsky, Leonid Reyzin, and Adam Smith.
\newblock Fuzzy {Extractors}: {How} to {Generate} {Strong} {Keys} from
  {Biometrics} and {Other} {Noisy} {Data}.
\newblock {\em SIAM Journal on Computing}, 38(1):97--139, January 2008.

\bibitem[FS89]{fredman_cell_1989}
Michael~L. Fredman and Michael~E. Saks.
\newblock The {Cell} {Probe} {Complexity} of {Dynamic} {Data} {Structures}.
\newblock In {\em Proceedings of the 21st {Annual} {ACM} {Symposium} on
  {Theory} of {Computing}}, pages 345--354, 1989.

\bibitem[GKKP20]{golan_streaming_2020}
Shay Golan, Tomasz Kociumaka, Tsvi Kopelowitz, and Ely Porat.
\newblock {The Streaming $k$-Mismatch Problem: Tradeoffs Between Space and
  Total Time}.
\newblock In {\em 31st Annual Symposium on Combinatorial Pattern Matching (CPM
  2020)}, volume 161 of {\em LIPIcs}, pages 15:1--15:15, 2020.

\bibitem[GU18]{gawrychowski_towards_2018}
Pawe{\l} Gawrychowski and Przemys{\l}aw Uzna{\'n}ski.
\newblock {Towards Unified Approximate Pattern Matching for Hamming and $L_1$
  Distance}.
\newblock In {\em 45th International Colloquium on Automata, Languages, and
  Programming (ICALP 2018)}, volume 107 of {\em LIPIcs}, pages 62:1--62:13,
  2018.

\bibitem[HF98]{henzinger_lower_1998}
Monika~R. Henzinger and Michael~L. Fredman.
\newblock Lower {Bounds} for {Fully} {Dynamic} {Connectivity} {Problems} in
  {Graphs}.
\newblock {\em Algorithmica}, 22(3):351--362, November 1998.

\bibitem[HKM{\etalchar{+}}12]{haeupler_incremental_2012}
Bernhard Haeupler, Telikepalli Kavitha, Rogers Mathew, Siddhartha Sen, and
  Robert~E. Tarjan.
\newblock Incremental {Cycle} {Detection}, {Topological} {Ordering}, and
  {Strong} {Component} {Maintenance}.
\newblock {\em ACM Transactions on Algorithms}, 8(1):1--33, January 2012.

\bibitem[HKNS15]{henzinger_unifying_2015}
Monika Henzinger, Sebastian Krinninger, Danupon Nanongkai, and Thatchaphol
  Saranurak.
\newblock Unifying and {Strengthening} {Hardness} for {Dynamic} {Problems} via
  the {Online} {Matrix}-{Vector} {Multiplication} {Conjecture}.
\newblock In {\em Proceedings of the {Forty}-seventh {Annual} {ACM} {Symposium}
  on {Theory} of {Computing}}, {STOC} '15, pages 21--30, New York, NY, USA,
  2015. ACM.
\newblock Full version:
  \href{https://arxiv.org/abs/1511.06773}{arXiv:1511.06773}.

\bibitem[IPZ01]{impagliazzo_which_2001}
Russell Impagliazzo, Ramamohan Paturi, and Francis Zane.
\newblock Which {Problems} {Have} {Strongly} {Exponential} {Complexity}?
\newblock {\em Journal of Computer and System Sciences}, 63(4):512--530,
  December 2001.

\bibitem[JRS09]{jain_property_2009}
Rahul Jain, Jaikumar Radhakrishnan, and Pranab Sen.
\newblock A property of quantum relative entropy with an application to privacy
  in quantum communication.
\newblock {\em Journal of the ACM}, 56(6):1--32, September 2009.

\bibitem[JX22]{jin_tight_2022}
Ce~Jin and Yinzhan Xu.
\newblock Tight dynamic problem lower bounds from generalized {BMM} and {OMv}.
\newblock In {\em Proceedings of the 54th {Annual} {ACM} {SIGACT} {Symposium}
  on {Theory} of {Computing}}, pages 1515--1528, Rome Italy, June 2022. ACM.
\newblock Full version:
  \href{https://arxiv.org/abs/2202.11250}{arXiv:2202.11250}.

\bibitem[KLMS26]{koucky_natural_2026}
Michal Koucký, Bruno Loff, Tulasimohan Molli, and Michael~E. Saks.
\newblock The {Natural} {Proofs} {Barrier} against {Data}-{Structure}
  {Lower}-{Bounds}.
\newblock In {\em Proceedings of the 58th {Annual} {ACM} {Symposium} on
  {Theory} of {Computing}}, pages 1338--1343, Salt Lake City UT USA, June 2026.
  ACM.

\bibitem[Ko26]{ko_omega_2026}
Young~Kun Ko.
\newblock {An $\Omega((\log n/\log\log n)^2)$ Cell-Probe Lower Bound for
  Dynamic Boolean Data Structures}.
\newblock In {\em Proceedings of the 67th {IEEE} {Symposium} on {Foundations}
  of {Computer} {Science} ({FOCS}) (to appear)}. IEEE, 2026.
\newblock Full version:
  \href{https://arxiv.org/abs/2603.25914v2}{arXiv:2603.25914v2}.

\bibitem[KPP15]{kopelowitz_dynamic_2015}
Tsvi Kopelowitz, Seth Pettie, and Ely Porat.
\newblock Dynamic set intersection.
\newblock In {\em Workshop on {Algorithms} and {Data} {Structures}}, pages
  470--481. Springer, 2015.

\bibitem[KW20]{ko_adaptive_2020}
Young~Kun Ko and Omri Weinstein.
\newblock An {Adaptive} {Step} {Toward} the {Multiphase} {Conjecture}.
\newblock In {\em 2020 {IEEE} 61st {Annual} {Symposium} on {Foundations} of
  {Computer} {Science} ({FOCS})}, pages 752--761, Durham, NC, USA, November
  2020. IEEE.

\bibitem[Lar12]{larsen_cell_2012}
Kasper~Green Larsen.
\newblock The {Cell} {Probe} {Complexity} of {Dynamic} {Range} {Counting}.
\newblock In {\em Proceedings of the {Forty}-fourth {Annual} {ACM} {Symposium}
  on {Theory} of {Computing}}, {STOC} '12, pages 85--94, New York, NY, USA,
  2012. ACM.

\bibitem[LWY20]{larsen_crossing_2020}
Kasper~Green Larsen, Omri Weinstein, and Huacheng Yu.
\newblock Crossing the {Logarithmic} {Barrier} for {Dynamic} {Boolean} {Data}
  {Structure} {Lower} {Bounds}.
\newblock {\em SIAM Journal on Computing}, 49(5):STOC18--323--STOC18--367,
  January 2020.

\bibitem[LY25]{larsen_super-logarithmic_2025}
Kasper~Green Larsen and Huacheng Yu.
\newblock Super-{Logarithmic} {Lower} {Bounds} for {Dynamic} {Graph}
  {Problems}.
\newblock {\em SIAM Journal on Computing}, pages FOCS23--42--FOCS23--69,
  February 2025.

\bibitem[MSVT94]{miltersen_complexity_1994}
Peter~Bro Miltersen, Sairam Subramanian, Jeffrey~Scott Vitter, and Roberto
  Tamassia.
\newblock Complexity models for incremental computation.
\newblock {\em Theoretical Computer Science}, 130(1):203--236, August 1994.

\bibitem[P{\u{a}}t08]{patrascu_unifying_2008}
Mihai P{\u{a}}tra{\c{s}}cu.
\newblock Unifying the {Landscape} of {Cell}-{Probe} {Lower} {Bounds}.
\newblock In {\em {FOCS} 2008}, pages 434--443, 2008.

\bibitem[P{\u{a}}t10]{patrascu_mihai_towards_2010}
Mihai P{\u{a}}tra{\c{s}}cu.
\newblock Towards polynomial lower bounds for dynamic problems.
\newblock In Leonard~J. Schulman, editor, {\em Proceedings of the 42nd {ACM}
  {Symposium} on {Theory} of {Computing}, {STOC} 2010, {Cambridge},
  {Massachusetts}, {USA}, 5-8 {June} 2010}, pages 603--610. ACM, 2010.

\bibitem[PD06]{patrascu_logarithmic_2006}
Mihai P{\u{a}}tra{\c{s}}cu and Erik~D. Demaine.
\newblock Logarithmic {Lower} {Bounds} in the {Cell}-{Probe} {Model}.
\newblock {\em SIAM J. Comput.}, 35(4):932--963, April 2006.

\bibitem[Tar75]{tarjan_efficiency_1975}
Robert~Endre Tarjan.
\newblock Efficiency of a {Good} {But} {Not} {Linear} {Set} {Union}
  {Algorithm}.
\newblock {\em J. ACM}, 22(2):215--225, 1975.

\bibitem[Vad12]{vadhan_pseudorandomness_2012}
Salil~P. Vadhan.
\newblock Pseudorandomness.
\newblock {\em Foundations and Trends® in Theoretical Computer Science},
  7(1–3):1--336, 2012.

\bibitem[vdBNS19]{van_den_brand_dynamic_2019}
Jan van~den Brand, Danupon Nanongkai, and Thatchaphol Saranurak.
\newblock Dynamic {Matrix} {Inverse}: {Improved} {Algorithms} and {Matching}
  {Conditional} {Lower} {Bounds}.
\newblock In {\em 2019 {IEEE} 60th {Annual} {Symposium} on {Foundations} of
  {Computer} {Science} ({FOCS})}, pages 456--480, Baltimore, MD, USA, November
  2019. IEEE.
\newblock Full version:
  \href{https://arxiv.org/abs/1905.05067}{arXiv:1905.05067}.

\bibitem[Yao81]{yao_should_1981}
Andrew Chi-Chih Yao.
\newblock Should {Tables} {Be} {Sorted}?
\newblock {\em J. ACM}, 28(3):615--628, July 1981.

\end{thebibliography}

\clearpage
\appendix
\crefalias{section}{appendix}

\section{Reductions to other problems}\label{sec:graphs}

For completeness and reader convenience, we collect the reductions yielding the applications of \cref{thm:main}. Each construction identifies its published source and any adaptations. \Cref{lem:app-transfer} gives the common simulation, \cref{tab:appendix-transfers} indexes the bounds, and \cref{sec:streaming-regex} proves the streaming separation.

\paragraph{Conventions.}
Let $S_1,\dots,S_m\subseteq[n]$ be the preprocessed sets, $X\subseteq[n]$ the evolving set, and $q\in[m]$ the final query. Sets are indexed by $i$ and elements by $j$. We write $M_{ij}:=\mathbf 1_{j\in S_i}$, $x_j:=\mathbf 1_{j\in X}$ and $c_i:=\abs{S_i\cap X}$, so that $S_q\cap X=\emptyset$ exactly when $c_q=0$, and we take $m=n$. The data structures we build keep the characteristic vector of $X$ and ignore repeated insertions, so every inserted element is new. The initial target instance is built during Phase~I, which is not charged. Every target instance has size $N=n^{\Theta(1)}$; padding with parts that no query sees, such as isolated vertices, extends it to every larger size with $\log N=\Theta(\log n)$, and we discuss the padding only when it needs care. Fixed thresholds may depend on $N$ but are chosen before updates. Numbered references to \cite{abboud_popular_2014,henzinger_unifying_2015,abboud_popular_2016,van_den_brand_dynamic_2019,dallant_conditional_2022,jin_tight_2022} are to the full versions listed in the bibliography, unless stated otherwise.

For a target data structure on instances of size $N$, in the model of \cref{sec:model}, with word size $w=\Theta(\log N)$, worst-case update time $t_u$ and expected query time $t_q$, let
\[
\mathcal L_N:=\left(\frac{\log N}{\log(t_u w)}\right)^2.
\]

\begin{lemma}[Transfer]\label{lem:app-transfer}
Consider a reduction from insert-only multiphase Disjointness with $m=n$ to a dynamic problem on instances of size $N=n^{\Theta(1)}$. The initial target instance depends only on $S_1,\dots,S_m$, and the reduction has the following properties.
\begin{enumerate}[label=(\roman*), leftmargin=*]
  \item Each insertion into $X$ is simulated by $O(1)$ target updates, and the query $q$ by $O(1)$ target updates and target queries, in some order. Each source operation also probes $O(1)$ further cells, and the answer to $q$ is determined by the answers to the target queries and the contents of these cells.
  \item The target operations form a valid operation sequence for the target problem, determined by $S_1,\dots,S_m$, the sequence of insertions and $q$, and each of their arguments occupies $O(1)$ words.
\end{enumerate}
Then every zero-error target data structure satisfies $t_u+t_q=\Omega(\mathcal L_N)$, and even $t_q=\Omega(\mathcal L_N)$ if the query $q$ is simulated without target updates. In either case, $\max(t_u,t_q)=\Omega\big((\log N/\log\log N)^2\big)$.
\end{lemma}
\begin{proof}
We build a data structure for insert-only multiphase Disjointness. It stores the memory of the target data structure and $O(1)$ auxiliary tables, such as the characteristic vector of $X$, in disjoint address ranges, which changes the word size by at most a constant factor; so each insertion probes $O(t_u)$ cells. During the query, it keeps every write of the simulated target operations, including writes by target queries, in temporary workspace, from which it answers later reads of the same cells. This reproduces the target execution while leaving the memory unchanged, as \cref{prob:multiphase} requires and \cref{sec:model} allows. By~(ii), each target query runs on an operation sequence that does not depend on the random tape, so it probes at most $t_q$ cells in expectation. Hence the query probes $O(t_u+t_q)$ cells in expectation, and $O(t_q)$ if it uses no target updates. Since $\log N=\Theta(\log n)$, \cref{thm:main} gives the first two bounds. For the last one, if $t_u\ge\log^2N$ there is nothing to prove; otherwise $\log(t_uw)=O(\log\log N)$.
\end{proof}

\begingroup
\small
\renewcommand{\arraystretch}{1.12}
\setlength{\LTcapwidth}{\textwidth}
\begin{longtable}{@{}>{\raggedright\arraybackslash}p{0.68\textwidth}>{\raggedright\arraybackslash}p{0.14\textwidth}p{\dimexpr0.18\textwidth-4\tabcolsep\relax}@{}}
\caption{The bounds of this appendix. Each listed quantity is $\Omega(\mathcal L_N)$, so in every case $\max(t_u,t_q)=\Omega\big((\log N/\log\log N)^2\big)$.}\label{tab:appendix-transfers}\\
\toprule
Problem & Bounded quantity & Section\\
\midrule
\endfirsthead
\multicolumn{3}{@{}l}{\tablename~\thetable\ (continued)}\\
\toprule
Problem & Bounded quantity & Section\\
\midrule
\endhead
\multicolumn{3}{r@{}}{\footnotesize Continued on the next page.}\\
\endfoot
\bottomrule
\endlastfoot
Directed reachability; undirected shortest paths, exact or within a factor $\alpha<2$; subgraph connectivity & $t_q$ & \ref{sec:graph-problems}\\
$s$-$t$ reachability; \#SSR; undirected $s$-$t$ shortest paths, exact or within a factor $\alpha<5/3$; strong connectivity; cycle detection; triangle detection & $t_u+t_q$ & \ref{sec:graph-problems}\\
Bipartite perfect matching; maximum matching size & $t_u+t_q$ & \ref{sec:bipartite-matching}\\
Exact densest subgraph, or a fixed density threshold & $t_u+t_q$ & \ref{sec:app-density}\\
Nearest marked vertex, exact or within a one-sided factor $\alpha<3$ & $t_q$ & \ref{sec:app-marked-distance}\\
Planar distance, unit or positive weights & $t_q$ & \ref{sec:planar-distance}\\
Planar directed girth, unit or positive weights & $t_u+t_q$ & \ref{sec:planar-girth}\\
Planar diameter, positive weights & $t_u+t_q$ & \ref{sec:planar-diameter}\\
Planar bipartite maximum-weight matching & $t_u+t_q$ & \ref{sec:planar-matching}\\
Set disjointness; Pagh's problem; range mode & $t_q$ & \ref{sec:other-problems}\\
Langerman's problem; Erickson's problem & $t_u+t_q$ & \ref{sec:other-problems}\\
Inverse-entry zero testing over prime fields & $t_q$ & \ref{sec:app-inverse}\\
Determinant over prime fields, or whether it equals one & $t_u+t_q$ & \ref{sec:app-determinant}\\
Three-dimensional convex-hull vertex counting & $t_u+t_q$ & \ref{sec:hull}\\
Three-dimensional maximal-point counting & $t_u+t_q$ & \ref{sec:maxima}\\
Two-dimensional distinct-color counting & $t_q$ & \ref{sec:distinct-color-counting}\\
Internal dictionary matching & $t_q$ & \ref{sec:internal-dictionary}\\
Binary Hamming-distance threshold & $t_q$ & \ref{sec:hamming-threshold}\\
Active-document existence & $t_q$ & \ref{sec:active-documents}\\
\end{longtable}
\endgroup

\subsection{Graph problems}\label{sec:graph-problems}

Here $N$ is the number of vertices. Except in subgraph connectivity and in \cref{sec:app-marked-distance}, the graph initially has no edges, and an update inserts an edge.

In \emph{directed reachability}, queries test whether one vertex reaches another; \emph{$s$-$t$ reachability} fixes both vertices, and \emph{\#SSR} counts vertices other than a fixed source reachable from it. \emph{Strong-connectivity} queries test a pair, count strongly connected components, or report the largest component's size~\cite{abboud_popular_2014,haeupler_incremental_2012}. \emph{Cycle detection} tests for a directed cycle~\cite{haeupler_incremental_2012,chen_almost-linear_2024}, and \emph{triangle detection} for an undirected triangle. \emph{Undirected shortest-path} queries return a distance $d$, or an approximation in $[d,\alpha d]$, preserving $d=\infty$; the $s$-$t$ version fixes both vertices. In \emph{subgraph connectivity}, updates turn on vertices of a fixed undirected graph, and queries test connectivity in the induced subgraph of vertices that are on.

\paragraph{The reduction.}
The reductions for reachability, shortest paths and subgraph connectivity are those of \Patrascu~\cite[Appendix~A]{patrascu_mihai_towards_2010}. Let $G$ be the graph with vertices $u_1,\dots,u_m$, $v_1,\dots,v_n$ and $t$, and with the edges $\{u_i,v_j\}$ for $j\in S_i$, inserted in Phase~I, and $\{v_j,t\}$ for $j\in X$: inserting $j$ into $X$ inserts $\{v_j,t\}$. In the directed problems, these edges are directed from $u_i$ to $v_j$ and from $v_j$ to $t$, so $G$ is a DAG of depth $2$. For subgraph connectivity, the fixed graph has the edges $\{u_i,v_j\}$ for $j\in S_i$ and $\{v_j,t\}$ for all $j$; initially exactly the $u_i$ and $t$ are on, and inserting $j$ into $X$ turns $v_j$ on. (\Patrascu instead turns on only the queried vertex $u_q$, in the final phase.)

\paragraph{Queries without an update.}
The out-neighbors of $u_q$ are the $v_j$ with $j\in S_q$, and $v_j$ has an edge to $t$ exactly when $j\in X$, so $t$ is reachable from $u_q$ exactly when $S_q\cap X\ne\emptyset$. Reversing all edges gives the version with the fixed source $t$. In the undirected graph, the distance between $u_q$ and $t$ is $2$ in this case, and at least $4$ otherwise, since $G$ is bipartite with $u_q$ and $t$ on the same side; so an approximation within any factor $\alpha<2$ tells the two cases apart. In subgraph connectivity, the neighbors of $u_q$ that are on are the $v_j$ with $j\in S_q\cap X$, and each of them is adjacent to $t$, so $u_q$ and $t$ are connected exactly when $S_q\cap X\ne\emptyset$.

\paragraph{Queries with one update.}
\begin{itemize}[leftmargin=*]
  \item \emph{$s$-$t$ reachability, \#SSR and $s$-$t$ shortest paths.} The directed cases are selector variants of \Patrascu's reduction, and the undirected case specializes Henzinger, Krinninger, Nanongkai and Saranurak~\cite[Lemma~3.2]{henzinger_unifying_2015} to a single set. Add a vertex $r$, isolated in Phases~I and~II, as the fixed source, with the fixed target $t$. The query $q$ inserts the edge $(r,u_q)$, or $\{r,u_q\}$ in the undirected case. In the directed case, $t$ is reachable from $r$ exactly when $S_q\cap X\ne\emptyset$, and the number of vertices other than $r$ that are reachable from $r$ is $\abs{S_q}+1+\mathbf 1_{S_q\cap X\ne\emptyset}$, which determines the answer since we store $\abs{S_q}$. In the undirected case, the distance between $r$ and $t$ is $3$ if $S_q\cap X\ne\emptyset$, and at least $5$ otherwise, so an approximation within any factor $\alpha<5/3$ tells the two cases apart.
  \item \emph{Strong connectivity and cycle detection.} The query $q$ inserts the edge $(t,u_q)$. Since $G$ is acyclic, every cycle uses this edge, so a cycle exists exactly when $t$ is reachable from $u_q$ in $G$, that is, when $S_q\cap X\ne\emptyset$. In that case $u_q$, $t$ and some $v_j$ lie in one strongly connected component; otherwise every strongly connected component is a single vertex. So each of the three strong-connectivity queries tells the two cases apart.
  \item \emph{Triangle detection.} This specializes \cite[Lemma~3.3]{henzinger_unifying_2015} to a single set. The query $q$ inserts the edge $\{u_q,t\}$. Since $G$ is bipartite with $u_q$ and $t$ on the same side, every triangle consists of $u_q$, $t$ and some $v_j$ with $j\in S_q\cap X$.
\end{itemize}

\begin{corollary}[Dynamic graph problems]\label{thm:reach}
Consider a zero-error cell-probe data structure on $N$ vertices with word size $w=\Theta(\log N)$, worst-case update time $t_u$ and expected query time $t_q$.
\begin{enumerate}[label=(\alph*), leftmargin=*]
  \item For incremental directed reachability, incremental undirected shortest paths, exact or within any factor $\alpha<2$, and incremental subgraph connectivity, $t_q=\Omega(\mathcal L_N)$. For directed reachability, this holds already on depth-$2$ DAGs whose queries ask whether a fixed sink is reachable, or whether a vertex is reachable from a fixed source. For undirected shortest paths, it holds already for distances from a fixed source.
  \item For incremental $s$-$t$ reachability, \#SSR, undirected $s$-$t$ shortest paths, exact or within any factor $\alpha<5/3$, strong connectivity in each of the three forms above, cycle detection, and triangle detection, $t_u+t_q=\Omega(\mathcal L_N)$.
\end{enumerate}
In particular, $\max(t_u,t_q)=\Omega\big((\log N/\log\log N)^2\big)$ for all of these problems.
\end{corollary}
\begin{proof}
The graph has at most $2n+2$ vertices. Each insertion into $X$ is simulated by one update, and the query by one query, preceded by one update in~(b). So \cref{lem:app-transfer} gives the bounds.
\end{proof}

\subsubsection{Bipartite matching}\label{sec:bipartite-matching}

The graph is bipartite and unweighted, an update inserts an edge, and a query asks for the size of a maximum matching, or whether a perfect matching exists. We specialize the reduction from reachability to bipartite matching of Abboud and Vassilevska Williams~\cite[Lemma~6.2]{abboud_popular_2014}. The two sides are
\[
 L=\{s_L\}\cup\{u_{i,L}:i\in[n]\}\cup\{v_{j,L}:j\in[n]\},
 \qquad
 R=\{t_R\}\cup\{u_{i,R}:i\in[n]\}\cup\{v_{j,R}:j\in[n]\}.
\]
In Phase~I, insert the edges $\{u_{i,L},u_{i,R}\}$ and $\{v_{j,L},v_{j,R}\}$, and the edges $\{u_{i,L},v_{j,R}\}$ for $j\in S_i$. Inserting $j$ into $X$ inserts $\{v_{j,L},t_R\}$. The query $q$ inserts $\{s_L,u_{q,R}\}$, and asks whether there is a perfect matching.

The edges $\{u_{i,L},u_{i,R}\}$ and $\{v_{j,L},v_{j,R}\}$ form a matching of size $2n$, which leaves exactly $s_L$ and $t_R$ unmatched, and its augmenting paths are exactly
\[
 s_L,\ u_{q,R},\ u_{q,L},\ v_{j,R},\ v_{j,L},\ t_R
 \qquad (j\in S_q\cap X).
\]
So a perfect matching, of size $2n+1$, exists exactly when $S_q\cap X\ne\emptyset$, and otherwise the maximum matching has size $2n$. The graph has $4n+2$ vertices. For perfect matching, take even $N$ and pad with disjoint edges; for maximum matching size, an isolated vertex also permits odd $N$. Each insertion into $X$ is one update, and the query is one update and one query, so \cref{lem:app-transfer} gives $t_u+t_q=\Omega(\mathcal L_N)$.

\subsubsection{Exact densest subgraph}\label{sec:app-density}

The density of a nonempty vertex set $U$ is $\abs{E(G[U])}/\abs U$. The graph is simple and unweighted, an update inserts an edge, and a query asks for the maximum density $\rho^*(G)$ over all nonempty $U$, or whether $\rho^*(G)\ge\theta$ for a threshold $\theta$ fixed in advance. We specialize the reduction of Henzinger, Krinninger, Nanongkai and Saranurak~\cite[Section~3.6, Theorems~3.23 and~3.25]{henzinger_unifying_2015}, whose vector-matrix-vector formulation we use with the row vector $e_q$, supplied in the final phase.

Let $k:=6n$. For each $i\in[n]$, the graph has a triple $R_i$ of vertices with a designated vertex $r_i\in R_i$, and for each $j\in[n]$, a triple $Y_j$ with a designated vertex $y_j\in Y_j$. For each $(i,j)\in[n]^2$, it has a set $P_{ij}$ of $k$ further vertices with two distinct endpoints $a_{ij}$ and $b_{ij}$. In Phase~I, insert the edges $r_ia_{ij}$ and $b_{ij}y_j$ for all $i,j$, and, if $M_{ij}=1$, a path from $a_{ij}$ to $b_{ij}$ through all vertices of $P_{ij}$. A triple is \emph{active} if its vertices form a triangle. Inserting $j$ into $X$ inserts the three edges of a triangle on $Y_j$. The query $q$ inserts the three edges of a triangle on $R_q$, and asks whether $\rho^*(G)\ge\theta$, where
\[
\theta:=\frac{k+7}{k+6}.
\]
The graph has $kn^2+6n=6n^3+6n$ vertices, and isolated vertices pad it without changing $\rho^*(G)$.

\begin{lemma}\label{lem:app-density}
After the query inserts the triangle on $R_q$, we have $\rho^*(G)\ge\theta$ if and only if $S_q\cap X\ne\emptyset$.
\end{lemma}
\begin{proof}
If $j\in S_q\cap X$, then $R_q\cup P_{qj}\cup Y_j$ has $k+6$ vertices and induces $(k-1)+2+3+3=k+7$ edges, so its density is $\theta$.

Conversely, suppose that $S_q\cap X=\emptyset$, and that some nonempty $U$ has $\Phi(U):=\abs{E(G[U])}-\theta\abs U\ge0$. We modify $U$ in the following steps, each of which strictly increases $\Phi(U)$; in particular, $U$ never becomes empty, since $\Phi(\emptyset)=0$.
\begin{itemize}[leftmargin=*]
  \item For each $(i,j)$, remove every component of $G[U\cap P_{ij}]$, unless it is all of $P_{ij}$, $M_{ij}=1$, and $r_i,y_j\in U$. A removed component with $v$ vertices spans at most $v-1$ edges and has at most one edge to $r_i$ or $y_j$, so its removal deletes at most $v$ edges and increases $\Phi(U)$ by at least $(\theta-1)v>0$.
  \item For each triple, if it is not active, remove its vertices other than its designated vertex; they are isolated. If it is active and its designated vertex is not in $U$, remove its remaining vertices, which span at most one edge. If it is active and its designated vertex is in $U$, add its missing vertices: adding one vertex adds two edges, and adding two vertices adds three edges, so $\Phi(U)$ increases since $\theta<3/2$. These steps do not create new edges to the sets $P_{ij}$, since only designated vertices have such edges.
\end{itemize}
Now let $\pi$ be the number of sets $P_{ij}\subseteq U$, $\tau$ the number of active triples contained in $U$, and $\delta$ the number of designated vertices of inactive triples in $U$. Then $G[U]$ has $v=k\pi+3\tau+\delta$ vertices and $e=(k+1)\pi+3\tau$ edges. The only active triple $R_i$ is $R_q$, and the triples $Y_j$ with $j\in S_q$ are inactive, since $S_q\cap X=\emptyset$; so each $P_{ij}\subseteq U$ has an endpoint adjacent to the designated vertex of an inactive triple in $U$. Every designated vertex is adjacent to at most $n$ of the sets $P_{ij}$, so $\pi\le n\delta$. Moreover $\tau+\delta>0$, since $U$ is nonempty and every $P_{ij}\subseteq U$ comes with $r_i,y_j\in U$. Using $k=6n$,
\[
(k+7)v-(k+6)e=3\tau+(k+7)\delta-6\pi\ge3\tau+7\delta>0,
\]
so $e/v<\theta$, which contradicts $\Phi(U)\ge0$.
\end{proof}

Each insertion into $X$ is three updates, and the query is three updates and one query, so \cref{lem:app-transfer} gives $t_u+t_q=\Omega(\mathcal L_N)$. The numerator and denominator of $\theta$ have $O(\log n)$ bits. The bound concerns exact densities: $\theta$ tends to $1$, so it does not extend to constant-factor approximation.

\subsubsection{Nearest marked vertex}\label{sec:app-marked-distance}

The graph is fixed, undirected and unweighted, and every vertex is marked or unmarked. An update marks a vertex, and a query asks for the distance $d(v)$ from a given vertex $v$ to the nearest marked vertex, or for a one-sided approximation $\widehat d(v)$ with $d(v)\le\widehat d(v)\le\alpha\,d(v)$, where $1\le\alpha<3$. We adapt the incidence-graph construction of Henzinger, Krinninger, Nanongkai and Saranurak~\cite[Lemma~3.7]{henzinger_unifying_2015} to marking only, adding a sentinel and hubs so that the graph is connected and all distances are finite.

The graph has vertices $r_1,\dots,r_n$, $v_1,\dots,v_n$, $a$, $b$ and $z$, and the edges $r_iv_j$ for $j\in S_i$, and $r_ia$, $ab$, $bz$ and $bv_j$ for all $i$ and $j$. It is connected and bipartite, with sides $\{r_1,\dots,r_n,b\}$ and $\{v_1,\dots,v_n,a,z\}$. Initially only $z$ is marked. Inserting $j$ into $X$ marks $v_j$, and the query $q$ asks for $d(r_q)$.

The marked vertices are $z$ and the $v_j$ with $j\in X$. So $r_q$ has a marked neighbor, and $d(r_q)=1$, exactly when $S_q\cap X\ne\emptyset$. Otherwise all marked vertices lie on the other side of the bipartition and are not adjacent to $r_q$, so their distance from $r_q$ is odd and at least $3$, and the path $r_q,a,b,z$ shows that $d(r_q)=3$. So the test $\widehat d(r_q)<3$ tells the two cases apart. The graph has $2n+3$ vertices, and unmarked leaves attached to $b$ pad it without changing these distances. The costs are one update per insertion and one query, so \cref{lem:app-transfer} gives $t_q=\Omega(\mathcal L_N)$.

\subsection{Planar graph problems}\label{sec:planar}

Here $N$ is the number of vertices, an update inserts an edge, and every graph that arises is planar. The reductions adapt the two-grid construction of Abboud and Dahlgaard~\cite[Definition~1, Corollary~1, and Sections~3 and~6]{abboud_popular_2016} and their selectors for diameter and girth. We add a permanent dummy connection and choose positive connector lengths, so that all updates are insertions. A local shortcut first encodes one matrix entry in a grid distance; a second grid and the connectors turn these entry tests into an OR over the inserted elements.

\paragraph{The grid.}
For a Boolean $p\times p$ matrix $B$ with $p\ge2$, the grid $G(B)$ has vertices $u_{rc}$, $v_{rc}$, $w_{rc}$ and $x_{rc}$ for $r,c\in[p]$, top terminals $a_c$, right terminals $b_r$, and the following edges.
\begin{center}
\begin{tabular}{@{}ll@{}}
\toprule
Edge & Length\\
\midrule
$a_cv_{1c}$, and $u_{r-1,c}v_{rc}$ for $r>1$ & $2c-1$\\
$v_{rc}u_{rc}$ & $1$\\
$u_{rc}w_{rc}$ & $2$\\
$w_{rc}u_{r,c+1}$ for $c<p$, and $w_{rp}b_r$ & $2p-2$\\
$v_{rc}x_{rc}$ & $1$\\
$x_{rc}w_{rc}$, present exactly when $B_{rc}=1$ & $1$\\
\bottomrule
\end{tabular}
\end{center}
When $B_{rc}=0$, the vertex $x_{rc}$ is a leaf, which does not affect distances between terminals; so $G(B)$ has $4p^2+2p$ vertices for every $B$.

\begin{lemma}[Grid distances]\label{lem:planar-grid-distance}
For every Boolean $p\times p$ matrix $B$ with $p\ge2$ and all $q,j\in[p]$,
\[
 d_{G(B)}(a_q,b_j)=2p(p-q+1)+2qj-B_{jq},
\]
which is attained by the path down column $q$ and then right along row $j$, through $x_{jq}$ if $B_{jq}=1$.
\end{lemma}
\begin{proof}
\emph{Compressing each cell.} Identify each vertex $u_{rc}$ with the corner $(r,c)$, the top terminal $a_c$ with $(0,c)$, and the right terminal $b_r$ with $(r,p+1)$. The cell $\{v_{rc},w_{rc},x_{rc}\}$ meets the rest of the graph only at the three corners $(r-1,c)$, $(r,c)$ and $(r,c+1)$. Within the cell, the distances from $(r-1,c)$ to $(r,c)$, from $(r,c)$ to $(r,c+1)$, and from $(r-1,c)$ to $(r,c+1)$ are $2c$, $2p$ and $2c+2p-B_{rc}$: the edge $x_{rc}w_{rc}$ saves one unit only between the first and third corners. Replace the cell by the triangle on its three corners with these edge lengths. Since the cells are disjoint, this preserves all distances between corners. The compressed graph has a vertical, a horizontal and a diagonal edge for every $r,c\in[p]$, and no horizontal edges on row $0$.

\smallskip\noindent\emph{A lower-bound potential.} Fix $q,j$, and let $\beta:=B_{jq}$ and $z_r:=\max\{j-r,0\}$. Define
\[
 F(r,c):=2qj+2p(c-q)-2\min\{c,q\}z_r+\eta(r,c),
 \qquad
 \eta(r,c):=
 \begin{cases}
  1,&r<j,\ c<q,\\
  -1,&r<j,\ c>q,\\
  -\beta,&r\ge j,\ c>q,\\
  0,&\text{otherwise}.
 \end{cases}
\]
Let $V_{rc}:=F(r,c)-F(r-1,c)$, $H_{rc}:=F(r,c+1)-F(r,c)$ and $D_{rc}:=V_{rc}+H_{rc}$ be the increments along the vertical, horizontal and diagonal edges. Direct subtraction gives $0\le V_{rc}\le2c$ and $0\le H_{rc}\le2p$; for the lower bound on $H_{rc}$, use $r\ge1$, so that $z_r\le p-1$ and the smallest possible value is $2p-2z_r-1\ge1$. The diagonal increments are as follows, with an upper bound in the top-left entry and equality elsewhere:
\[
\begin{array}{c|ccc}
 &c<q&c=q&c>q\\\hline
 r<j&\le2p+2c-2&2p+2q-1&2p+2q\\
 r=j&2p+2c-1&2p+2q-\beta&2p+2q+1-\beta\\
 r>j&2p&2p-\beta&2p
\end{array}
\]
So $0\le D_{rc}\le2p+2c-1$ for $(r,c)\ne(j,q)$, and $D_{jq}=2p+2q-\beta$. Since $B$ is Boolean, every diagonal increment is at most the length $2p+2c-B_{rc}$ of its edge. Hence the potential changes by at most the length of each compressed edge, in either direction.

\smallskip\noindent\emph{Attaining the bound.} Since $F(0,q)=0$ and $F(j,p+1)=2qj+2p(p-q+1)-\beta$, every path from $a_q$ to $b_j$ has at least the claimed length. The path consisting of $j-1$ vertical edges in column $q$, the diagonal edge of cell $(j,q)$, and $p-q$ horizontal edges in row $j$ has exactly this length, and expanding its edges gives the stated path.
\end{proof}

\paragraph{The core.}
Let $p:=n+1$, pad $M$ with a zero row and a zero column, and let $B:=M^{\mathsf T}$; so the rows of $B$ are indexed by elements and its columns by sets, and row $p$ of $B$ is zero. To the right of $G(B)$, place a horizontally reflected copy of $G(0)$, with terminals $b'_1,\dots,b'_p$ facing $b_1,\dots,b_p$ and top terminals $a'_1,\dots,a'_p$, and let $r_q:=a'_{p-q+1}$. By \cref{lem:planar-grid-distance}, applied to $G(0)$, the distance from $b'_j$ to $r_q$ in the copy is $2pq+2(p-q+1)j$. Let
\[
 H:=4p(p+1),\qquad \lambda_j:=H+2(p+1)(p-j),\qquad D:=2H.
\]
For $j\in[p]$, the \emph{connector} of $j$ is an edge $b_jb'_j$ of length $\lambda_j$; the connectors join facing terminals, so they do not cross. The connector of $p$ is present from the start, and the connector of $j\in[n]$ is inserted when $j$ is inserted into $X$. The core has $8p^2+4p$ vertices, and every edge length is at most $\lambda_1=6p^2+4p-2<D$.

A path from $a_q$ to $r_q$ that uses the connector of $j$ once, and shortest paths in the two grids (\cref{fig:planar-core}), has length
\[
 \bigl(2p(p-q+1)+2qj-B_{jq}\bigr)+\lambda_j+\bigl(2pq+2(p-q+1)j\bigr)=D-B_{jq}.
\]
A path from one grid to the other uses an odd number of connectors, and a path that uses at least three of them has length at least $3H>D$, while the connector of $p$ gives a path of length $D$. Hence
\begin{equation}\label{eq:planar-bit-distance}
 d(a_q,r_q)=D-f_q,\qquad\text{where}\qquad f_q:=\max_{j\in X\cup\{p\}}B_{jq}=\mathbf 1_{S_q\cap X\ne\emptyset}.
\end{equation}

\begin{figure}[ht]
\centering
\begin{tikzpicture}[x=0.9cm,y=0.65cm,font=\small,
  route/.style={very thick,blue!65!black},
  terminal/.style={circle,fill=black,inner sep=1.4pt}]
  \foreach \c in {1,...,4} {
    \draw[gray!65] (\c,0)--(\c,-4);
    \draw[gray!65] (\c+8,0)--(\c+8,-4);
  }
  \foreach \r in {1,...,4} {
    \draw[gray!65] (1,-\r)--(5,-\r);
    \draw[gray!65] (8,-\r)--(12,-\r);
    \node[terminal] at (5,-\r) {};
    \node[terminal] at (8,-\r) {};
  }
  \draw (5,-1)--(8,-1);
  \draw[gray!65,dashed] (5,-2)--(8,-2);
  \draw (5,-4)--(8,-4);
  \draw[route] (2,0)--(2,-2)--(3,-3)--(5,-3)--(8,-3)--(10,-3)--(10,0);
  \node[terminal,label=above:{$a_q$}] at (2,0) {};
  \node[terminal,label=above:{$r_q$}] at (10,0) {};
  \node at (3,1.35) {$G(B)$};
  \node at (10,1.35) {reflected $G(0)$};
  \node[above] at (5,-3) {$b_j$};
  \node[above] at (8,-3) {$b'_j$};
  \node[above] at (6.5,-3) {$\lambda_j$};
  \node[below] at (6.5,-4) {dummy connector};
  \node[left] at (2.5,-2.5) {$B_{jq}$};
\end{tikzpicture}
\caption{The planar core, with cell interiors suppressed ($p=4$, $q=2$, $X=\{1,3\}$). The highlighted path uses the cell $(j,q)$ and the connector of $j=3$; its length is $D-B_{jq}$. The second grid and connector length cancel the dependence on $j$. Solid connectors are present; the dashed connector is absent. The dummy connector guarantees distance at most $D$.}
\label{fig:planar-core}
\end{figure}

\subsubsection{Distance}\label{sec:planar-distance}

A query asks for the distance between two vertices. The query $q$ asks for $d(a_q,r_q)$, which by~\eqref{eq:planar-bit-distance} is at most $D-1$ exactly when $S_q\cap X\ne\emptyset$; so the bound holds already for this threshold query.

For unit weights, subdivide every edge of the two grids into edges of length $1$, and replace the connector of $j$ by a path of length $\lambda_j$ from $b_j$ to $b'_j$. The path of $p$ is present from the start, and the path of $j\in[n]$ is present from the start except for its edge at $b_j$, which is inserted when $j$ is inserted into $X$. The incomplete paths hang from the right grid, so the graph stays connected, and~\eqref{eq:planar-bit-distance} still holds. Subdividing a grid adds its total edge length $3p^3+2p^2+\sum_{r,c}B_{rc}$, minus its number of edges $5p^2+\sum_{r,c}B_{rc}$, to its $4p^2+2p$ vertices, giving $3p^3+p^2+2p$ vertices; the connector paths add $\sum_j(\lambda_j-1)=5p^3+4p^2-2p$ vertices. So the graph has $11p^3+6p^2+2p$ vertices, for every $B$.

In both versions, pendant leaves pad the graph, each insertion into $X$ is one update, and the query is one query, so \cref{lem:app-transfer} gives $t_q=\Omega(\mathcal L_N)$.

\subsubsection{Directed girth}\label{sec:planar-girth}

A query asks for the length of a shortest directed cycle. As in \cite[Section~6]{abboud_popular_2016}, direct the horizontal grid edges to the right, the vertical edges of the left grid downwards and those of the right grid upwards, each path $v_{rc}x_{rc}w_{rc}$ from $v_{rc}$ to $w_{rc}$, the remaining leaf edges away from the grids, and each connector from $b_j$ to $b'_j$. Then both grids are acyclic, and all connectors go from the left grid to the right one, so the core is acyclic. The paths that attain the distances of \cref{lem:planar-grid-distance} in the two grids respect these directions, and directing the edges cannot shorten paths, so the directed distance from $a_q$ to $r_q$ is still $D-f_q$.

The query $q$ inserts an edge of length $1$ from $r_q$ to $a_q$, which can be drawn in the outer face, and asks for the girth $g$. Every cycle uses this edge, so $g=D+1-f_q$, and $g\le D$ exactly when $S_q\cap X\ne\emptyset$. This holds for positive weights and, after the subdivision of \cref{sec:planar-distance}, with each path directed like the edge it replaces, for unit weights. Directed leaves pad the graph. Each insertion into $X$ is one update, and the query is one update and one query, so \cref{lem:app-transfer} gives $t_u+t_q=\Omega(\mathcal L_N)$.

\subsubsection{Diameter}\label{sec:planar-diameter}

The graph is connected and has positive integer weights, and a query asks for the largest distance between two vertices. We add to the undirected core the selectors of \cite[Section~6]{abboud_popular_2016}, with permanent fallback edges that keep the graph connected. Let $L:=64p^3$. For each $c\in[p]$, add vertices $z_c$ and $z'_c$ with edges $z_ca_c$ and $z'_cr_c$ of length $1$, and add vertices $s$ and $t$ with edges $sz_c$ and $tz'_c$ of length $3L$. These two fans lie in the outer face, next to the two grids, and the edges $sa_c$ and $tr_c$ that a query may insert can be drawn next to the paths $sz_ca_c$ and $tz'_cr_c$; so the graph stays planar. It has $8p^2+6p+2$ vertices. Let $C$ be the graph without $s$ and $t$. Even with all connectors present, the total edge length of $C$ is at most $11p^3+9p^2+p<16p^3$. To pad, attach leaves to $C$; fewer than $16p+14$ are needed between consecutive values of $p$, so the total edge length of $C$ stays below $L/2$, and any two vertices of $C$ are at distance less than $L/2$.

The query $q$ inserts the edges $sa_q$ and $tr_q$, of length $L$ each, and asks for the diameter. A path between two vertices of $C$ through $s$ or $t$ has length at least $2L$, so it does not shorten their distance in $C$. A path from $s$ to $t$ that uses an edge $sz_c$ or $tz'_c$ has length at least $4L$, while the two new edges give a path of length at most $2L+D<3L$; so
\[
 d(s,t)=2L+d_C(a_q,r_q)=2L+D-f_q.
\]
Every other distance is less than $2L$: it is less than $L/2$ within $C$, and less than $L+L/2$ between $s$ or $t$ and a vertex of $C$. So the diameter is $2L+D-f_q$, which is at most $2L+D-1$ exactly when $S_q\cap X\ne\emptyset$. The largest weight is $3L=O(n^3)$; we do not subdivide the selector edges, which would change the diameter. Each insertion into $X$ is one update, and the query is two updates and one query, so \cref{lem:app-transfer} gives $t_u+t_q=\Omega(\mathcal L_N)$.

\subsubsection{Maximum-weight bipartite matching}\label{sec:planar-matching}

The graph is bipartite and has positive integer weights, and a query asks for the maximum weight of a matching. We adapt the vertex-splitting and weight-complement construction of Abboud and Dahlgaard~\cite[Section~4]{abboud_popular_2016} to the directed core of \cref{sec:planar-girth}, without the edge from $r_q$ to $a_q$. Let $N_{\mathrm{core}}:=8p^2+4p$ be the number of vertices of the core, and let $K:=D+1$. Replace each vertex $v$ of the core by two vertices $v^{\mathrm{in}}$ and $v^{\mathrm{out}}$, joined by an edge of weight $K$, and each directed edge from $u$ to $v$ of length $\ell$ by an edge $u^{\mathrm{out}}v^{\mathrm{in}}$ of weight $K-\ell$, which is positive since $\ell<D$. The graph is bipartite, with the vertices $v^{\mathrm{out}}$ on one side and the vertices $v^{\mathrm{in}}$ on the other. At each vertex of the core, the incoming and the outgoing edges occupy two contiguous sectors: in the left grid, the incoming edges come from the north and west and the outgoing edges leave to the south and east, and analogously in the right grid. So splitting the vertices preserves planarity, as in~\cite{abboud_popular_2016}. Add a vertex $\sigma$ on the side of the vertices $v^{\mathrm{out}}$ and a vertex $\tau$ on the other side, isolated in Phases~I and~II. Inserting $j$ into $X$ inserts the edge that replaces the connector of $j$. The query $q$ inserts the edges $\sigma a_q^{\mathrm{in}}$ and $r_q^{\mathrm{out}}\tau$, of weight $K$ each, which can be drawn in the outer face, and asks for the maximum weight of a matching.

Compare a perfect matching with the matching that consists of the $N_{\mathrm{core}}$ edges $v^{\mathrm{in}}v^{\mathrm{out}}$. Their symmetric difference consists of an alternating path from $\sigma$ to $\tau$, which corresponds to a directed path from $a_q$ to $r_q$ in the core, and of alternating cycles, which would correspond to directed cycles and hence do not exist. Conversely, every directed path from $a_q$ to $r_q$ gives a perfect matching in this way, and if the path has length $d$, then the perfect matching has weight $(N_{\mathrm{core}}+1)K-d$. Since $d(a_q,r_q)\le D$, some perfect matching has weight at least $(N_{\mathrm{core}}+1)K-D=N_{\mathrm{core}}K+1$, while a matching that is not perfect has at most $N_{\mathrm{core}}$ edges, each of weight at most $K$. So the maximum weight of a matching is
\[
 (N_{\mathrm{core}}+1)K-(D-f_q)=N_{\mathrm{core}}K+1+f_q,
\]
which is at least $N_{\mathrm{core}}K+2$ exactly when $S_q\cap X\ne\emptyset$. The graph has $2N_{\mathrm{core}}+2$ vertices and weights $O(n^2)$, and isolated vertices pad it. Each insertion into $X$ is one update, and the query is two updates and one query, so \cref{lem:app-transfer} gives $t_u+t_q=\Omega(\mathcal L_N)$.

\subsection{Set and array problems}\label{sec:other-problems}

\Patrascu~\cite{patrascu_mihai_towards_2010} reduced the multiphase problem to Langerman's, Pagh's and Erickson's problems, and Chan, Durocher, Larsen, Morrison and Wilkinson~\cite{chan_linear-space_2014} reduced it to range mode. For set disjointness and Pagh's problem, $N$ is the size of the universe, and the instances have $O(N)$ sets; for the other problems, $N$ bounds the number of array or matrix entries, and the constructions use $\Theta(n^2)$ entries.

\subsubsection{Set disjointness}\label{sec:set-disjointness}

A family of subsets of $[N]$ is maintained under insertions of elements into its sets, and a query asks whether two sets of the family are disjoint \cite{kopelowitz_dynamic_2015}. This problem contains multiphase Disjointness directly: the family consists of $S_1,\dots,S_m$ and $X$, inserting an element into $X$ is an update, and the query $q$ asks whether $S_q$ and $X$ are disjoint.

\subsubsection{Langerman's problem}\label{sec:langerman}

An array of integers is maintained under updates that assign a value to an entry, and a query asks whether some nonempty prefix has sum zero. The reduction is \Patrascu's~\cite[Appendix~A]{patrascu_mihai_towards_2010}. The array consists of $A[1]$, initially $1$, followed by a block for each $j\in[n]$: a head $h_j$, then a pair $(a_{j,i},b_{j,i})$ for each $i\in[m]$, and then a tail $t_j$. In Phase~I, set $(a_{j,i},b_{j,i}):=(1,1)$ if $j\in S_i$ and $(2,0)$ otherwise, and $(h_j,t_j):=(0,-2m)$. Inserting $j$ into $X$ sets $(h_j,t_j):=(-2m,0)$. The query $q$ sets $A[1]:=2(m-q)+1$, and asks whether some prefix sum is zero.

Every block sums to zero, so a prefix that ends at $A[1]$ or at a tail has sum $A[1]>0$. A prefix that ends at the head or at a pair entry of block $j$ has sum $A[1]+h_j+y$, where $y\in[0,2m]$ is the sum of the pair entries of block $j$ in the prefix; $y$ is even, or $y=2i-1$ for some $i$ with $j\in S_i$. If $j\notin X$, then $h_j=0$, and this sum is positive. If $j\in X$, then it equals $y-(2q-1)$, which is zero exactly when $y=2q-1$, and block $j$ has such a prefix exactly when $j\in S_q$. So some prefix sum is zero exactly when $S_q\cap X\ne\emptyset$.

\subsubsection{Pagh's problem}\label{sec:pagh}

A family of subsets of a universe is maintained under updates that add the intersection of two sets of the family to the family, and a query asks whether a given element belongs to a given set of the family. The reduction is \Patrascu's~\cite[Appendix~A]{patrascu_mihai_towards_2010}, with the universe added to the initial family to handle $X=\emptyset$. The universe is $[m]$, and the initial family consists of $Y_0:=[m]$ and $Y_j:=\{i\in[m]:j\notin S_i\}$ for $j\in[n]$. We keep the index of a current set $C$, initially $Y_0$. Inserting $j$ into $X$ adds $C\cap Y_j$ to the family and makes it the current set. Then $C=\{i\in[m]:S_i\cap X=\emptyset\}$, so the query $q$ asks whether $q\in C$.

\subsubsection{Erickson's problem}\label{sec:erickson}

A matrix of integers is maintained under updates that increment all entries of a row or of a column, and a query asks for the largest entry. The reduction is \Patrascu's~\cite[Section~1.2, proof of Theorem~2]{patrascu_mihai_towards_2010}. Initially the matrix is $(M_{ij})_{i\in[m],j\in[n]}$. Inserting $j$ into $X$ increments column $j$, and the query $q$ increments row $q$ and asks for the largest entry. The entry in row $i$ and column $j$ is then $M_{ij}+x_j+\mathbf 1_{i=q}$, so the largest entry is $3$ exactly when $S_q\cap X\ne\emptyset$. A quadratic-logarithmic lower bound for a counting version of this problem, whose query asks for the parity of the number of largest entries, is proved in~\cite[Corollary~1.2]{ko_omega_2026}.

\subsubsection{Range mode}\label{sec:range-mode}

An array is maintained under updates that assign a value to an entry, and a query asks for a mode, that is, a most frequent value, of a given subarray. We adapt the reduction in \cite[Section~9.2]{chan_linear-space_2014}, using a returned mode instead of its frequency. The array consists of $m+1$ chunks of $n$ entries each. In Phase~I, chunk $i\in[m]$ lists the elements of $[n]\setminus S_i$ followed by those of $S_i$, and the last chunk consists of zeros; we also store the sizes $\abs{S_i}$ and the entries $M_{ij}$. When $j$ is inserted into $X$ as its $k$-th element, set $A[mn+k]:=j$; we also keep $\abs X$. If $S_q=\emptyset$, the query $q$ answers ``disjoint''. Otherwise it asks for a mode $z$ of $A[qn-\abs{S_q}+1:mn+\abs X]$, and answers ``disjoint'' exactly when $z\notin S_q\cap X$.

The queried subarray consists of the elements of $S_q$, then $m-q$ chunks that are each a permutation of $[n]$, and then the elements of $X$. So each $j\in[n]$ occurs $(m-q)+M_{qj}+x_j$ times in it, and no other value occurs. If $S_q\cap X\ne\emptyset$, the modes are exactly the elements of $S_q\cap X$, each with $m-q+2$ occurrences. So $z\in S_q\cap X$ exactly when $S_q\cap X\ne\emptyset$.

\paragraph{Costs.}
Each insertion into $X$ is at most two updates. The query is one query for set disjointness, Pagh's problem and range mode, and one update and one query for Langerman's and Erickson's problems. So \cref{lem:app-transfer} gives $t_q=\Omega(\mathcal L_N)$ for the former three problems, and $t_u+t_q=\Omega(\mathcal L_N)$ for the latter two.

\subsection{Matrix problems}\label{sec:app-matrices}

Here $N$ is the dimension of a matrix over a prime field, which is maintained under updates that change a single entry. The reductions specialize van den Brand, Nanongkai and Saranurak~\cite[Theorems~C.1 and~C.3]{van_den_brand_dynamic_2019}; our choice of the field prevents cancellation and keeps the matrices nonsingular. We work over $\mathbb F_p$ for a prime $p>n+1$ with $p=\Theta(n)$, so that field elements occupy $\Theta(\log n)$ bits. Let $d:=2n+1$, and for $x=(x_j)_{j\in[n]}$, let
\[
A(x):=\begin{pmatrix}
I_n&M&0\\
0&I_n&x\\
0&0&1
\end{pmatrix},
\qquad\text{so that}\qquad
A(x)^{-1}=\begin{pmatrix}
I_n&-M&Mx\\
0&I_n&-x\\
0&0&1
\end{pmatrix}.
\]
In Phase~I, the matrix is $A(0)$, and inserting $j$ into $X$ changes the entry $(n+j,d)$ from $0$ to $1$. So the matrix is always $A(x)$, which has entries in $\{0,1\}$ and is unit upper triangular.

\paragraph{Field size.}
The bounds below hold for every family of prime fields with $p=\Theta(N)$: take the largest $n$ with $2n+1\le N$ and $n+2\le p$, which is $\Theta(N)$, and pad the matrix to dimension $N$ by a direct sum with an identity matrix. This preserves the queried entries of the inverse, the determinant, and nonsingularity.

\subsubsection{Inverse-entry zero testing}\label{sec:app-inverse}

The matrix is always nonsingular, and a query asks whether a given entry of its inverse is zero. The query $q$ asks whether the entry $(q,d)$ of $A(x)^{-1}$ is zero. This entry is $c_q\in[0,n]$, which is zero in $\mathbb F_p$ exactly when $c_q=0$, since $p>n$. The costs are one update per insertion and one query, so \cref{lem:app-transfer} gives $t_q=\Omega(\mathcal L_N)$, also for queries that return the entry itself.

\subsubsection{Determinant}\label{sec:app-determinant}

A query asks for the determinant, or whether it equals one. As in~\cite[Equation~(11)]{van_den_brand_dynamic_2019}, the query $q$ changes the entry $(d,q)$ from $0$ to $1$, and asks for the determinant of the resulting matrix $A(x)+e_de_q^{\mathsf T}$. By the matrix determinant lemma $\det(A+uv^{\mathsf T})=\det(A)\,(1+v^{\mathsf T}A^{-1}u)$, which follows by factoring out $A$ and expanding $\det(I+A^{-1}uv^{\mathsf T})$,
\[
\det\big(A(x)+e_de_q^{\mathsf T}\big)=1+\big(A(x)^{-1}\big)_{q,d}=1+c_q.
\]
Since $1\le1+c_q\le n+1<p$, the matrix is nonsingular, and its determinant is one exactly when $c_q=0$. Each insertion into $X$ is one update, and the query is one update and one query, so \cref{lem:app-transfer} gives $t_u+t_q=\Omega(\mathcal L_N)$. Testing whether the determinant is zero does not tell the two cases apart.

\subsection{Geometric problems}\label{sec:geometric-counting}

Here $N$ bounds the number of points, which are inserted one at a time, and all coordinates are integers of $O(\log N)$ bits. Each construction has $O(n^2)$ points, so we take $n=\Theta(\sqrt N)$.

\subsubsection{Convex-hull vertex counting}\label{sec:hull}

A query asks for the number of vertices of the convex hull of the current points in $\mathbb R^3$. The construction is that of Dallant and Iacono~\cite[ESA version, Section~3, Theorem~10 and Lemma~11]{dallant_conditional_2022}. Their fully dynamic version deletes the blockers of a supplied set $J$; as in their incremental version, we start without blockers and insert the blockers of the elements of $X$, so that $J=[n]\setminus X$. The points lie on $n$ columns, one for each element. A blocker hides its column, and a selector hides all levels above a cutoff, so two consecutive selectors isolate the surviving points of row $q$. Let $n\ge5$, and let
\[
 H:=2n+1,\qquad R:=4n^4,\qquad
 e_j:=(\cos(2\pi j/n),\sin(2\pi j/n)),\qquad
 v_j(h):=((R-h^2)e_j,h).
\]
In Phase~I, insert the points
\[
 g_{i,j}:=v_j(2i+1)\quad(0\le i\le n,\ j\in[n]),
 \qquad
 p_{i,j}:=v_j(2i)\quad(i\in[n],\ j\in S_i).
\]
The points that may be inserted later are the blockers and the selectors
\[
 b_j:=((R-1)e_j,H+1)\quad(j\in[n]),\qquad
 t_a:=(0,0,T_{2a-1})\quad(a\in[n+1]),\qquad
 \text{where}\quad T_s:=\frac{R+s^2}{2s}.
\]

\begin{lemma}\label{lem:hull-selector-count}
Add any set of blockers $b_j$ and any set of selectors $t_a$ to the points inserted in Phase~I, and let $a$ be the least index of an added selector, or $n+2$ if there is none. Then the vertices of the convex hull are exactly
\begin{itemize}[leftmargin=*]
  \item the points $g_{0,j}$ and the added blockers;
  \item the selector $t_a$, if a selector was added;
  \item the points $g_{i,j}$ and the present points $p_{i,j}$ with $1\le i<a$ whose blocker $b_j$ was not added.
\end{itemize}
This classification is unchanged if every point is moved by Euclidean distance at most $1/R^2$. Consequently, all the points can be replaced by integer points with coordinates of absolute value $O(n^{12})$.
\end{lemma}

\begin{proof}
Write points as $(x,z)$ with $x\in\mathbb R^2$ and $z\in\mathbb R$. All radii $R-h^2$ with $1\le h\le H$ are positive, $T_s$ decreases for $s\in[1,H]$, and $T_H>H+1$.

\smallskip\noindent\emph{Exposed vertices.} The functional $e_j\cdot x-z$ uniquely exposes $g_{0,j}$, and $e_j\cdot x+z/2$ uniquely exposes an added blocker $b_j$. The highest added selector uniquely maximizes $z$. For a point $v_j(h)$ with $h>1$ below the cutoff whose blocker was not added, use
\[
 L_{j,h}(x,z):=e_j\cdot x+(2h-\tfrac12)z.
\]
The deficit of another point $v_\ell(h')$ inserted in Phase~I, relative to $v_j(h)$, is
\[
 (h'-h)^2+\frac{h'-h}{2}
 +(R-h'^2)(1-e_j\cdot e_\ell)>0;
\]
here $h$ and $h'$ are integers, and equality is possible only for the same point. A blocker in another column also has smaller value: since
\[
 1-\cos(2\pi/n)\ge\frac8{n^2}\qquad\text{and}\qquad R-H^2\ge R/2,
\]
its deficit is at least $(R-1)(1-\cos(2\pi/n))-(H+1)^2+1>1$. Every added selector has a parameter $s\ge h$, and its deficit is
\[
 (s-h)(R/s-h)+\tfrac12(T_s-h)>0.
\]
These calculations also show that every other point is at Euclidean distance at least $1/(10n)$ from the exposing plane of each claimed vertex. For the highest selector, consecutive possible heights differ by more than one, since
\[
 T_s-T_{s'}=\frac{(s'-s)(R-ss')}{2ss'}>1
 \qquad(1\le s<s'\le H,\ s,s'\text{ odd}).
\]

\smallskip\noindent\emph{Points hidden by a selector.} We show next that every remaining point is strictly interior. If a selector with parameter $s=2a-1$ is present, the cone with apex $t_a$ and base $\operatorname{conv}\{g_{a-1,j}:j\in[n]\}$ contains every point $v_j(h)$ inserted in Phase~I with $h>s$ in its interior. Indeed, its cross-section at height $h$ is the regular polygon with vertices
\[
 \bigl((R+s^2-2sh)e_j,h\bigr)
 =\bigl((R-h^2+(h-s)^2)e_j,h\bigr),
\]
and $v_j(h)$ lies strictly between its center and the corresponding vertex. More quantitatively, if $u$ is the outward unit normal of an edge of the base polygon and $c:=\cos(\pi/n)$, the corresponding side face has the equation
\[
 u\cdot x+2scz=c(R+s^2).
\]
At $v_j(h)$, the slack is at least $c(h-s)^2$, so the distance of $v_j(h)$ from every side face is at least $1/(10n)$, and its distance from the base is at least one. A lower selector is likewise strictly interior to the cone with apex $t_a$ and base $\operatorname{conv}\{g_{0,j}:j\in[n]\}$: its distance from the base is greater than one, and its distance from each side face is greater than $1/2$.

\smallskip\noindent\emph{Points hidden by a blocker.} Suppose that $b_j$ is present. Rotate the coordinates so that $e_j=(1,0)$, and let
\[
 c':=\cos(2\pi/n),\qquad d:=(R-H^2)\sin(2\pi/n),\qquad
 D:=R-1-c'(R-H^2).
\]
The tetrahedron with vertices $g_{0,j}$, $b_j$, $g_{n,j-1}$ and $g_{n,j+1}$, with column indices modulo $n$, contains every $v_j(h)$ with $2\le h\le H$ in its interior. To see this, let $\lambda:=(h^2-1)/D$. The barycentric weights of $v_j(h)$ are $\lambda/2$ for each of $g_{n,j-1}$ and $g_{n,j+1}$, and
\[
 \frac{h-1-\lambda(H-1)}H
 \qquad\text{and}\qquad
 \frac{H+1-h-\lambda}H
\]
for $b_j$ and $g_{0,j}$, respectively. All of them are positive, since $D>H^2-1$. Quantitatively, $D\le2R$ and
\[
 1-\frac{H^2-1}{D}
 =\frac{(1-c')(R-H^2)}D\ge\frac2{n^2},
\]
so every weight is at least $3/(4R)$. The absolute value of the determinant of the tetrahedron is $2HDd\ge32R^2/n^2$, while the cross product of the edge vectors of every face has norm at most $12R^2$. So every altitude is at least $8/(3n^2)$, and $v_j(h)$ is at distance at least $2/(Rn^2)>2/R^2$ from every face.

\smallskip\noindent\emph{Rounding.} All the separation and interior margins above exceed $2/R^2$. Moving a point and all vertices of a witnessing hull by at most $1/R^2$ preserves strict separation and strict interior containment, since the support function of the hull changes by at most $1/R^2$ in every unit direction. This proves the perturbation statement. Round every coordinate to the nearest multiple of $1/(2R^2)$, which moves every point by at most $\sqrt3/(4R^2)<1/R^2$, and multiply all coordinates by $2R^2$. The resulting coordinates are integers of absolute value at most $2R^3+1=O(n^{12})$, and scaling preserves the vertex counts. All these coordinates are fixed in Phase~I.
\end{proof}

Inserting $j$ into $X$ inserts the blocker $b_j$. The query $q$ inserts $t_{q+1}$ and asks for the number $\kappa$ of hull vertices, and then inserts $t_q$ and asks for the number $\kappa'$. By \cref{lem:hull-selector-count},
\[
 \kappa=n+\abs X+1+\sum_{i=1}^{q}\bigl(\abs J+\abs{S_i\cap J}\bigr),\qquad
 \kappa'=n+\abs X+1+\sum_{i=1}^{q-1}\bigl(\abs J+\abs{S_i\cap J}\bigr),
\]
so $\abs{S_q\cap X}=\abs{S_q}+\abs J-(\kappa-\kappa')$, which we compute from the stored $\abs{S_q}$ and $\abs X$. There are at most $2n^2+2n+2$ points. Each insertion into $X$ is one update, and the query is two updates and two queries, which do not depend on the answer to the first query; so \cref{lem:app-transfer} gives $t_u+t_q=\Omega(\mathcal L_N)$.

\subsubsection{Maximal-point counting}\label{sec:maxima}

A point is \emph{maximal} if no other point is at least as large in every coordinate, and a query asks for the number of maximal points among the current points in $\mathbb R^3$. Up to the order of the coordinates and of the columns, the construction is the incremental version of the reduction of Dallant and Iacono~\cite[Theorem~14]{dallant_conditional_2022}; Jin and Xu~\cite[Section~5.1, proof of Theorem~1.10]{jin_tight_2022} use the same structure, with $k=2$ and infinitesimal coordinates, for a stronger conditional bound. In Phase~I, insert the points
\[
 p_{i,j}:=(4nj-i,\ n-j,\ i)\qquad(i\in[n],\ j\in S_i).
\]
Inserting $j$ into $X$ inserts the blocker $b_j:=(4nj,\ n-j,\ n+1)$, and for $0\le r\le n$, let $t_r:=(4n(n+1),\ n+1,\ r)$. The query $q$ inserts $t_{q-1}$ and asks for the number $\kappa$ of maximal points, and then inserts $t_q$ and asks for the number $\kappa'$.

The points $p_{i,j}$ are pairwise incomparable: if $p_{i',j'}$ dominates $p_{i,j}$, then the second coordinate forces $j'\le j$, and the first forces $j'\ge j$, since $\abs{i-i'}<4n$; so $j'=j$, and then the first and third coordinates force $i'=i$. The blocker $b_j$ dominates exactly the points $p_{i,j}$ in column $j$: the second coordinate forces their column to be at least $j$, and the first forces it to be at most $j$. The blockers are pairwise incomparable; their third coordinate prevents the selectors and the points $p_{i,j}$ from dominating them, and their first coordinate prevents them from dominating the selectors. Finally, $t_r$ dominates exactly the points $p_{i,j}$ with $i\le r$, and $t_q$ dominates $t_{q-1}$. Hence
\[
 \kappa=\abs X+1+\sum_{i=q}^{n}\abs{S_i\setminus X},\qquad
 \kappa'=\abs X+1+\sum_{i=q+1}^{n}\abs{S_i\setminus X},
\]
so $\kappa-\kappa'=\abs{S_q\setminus X}$, and $S_q\cap X=\emptyset$ exactly when $\kappa-\kappa'=\abs{S_q}$. There are at most $n^2+n+2$ points, with coordinates of absolute value $O(n^2)$. Each insertion into $X$ is one update, and the query is two updates and two queries, which do not depend on the answer to the first query; so \cref{lem:app-transfer} gives $t_u+t_q=\Omega(\mathcal L_N)$.

\subsubsection{Distinct-color counting}\label{sec:distinct-color-counting}

Colored points in the plane are inserted, and a query asks for the number of distinct colors among the points in a given axis-parallel rectangle. The reduction specializes the construction of Jin and Xu~\cite[Section~3.4, proof of Theorem~1.8]{jin_tight_2022} to one family of sets and one inserted set; we use dominance queries and recover the intersection by inclusion--exclusion. Let $W:=4n+4$. In Phase~I, insert the point $(Wi+j,\,W(n+1-i))$ with color $j$ for each $i\in[n]$ and $j\in S_i$. Inserting $j$ into $X$ inserts the point $(j,0)$ with color $j$; we keep $\abs X$. The query $q$ asks for the number $C_q$ of colors in $[0,Wq+n]\times[0,W(n+1-q)]$.

The bound on the second coordinate excludes the points of rows $i<q$, and the bound on the first coordinate excludes those of rows $i>q$, since $W(q+1)+1>Wq+n$. So the rectangle contains exactly the points of row $q$ and the inserted points, whose colors form $S_q\cup X$, and $\abs{S_q\cap X}=\abs{S_q}+\abs X-C_q$. All coordinates are nonnegative, so the rectangle is a dominance range. There are at most $n^2+n$ points, with coordinates $O(n^2)$; to keep $\Theta(n^2)$ points at all times, add points of a new color to the right of all queried rectangles. The costs are one update per insertion and one query, so \cref{lem:app-transfer} gives $t_q=\Omega(\mathcal L_N)$.

\subsection{String problems}\label{sec:strings}

Here $N$ bounds the total length of the texts and patterns. All constructions have size $\Theta(n^2)$, so we take $n=\Theta(\sqrt N)$.

\subsubsection{Internal dictionary matching}\label{sec:internal-dictionary}

A text $T$ is fixed. Patterns, each given by the endpoints of an occurrence in $T$, are inserted into a dictionary that is initially empty, and a query $\mathsf{Exists}(\ell,r)$ asks whether some pattern of the dictionary occurs in $T[\ell:r]$. The reduction is from \cite[Theorem~24]{charalampopoulos_internal_2019}, with a suffix added to $T$ so that every pattern occurs in $T$. Over the alphabet $\{0,1,\dots,n\}$, the text consists of the blocks $T[(i-1)n+1:in]$ for $i\in[n]$, where $T[(i-1)n+j]:=j$ if $j\in S_i$ and $T[(i-1)n+j]:=0$ otherwise, followed by the suffix $1\,2\cdots n$; so $\abs T=n^2+n$. Inserting $j$ into $X$ inserts the one-symbol pattern $j$, given by its occurrence at position $n^2+j$. The query $q$ asks $\mathsf{Exists}((q-1)n+1,qn)$. The block of $q$ contains the symbol $j\in[n]$ exactly when $j\in S_q$, the dictionary contains it exactly when $j\in X$, and the symbol $0$ is never in the dictionary; so the answer is yes exactly when $S_q\cap X\ne\emptyset$. The costs are one update per insertion and one query, so \cref{lem:app-transfer} gives $t_q=\Omega(\mathcal L_N)$.

\subsubsection{Hamming-distance threshold}\label{sec:hamming-threshold}

A binary text $T$ is fixed, and a binary pattern $P$ is maintained under updates that change one of its bits. A query asks whether the Hamming distance between $P$ and the substring of $T$ of length $\abs P$ that starts at a given position is at most $k$, for a threshold $k$ fixed in advance. The reduction combines the construction of Clifford, Gr{\o}nlund, Larsen and Starikovskaya~\cite[Theorem~15]{clifford_upper_2018}, which concatenates the rows of the matrix, with their binary encoding~\cite[Lemma~11]{clifford_upper_2018}. Let $L:=n^2$. Before the encoding, let $T_0$ be the concatenation of the rows $M_{i1}\cdots M_{in}$ for $i\in[n]$, followed by $0^L$, and let $P_0:=x_1\cdots x_n0^{L-n}$; so $\abs{T_0}=2L$ and $\abs{P_0}=L$. Encode the symbols of the pattern and of the text by
\[
f(0):=010,\qquad f(1):=111,
\qquad\text{and}\qquad
g(0):=100,\qquad g(1):=111,
\]
respectively, and let $P:=f(P_0)$ and $T:=g(T_0)$. For all $a,b\in\{0,1\}$, the Hamming distance between $f(a)$ and $g(b)$ is $2-2ab$. Initially $x=0^n$, and inserting $j$ into $X$ changes the codeword of $x_j$ from $010$ to $111$, by two updates. The query $q$ compares $P$ with the substring of $T$ that starts at position $3(q-1)n+1$, which fits in $T$ since $3(q-1)n+3L\le6L=\abs T$. There, the first $n$ codewords of $P$ are compared with row $q$, and each of the other $L-n$ codewords of $P$ encodes $0$ and contributes exactly two mismatches, whatever the aligned text symbol. So the Hamming distance is
\[
2L-2\sum_{j=1}^nM_{qj}x_j=2L-2c_q,
\]
and the threshold $k:=2L-1=2\abs P/3-1$ detects whether $c_q>0$. Each insertion into $X$ is two updates and the query is one query, so \cref{lem:app-transfer} gives $t_q=\Omega(\mathcal L_N)$, also for queries that return the Hamming distance. The threshold grows with the length of the pattern; we do not claim a bound for a constant threshold or for constant-factor approximation.

\subsubsection{Active-document existence}\label{sec:active-documents}

A collection of documents is fixed, and every document is active or inactive. An update activates a document, and a query asks whether some active document contains two given patterns. The reduction specializes the document construction of Jin and Xu~\cite[Section~3.3, proof of Theorem~1.7]{jin_tight_2022}: one of the two patterns is a marker that every document contains, and all updates are activations. Over the alphabet $[n]\cup\{\#,\$\}$, the document of $j\in[n]$ is
\[
D_j:=\#\,d_{1j}\cdots d_{nj},\qquad\text{where}\qquad
d_{ij}:=\begin{cases}
i,&j\in S_i,\\
\$,&j\notin S_i.
\end{cases}
\]
All documents are initially inactive, and their total length is $n(n+1)$. Inserting $j$ into $X$ activates $D_j$, and the query $q$ asks whether some active document contains both $q$ and $\#$. Every document contains $\#$, and $D_j$ contains $q$ exactly when $j\in S_q$; so the answer is yes exactly when $S_q\cap X\ne\emptyset$. The same holds for the query with the single pattern $q$. The costs are one update per insertion and one query, so \cref{lem:app-transfer} gives $t_q=\Omega(\mathcal L_N)$.

\subsection{Streaming regular-expression membership: amortized versus worst case}\label{sec:streaming-regex}

In streaming regular-expression membership, an expression $R$ is preprocessed, and a text then arrives one symbol at a time; after each symbol, the algorithm reports whether the text read so far belongs to $L(R)$, without knowing which symbol is the last one. Bille and Thorup~\cite{bille_regular_2010} and Dudek, Gawrychowski, Gourdel and Starikovskaya~\cite{naor_streaming_2022} give algorithms for this problem. Here we give a direct encoding of Disjointness with constant amortized cost. We count the cell probes made while processing each symbol, including its report; preprocessing is not charged, the space is polynomial, and the conventions of \cref{sec:model} apply.

\paragraph{The expression family.}
Fix sets $S_1,\dots,S_m\subseteq[n]$. The alphabet consists of the symbol $\#$, the data symbols $a_1,\dots,a_n$ and the selectors $c_1,\dots,c_m$, and we write $\cup$ for the union (alternation) operator of regular expressions. For each $i\in[m]$ with $S_i\ne[n]$, let
\[
\Gamma_i := \bigcup_{j\in[n]\setminus S_i} a_j
\]
be the alternation of the data symbols of the elements outside $S_i$, and let the $i$-th branch be $\#\,\Gamma_i^*\,c_i$; if $S_i=[n]$, let the $i$-th branch be $\#\,c_i$. Let $R_{\vS}$ be the union of the $m$ branches, an ordinary regular expression that uses only union, concatenation and Kleene star. For every $q\in[m]$ and every sequence $j_1,\dots,j_k\in[n]$, in any order and possibly with repetitions,
\begin{equation}\label{eq:streaming-regex-reduction}
\#\,a_{j_1}\cdots a_{j_k}\,c_q\in L(R_{\vS})
\qquad\Longleftrightarrow\qquad
S_q\cap\{j_1,\dots,j_k\}=\emptyset.
\end{equation}
Indeed, the final selector $c_q$ forces the $q$-th branch, which accepts exactly when every preceding data symbol is $a_j$ for some $j\notin S_q$.

\paragraph{Constant amortized time.}
In preprocessing, store the incidence table of the sets $S_i$, a bit $\mathsf{seen}[j]=0$ for each $j\in[n]$, an empty list, and the parser state. After the initial $\#$, when $a_j$ arrives, append $j$ to the list if $\mathsf{seen}[j]=0$, and set $\mathsf{seen}[j]:=1$. When the first selector $c_q$ arrives, scan the list, report ``yes'' exactly when none of its indices belongs to $S_q$, and enter a terminal state. All earlier reports are ``no''. An initial symbol other than $\#$, a second $\#$, or any symbol after a selector causes permanent rejection, at constant cost per later symbol. The scan examines at most as many indices as there are preceding data symbols, so every prefix of length $L$ costs $O(L)$ in total, also for malformed texts. The expression has length $m\le M=O(mn)$. For $m=n$, we have $\log M=\Theta(\log n)$, the alphabet has $2n+1$ symbols, and the $O(n^2)$-bit incidence table and the $O(n)$ auxiliary words use polynomial space.

\begin{theorem}[Streaming membership: amortized versus worst case]\label{thm:streaming-regex}
Let $\mathcal R_n:=\{R_{\vS}:S_1,\ldots,S_n\subseteq[n]\}$, and let $M_n:=\max_{R\in\mathcal R_n}|R|=\Theta(n^2)$. Streaming membership for this family has a deterministic algorithm with $O(1)$ amortized time per symbol on every input stream.

Consider a zero-error cell-probe algorithm for streaming membership for all of $\mathcal R_n$, with word size $w=\Theta(\log M_n)$ and arbitrary preprocessing of the given expression, which stores all information retained between symbols in its memory. Let $t_a$ be its worst-case cost on a data symbol, over all expressions, texts and random tapes, and let $t_c$ be an upper bound on its expected cost on a selector, for every expression and every text that ends with this selector. Then
\[
t_c=\Omega\!\left(\left(\frac{\log M_n}{\log(t_a w)}\right)^2\right).
\]
In particular, the worst-case time on some symbol, for some expression in the family and some text, is $\Omega\big((\log M_n/\log\log M_n)^2\big)$.
\end{theorem}
\begin{proof}
The algorithm above gives the upper bound. For the lower bound, given an instance $S_1,\dots,S_n$ of insert-only multiphase Disjointness, preprocess $R_{\vS}$ and process the symbol $\#$ during Phase~I. Inserting $j$ into $X$ feeds the symbol $a_j$, and the query $q$ feeds the selector $c_q$; by~\eqref{eq:streaming-regex-reduction}, the report on $c_q$ answers the query. Keeping the writes made while processing $c_q$ in temporary workspace, as in \cref{lem:app-transfer}, we obtain a zero-error data structure with worst-case update time $t_a$ and expected query time $t_c$. Since $\log M_n=\Theta(\log n)$, \cref{thm:main} gives the bound on $t_c$. For the last statement, let $t$ be the worst-case time on all symbols, so that $t_a,t_c\le t$. If $t\ge\log^2 M_n$ there is nothing to prove; otherwise $\log(t_a w)=O(\log\log M_n)$, so $t\ge t_c=\Omega\big((\log M_n/\log\log M_n)^2\big)$.
\end{proof}

\paragraph{Binary alphabet.}
Encode the symbols by distinct codewords of length $b=\Theta(\log n)$. The maximum expression length becomes $M^{\mathrm{bin}}_n=O(n^2\log n)$, with $\log M^{\mathrm{bin}}_n=\Theta(\log n)$; take $w=\Theta(\log M^{\mathrm{bin}}_n)$. Decoding preserves constant amortized time per bit: report ``no'' within incomplete codewords and reject invalid codewords permanently. A matcher with worst-case cost $t$ per bit yields source update and query costs at most $bt$. If $t>\log M^{\mathrm{bin}}_n$, the bound below is immediate. Otherwise $\log(btw)=O(\log\log M^{\mathrm{bin}}_n)$, so \cref{thm:main} gives $bt=\Omega\big((\log M^{\mathrm{bin}}_n/\log\log M^{\mathrm{bin}}_n)^2\big)$. Since $b=\Theta(\log M^{\mathrm{bin}}_n)$,
\[
t=\Omega\!\left(\frac{\log M^{\mathrm{bin}}_n}{(\log\log M^{\mathrm{bin}}_n)^2}\right).
\]

\paragraph{Related separations.}
Union--find admits amortized $\Theta(\alpha(n))$ time~\cite{tarjan_efficiency_1975,fredman_cell_1989} but requires $\Theta(\log n/\log\log n)$ worst-case time~\cite{blum_single-operation_1985,fredman_cell_1989}. Cohen, Duffield, Kaplan, Lund and Thorup~\cite{cohen_efficient_2011} obtain $O(\log\log k)$ amortized time and an $\Omega(\log k/\log\log k)$ worst-case cell-probe bound for variance-optimal reservoir sampling in their finite-precision model. Streaming $k$-mismatch has conditional bounds from combinatorial Boolean matrix multiplication~\cite{golan_streaming_2020,gawrychowski_towards_2018}.

\end{document}